\documentclass[preprint,10pt,times,3p,twocolumn]{elsarticle}

\usepackage{times,url,graphicx,latexsym,xspace,%
eurosym,amssymb,amsmath,helvet,courier}

\usepackage[ruled]{algorithm}
\usepackage[noend]{algorithmic}

\renewcommand{\geq}{\geqslant}
\renewcommand{\leq}{\leqslant}

\newcommand{\AU}{\ensuremath{\mathbf{U}}\xspace}
\newcommand{\AP}{\ensuremath{\mathbf{P}}\xspace}
\newcommand{\AB}{\ensuremath{\mathbf{B}}\xspace}
\newcommand{\AL}{\ensuremath{\mathbf{L}}\xspace}
\newcommand{\AV}{\ensuremath{\mathbf{V}}\xspace}
\newcommand{\AUBL}{\ensuremath{\mathbf{UBL}}\xspace}

\newcommand{\AUL}{\ensuremath{\mathbf{UL}}\xspace}
\renewcommand{\AA}{\ensuremath{\mathbf{A}}\xspace}
\newcommand{\AULV}{\ensuremath{\mathbf{ULV}}\xspace}
\newcommand{\AUV}{\ensuremath{\mathbf{UV}}\xspace}

\newcommand{\A}{\ensuremath{\mathcal{A}}\xspace}

\newcommand{\I}{\ensuremath{\mathcal{I}}\xspace}

\newcommand{\nd}{\noindent}

\newcommand{\tuple}[1]{\langle #1 \rangle }
\newcommand{\triple}[1]{(#1)}

\newcommand{\term}[1]{\ensuremath{\mbox{\small$\mathtt{#1}$}}}
\newcommand{\uri}[1]{\ensuremath{\term{#1}}}

\newcommand{\eval}[1]{[\![ #1 ]\!]}

\newcommand{\assign}{\mbox{$\colon\!\!\!\!=$}}

\newcommand{\intA}[1]{{#1}^{\I_{\A}} }
\renewcommand{\int}[1]{{#1}^{\I} }
\newcommand{\intP}[1]{{P[\![#1]\!]}}
\newcommand{\intC}[1]{{C[\![#1]\!]}}

\newcommand{\spp}{\ensuremath{\mathsf{sp}}}
\newcommand{\subclass}{\ensuremath{\mathsf{sc}}}
\newcommand{\type}{\ensuremath{\mathsf{type}}}
\newcommand{\dom}{\ensuremath{\mathsf{dom}}}
\newcommand{\range}{\ensuremath{\mathsf{range}}}
\newcommand{\eee}{\ensuremath{\mathsf{e}}}

\newcommand{\ii}[1]{\mbox{$(#1)$}}

\newcommand{\fuzzyg}[2]{\mbox{$#1\colon#2$}}

\newcommand{\aggr}{@}

\renewcommand{\vec}[1]{\bar{\mathbf #1}}

 \newtheorem{theorem}{Theorem}[section]
 \newtheorem{proposition}[theorem]{Proposition}
 
 \newtheorem{lemma}[theorem]{Lemma}

 \newtheorem{example}{Example}[section]
 \newtheorem{definition}{Definition}[section]
 \newtheorem{proof}{Proof}[section]
\newtheorem{remark}{Remark}[section]

\newcommand{\ie}{\textit{i.e.},\xspace}
\newcommand{\eg}{\textit{e.g.},\xspace}
\newcommand{\wrt}{w.r.t.\xspace}

\newcommand{\etal}{\textit{et al.}\xspace}
\renewcommand{\iff}{if and only if\xspace}
\newcommand{\st}{such that\xspace}

\newcommand{\rhodf}{\mbox{$\rho$df}\xspace}
\newcommand{\universe}{\ensuremath{\mathit{universe}}\xspace}

\renewcommand{\implies}{\Rightarrow}

\newcommand{\stt}[1]{\texttt{{\small#1}}}

\newcommand{\louter}{{\sqsupset\!\!\!}\bowtie}

\newcommand{\aSPARQL}{AnQL\xspace}

\newcommand{\bigoplusone}[1]{{{\bigoplus_1}\atop{\text{\tiny$#1$}}}}
\newcommand{\bigoplustwo}[1]{{{\bigoplus_2}\atop{\text{\tiny$#1$}}}}
\newcommand{\bigotimesone}[1]{{{\bigotimes_1}\atop{\text{\tiny$#1$}}}}
\newcommand{\bigotimestwo}[1]{{{\bigotimes_2}\atop{\text{\tiny$#1$}}}}

\newcommand{\intLabel}[1]{\pgfmathparse{2000+#1} \pgfmathtruncatemacro{\res}{\pgfmathresult} \res}

\newcommand{\makegraph}[2]{
  \begin{tikzpicture}[scale=0.65]
    \tikzstyle{every node}=[gray,font=\scriptsize]
    \draw (0,0) node[left]{$T_1$} -- (10,0) ;
    \foreach \startPoint/\endPoint in #1
    {
      \draw[line width=2mm] (\startPoint,0) node[above]{\intLabel{\startPoint}} -- (\endPoint,0) node[above]{\intLabel{\endPoint}};
    }

    \draw (0,-1) node[left]{$T_2$} -- (10,-1) ;
    \foreach \startPoint/\endPoint in #2
    {
      \draw[line width=2mm] (\startPoint,-1) node[below]{\intLabel{\startPoint}} -- (\endPoint,-1) node[below]{\intLabel{\endPoint}};
    }

  \end{tikzpicture}
}

\newcommand{\pow}[1]{2^{#1}} %

\let\Oldfootnote\footnote
\renewcommand{\footnote}[1]{\Oldfootnote{#1}}

\newtheorem{property}{Property}

\usepackage[pdfborder={0 0 0}]{hyperref}
\usepackage[english]{babel}

\usepackage{times}
\usepackage{alltt}
\usepackage{algorithmic}
\usepackage{epsfig}
\usepackage{url}
\usepackage{xspace}
\usepackage{latexsym}
\usepackage{rotating} 
\usepackage{subfigure}
\usepackage{fancyvrb}
\usepackage{paralist}
\usepackage{multirow}
\usepackage{amsmath}
\usepackage{tikz}
\usetikzlibrary{calc}
\RecustomVerbatimEnvironment{Verbatim}{Verbatim}{fontsize=\small}

\journal{Journal of Web Semantics}

\begin{document}

\begin{frontmatter}

\title{A General Framework for Representing, Reasoning and Querying with Annotated Semantic Web Data}

\author[lyon]{Antoine Zimmermann}
\ead{E-mail: antoine.zimmermann@insa-lyon.fr}

\author[deri]{Nuno Lopes}
\ead{nuno.lopes@deri.org}

\author[deri]{Axel Polleres}
\ead{axel.polleres@deri.org}

\author[isti]{Umberto Straccia}
\ead{straccia@isti.cnr.it}

\address[deri]{Digital Enterprise Research Institute, National University of Ireland Galway, Ireland}
\address[isti]{Istituto di Scienza e Tecnologie dell'Informazione (ISTI - CNR), Pisa, Italy}
\address[lyon]{INSA-Lyon, LIRIS, UMR5205, F-69621, France}

\begin{abstract}
  We describe a generic framework for representing and reasoning with annotated Semantic Web data, a task becoming more
  important with the recent increased amount of inconsistent and non-reliable meta-data on the web. We formalise the
  annotated language, the corresponding deductive system and address the query answering problem. 
  Previous contributions on specific RDF annotation domains are encompassed by our unified reasoning formalism as we show by
  instantiating it on (i) temporal, (ii) fuzzy, and (iii) provenance annotations. Moreover, we provide a generic method for
  combining multiple annotation domains allowing to represent, \eg temporally-annotated fuzzy RDF.
  Furthermore, we address the development of a query language -- AnQL -- that is inspired by SPARQL, including several
  features of SPARQL~1.1 (subqueries, aggregates, assignment, solution modifiers) along with the formal definitions of their
  semantics.
\end{abstract}


\end{frontmatter}

\section{Introduction}

RDF (Resource Description Framework)~\cite{rdf} is the widely used representation language for the Semantic Web and the
Web of Data.  RDF exposes data as triples, consisting of \emph{subject}, \emph{predicate} and \emph{object}, stating
that \emph{subject} is related to \emph{object} by the \emph{predicate} relation.
Several extensions of RDF were proposed in order to deal with
time~\cite{gutierrez-etal-2007,DBLP:conf/www/PuglieseUS08,tapp-bern-09}, truth or imprecise
information~\cite{Mazzieri08,Straccia09f}, trust~\cite{hartig-09,Schenk08} and provenance~\cite{Dividino09}.  All these
proposals share a common approach of extending the RDF language by attaching meta-information about the RDF graph or
triples.
RDF Schema (RDFS)~\cite{rdfs} is the specification of a restricted vocabulary that allows one to deduce further information
from existing RDF triples.
SPARQL~\cite{sparql} is the W3C-standardised query language for RDF.

In this paper, we present an extension of the RDF model to support meta-information in the form of annotations of triples.
We specify the semantics by conservatively extending the RDFS semantics and provide a deductive system for Annotated RDFS.
Further, we define a query language that extends SPARQL and include advanced features such as aggregates, nested queries and
variable assignments, which are part of the not-yet-standardised SPARQL 1.1 specification. The present paper is based on and
extends two previously published articles introducing Annotated RDFS~\cite{Straccia10a} and AnQL (our SPARQL extension)
\cite{Lopes10a}. In addition to improving the descriptions of this existing body of work, we provide the following novelties: 
\begin{enumerate}
\item we introduce a use case scenario that better reflects a realistic example of how annotations can be used;
\item we detail three concrete domains of annotations (temporal, fuzzy, provenance) that were only sketched in our previous publications;
\item we present a detailed and systematic approach for combining multiple annotation domains into a new single complex domain; this represents the most significant novel contribution of the paper;
\item we discuss the integration of annotated triples with standard, non-annotated triples, as well as the integration of data using different annotation domains;
\item we describe a prototype implementation.
\end{enumerate}
Section~\ref{sec:preliminaries} gives preliminary definitions of the RDFS semantics and query answering, restricting ourselves to the sublanguage \rhodf. Our extension of RDF is presented in Section~\ref{sec:annotated-rdfs} together with essential examples of primitive domains.  Our extension of SPARQL, is presented in Section~\ref{sec:annotated-sparql}.  Furthermore, Section~\ref{sec:twists-asparql} presents a discussion of important issues with respect to specific domains and their combination.
Finally, Section~\ref{sec:implementation-notes} describes our prototype implementation.

\subsubsection*{Related Work}
\label{sec:related}

The basis for Annotated RDF were first established by Udrea \etal~\cite{Udrea06a,Udrea10}, where they define triples
annotated with values taken from a \emph{finite partial order}. In their work, triples are of the form
$\triple{s,p\!:\!\lambda,o}$, where the property, rather than the triple is annotated. We instead rely on a richer, not
necessarily finite, structure and provide additional inference capabilities to~\cite{Udrea10}, such as a more involved
propagation of annotation values through schema triples.  
For instance, in the temporal domain, from $\fuzzyg{\triple{a, \subclass, b}}{[2,6]}$ and $\fuzzyg{\triple{b, \subclass,
    c}}{[3,8]}$, we will infer $\fuzzyg{\triple{a, \subclass, c}}{[3,6]}$ ($\subclass$ is the subclass
property). Essentially, Udrea \etal do not provide an operation to combine the annotation in such inferences, while the
algebraic structures we consider support such operations. Also, they require specific algorithms, while we show that a
simple extension to the classical RDF inference rules is sufficient.
The query language presented in this paper consists of conjunctive queries and, while SPARQL's Basic Graph Patterns are
compared to their conjunctive queries, they do not consider extending SPARQL with the possibility of querying
annotations. Furthermore, $\mathsf{OPTIONAL}$, $\mathsf{UNION}$ and $\mathsf{FILTER}$ SPARQL queries are not considered
which results in a subset of SPARQL that can be directly translated into their previously presented conjunctive query
system.

Adding annotations to logical statements was already proposed in the logic programming realm in which Kifer \& Subrahmanian~\cite{Kifer92} present a similar approach, where atomic formulas are annotated with a value taken from a lattice of annotation values, an annotation variable or a
complex annotation, \ie~a function applied to annotation values or variables.
Similarly, we can relate our work to annotated relational databases, especially Green
\etal~\cite{Green_Karvounarakis_Tannen:07} who provides a similar framework for the relational algebra. After presenting
a generic structure for annotations, they focus more specifically on the provenance domain. The specificities of the
relation algebra, especially Closed World Assumption, allows them to define a slightly more general structure for
annotation domains, namely semiring (as opposed to the residuated lattice in our initial
approach~\cite{Straccia10a,Lopes10a}).  In relation to our rule-based RDFS Reasoning, it should be mentioned that Green
\etal~\cite{Green_Karvounarakis_Tannen:07} also provide an algorithm that can decide ground query answers for annotated
Datalog, which might be used for RDFS rules; general query answering or materialisation though might not terminate, due
to the general structure of annotations, in their case.  Karvounarakis \etal~\cite{DBLP:conf/sigmod/KarvounarakisIT10}
extend the work of~\cite{Green_Karvounarakis_Tannen:07} towards various annotations -- not only provenance, but also
confidence, rank, etc. -- but do not specifically discuss their
combinations.%

For the Semantic Web, several extensions of RDF were proposed in order to deal with specific domains such as truth of
imprecise information~\cite{Mazzieri08,MazzieriDragoni:2005aa,Mazzieri:2004aa,Straccia09f},
time~\cite{gutierrez-etal-2007,DBLP:conf/www/PuglieseUS08,tapp-bern-09}, trust~\cite{hartig-09,Schenk08} and
provenance~\cite{Dividino09}.  These approaches are detailed in the following paragraphs.

Straccia~\cite{Straccia09f}, presents Fuzzy RDF in a general setting where triples are annotated with a degree of truth
in $[0,1]$. For instance, ``Rome is a big city to degree 0.8'' can be represented with $\fuzzyg{\triple{\term{Rome},
    \type, \term{BigCity}}}{0.8}$; the annotation domain is $[0,1]$.  For the query language, it formalises conjunctive
queries.
Other similar approaches for Fuzzy RDF~\cite{Mazzieri08,MazzieriDragoni:2005aa,Mazzieri:2004aa} provide the syntax and
semantics, along with RDF and RDFS interpretations of the annotated triples.  In~\cite{Mazzieri:2004aa} the author
describes an implementation strategy that relies on translating the Fuzzy triples into plain RDF triples by using
reification.  However these works focus mostly on the representation format and the query answering problem is not
addressed.

Guti\'errez \etal~\cite{gutierrez-etal-2007} presents the definitions of Temporal RDF, including reduction of the
semantics of Temporal RDF graphs to RDF graphs, a sound and complete inference system and shows that entailment of
Temporal graphs does not yield extra complexity than RDF entailment.  Our Annotated RDFS framework encompasses this work
by defining the temporal domain.  They present conjunctive queries with built-in predicates as the query language for
Temporal RDF, although they do not consider full SPARQL.
Pugliese \etal~\cite{DBLP:conf/www/PuglieseUS08} presents an optimised indexing schema for Temporal RDF, the notion of
normalised Temporal RDF graph and a query language for these graphs based on SPARQL.  The indexing scheme consists of
clustering the RDF data based on their temporal distance, for which several metrics are given.  For the query language
they only define conjunctive queries, thus ignoring some of the more advanced features of SPARQL.
Tappolet and Bernstein~\cite{tapp-bern-09} present another approach to the implementation of Temporal RDF, where each
temporal interval is represented as a named graph~\cite{DBLP:journals/ws/CarrollBHS05} containing all triples valid in
that time period. Information about temporal intervals, such as their relative relations, start and end points, is
asserted in the default graph.  The $\tau$-SPARQL query language allows to query the temporal RDF representation using
an extended SPARQL syntax that can match the graph pattern against the snapshot of a temporal graph at any given time
point and allows to query the start and endpoints of a temporal interval, whose values can then be used in other parts
of the query.

SPARQL extensions towards querying trust have been presented by Hartig~\cite{hartig-09}.  Hartig introduces a trust
aware query language, tSPARQL, that includes a new constructor to access the trust value of a graph pattern.  This value
can then be used in other statements such as $\mathsf{FILTER}{s}$ or $\mathsf{ORDER}$.
Also in the setting of trust management, Schenk~\cite{Schenk08} defines a \emph{bilattice} structure to model
\emph{trust} relying on the dimensions of knowledge and truth.  The defined knowledge about trust in information sources
can then be used to compute the trust of an inferred statement. An extension towards OWL is presented but there is no
query language defined.  Finally, this approach is used to resolve inconsistencies in ontologies arising from connecting
multiple data sources.

In~\cite{Dividino09} the authors also present a generic extension of RDF to represent \emph{meta information}, mostly
focused on provenance and uncertainty.  Such meta information is stored using named graphs and their extended semantics
of RDF, denoted RDF$^+$, assumes a predefined vocabulary to be interpreted as the meta information.  However they do not
provide an extension of the RDFS inference rules or any operations for combining meta information.  The authors also
provide an extension of the SPARQL query language, considering an additional expression that enables querying the RDF
meta information.

Our initial approach of using residuated lattices as the structure for representing
annotations~\cite{Straccia10a,Lopes10a} was extended to the more general semiring structure by Buneman \&
Kostylev~\cite{Buneman10}.  This paper also shows that, once the RDFS inferences of an RDF graph have been computed for
a specific domain, it is possible to reuse these inferences if the graph is annotated with a different domain.  Based on
this result the authors define a universal domain which is possible to transform to other domains by applying the
corresponding transformations.

Aidan Hogan's thesis~\cite[Chapter 6]{hogan-2011} provides a framework for a specific combination of annotations (authoritativeness, rank, blacklisting) within RDFS and (a variant of) OWL 2 RL. This work is orthogonal to ours, in that it does not focus on aspects of query answering, or providing a generic framework for combinations of annotations, but rather on scalable and efficient algorithms for materialising inferences for the specific combined annotations under consideration.

\section{Preliminaries -- Classical RDF and RDFS}
\label{sec:preliminaries}
In this section we present notions and definitions that are necessary for our discussions later. First we give a short
overview of RDF and RDFS.

\subsection{Syntax} \label{sec:rdf-syntax} Consider pairwise disjoint alphabets $\AU$, $\AB$, and $\AL$ denoting,
respectively, \emph{URI references}, \emph{blank nodes} and \emph{literals}.\footnote{We assume $\AU, \AB$, and $\AL$
  fixed, and for ease we will denote unions of these sets simply concatenating their names.}  We call the elements in
$\AUBL$ ($\AB$) \emph{terms} (\emph{variables}, denoted $x,y,z$).
An \emph{RDF triple} is $\tau=\triple{s,p,o} \in \AUBL \times \AU \times \AUBL$.\footnote{As in~\cite{Munoz07} we allow
  literals for $s$.} We call $s$ the \emph{subject}, $p$ the \emph{predicate}, and $o$ the \emph{object}. A \emph{graph}
$G$ is a set of triples, the \universe of $G$, $\universe(G)$, is the set of elements in $\AUBL$ that occur in the
triples of $G$, the \emph{vocabulary} of $G$, $voc(G)$, is $\universe(G) \cap \AUL$.

We rely on a fragment of RDFS, called $\rhodf$~\cite{Munoz07}, that covers essential features of RDFS.  $\rhodf$ is
defined as the following subset of the RDFS vocabulary: $\rhodf = \{ \spp, \subclass, \type, \dom,
\range\}$. Informally, \ii{i} $\triple{p, \spp, q}$ means that property $p$ is a \emph{subproperty} of property $q$;
\ii{ii} $\triple{c, \subclass, d}$ means that class $c$ is a \emph{subclass} of class $d$; \ii{iii} $\triple{a, \type,
  b}$ means that $a$ is of \emph{type} $b$; \ii{iv} $\triple{p, \dom, c}$ means that the \emph{domain} of property $p$
is $c$; and \ii{v} $\triple{p, \range, c}$ means that the \emph{range} of property $p$ is $c$.
In what follows we define a \emph{map} %
 as a function $\mu : \AUBL \to \AUBL$ preserving URIs and literals, \ie $\mu(t) =
t$, for all $t \in \AUL$. Given a graph $G$, we define $\mu(G)= \{ \triple{\mu(s), \mu(p), \mu(o)} \mid \triple{s, p, o}
\in G \}$. We speak of a map $\mu$ from $G_{1}$ to $G_{2}$, and write $\mu : G_{1} \to G_{2}$, if $\mu$
is such that $\mu(G_{1}) \subseteq G_{2}$.

\subsection{Semantics}  
An \emph{interpretation} $\I$ over a vocabulary $V$ is a tuple $\I =\tuple{\Delta_{R}, \Delta_{P}, \Delta_{C},
  \Delta_{L}, \intP{\cdot}, \intC{\cdot}, \int{\cdot}}$, where $\Delta_{R}, \Delta_{P}$, $\Delta_{C}, \Delta_{L}$ are
the interpretation domains of $\I$, which are finite non-empty sets, and $\intP{\cdot}, \intC{\cdot}, \int{\cdot}$ are
the interpretation functions of $\I$. They have to satisfy: {\small
\begin{enumerate}
\item $\Delta_{R}$ are the resources (the domain or universe of $\I$);
\item $\Delta_{P}$ are property names (not necessarily disjoint from $\Delta_{R}$); 
\item $\Delta_{C} \subseteq \Delta_{R}$ are the classes; 
\item $\Delta_{L} \subseteq \Delta_{R}$ are the literal values and contains $\AL \cap V$;
\item  $\intP{\cdot}$ is a function $\intP{\cdot}\colon \Delta_{P} \to 2^{\Delta_{R} \times \Delta_{R}}$;
\item  $\intC{\cdot}$  is a function  $\intC{\cdot}\colon \Delta_{C} \to 2^{\Delta_{R}}$;
\item $\int{\cdot}$ maps each $t \in \AUL \cap V$ into a value $\int{t} \in \Delta_{R} \cup \Delta_{P}$, and such that
  $\int{\cdot}$ is the identity for plain literals and assigns an element in $\Delta_{R}$ to each element in $\AL$;
\end{enumerate}
}

\nd An interpretation $\I$ is a \emph{model} of a ground graph $G$, denoted $\I \models G$, \iff $\I$ is an interpretation over
the vocabulary $\rhodf \cup \universe(G)$ that satisfies the following conditions:

 {\small
 \setlength{\leftmargini}{3mm}
  \begin{description}\label{condRDF}
 \item[Simple:] \
 \begin{enumerate}
\setlength{\itemindent}{-7mm}
 \item for each $\triple{s, p, o} \in G$, $\intA{p} \in\Delta_{P}$ and $(\intA{s}, \intA{o}) \in \intP{\intA{p}}$;
 \end{enumerate}
 \item[Subproperty:] \
 \begin{enumerate}
\setlength{\itemindent}{-7mm}
 \item $\intP{\intA{ \spp}}$ is transitive over $\Delta_{P}$;
 \item if $(p, q) \in \intP{\intA{ \spp}}$ then $p, q \in \Delta_{P}$ and $\intP{p} \subseteq \intP{q}$; 
 \end{enumerate}
 \item[Subclass:] \
 \begin{enumerate}
\setlength{\itemindent}{-7mm}
 \item $\intP{\intA{\subclass}}$ is transitive over $\Delta_{C}$; 
 \item if $(c, d) \in \intP{\intA{\subclass}}$ then $c, d \in \Delta_{C}$ and $\intC{c} \subseteq \intC{d}$; 
 \end{enumerate}
 \item[Typing I:] \
 \begin{enumerate}
\setlength{\itemindent}{-7mm}
 \item $x \in \intC{c}$ \iff $(x,c) \in \intP{\intA{\type}}$;
 \item if $(p, c) \in \intP{\intA{\dom}}$ and $(x, y) \in \intP{p}$ then $x \in \intC{c}$;
 \item if $(p, c) \in \intP{\intA{\range}}$ and $(x, y) \in \intP{p}$ then $y \in \intC{c}$;
 \end{enumerate}
 \item[Typing II:] \
 \begin{enumerate}
\setlength{\itemindent}{-7mm}
 \item For each $\eee \in \rhodf$, $\intA{\eee} \in \Delta_{P}$
 \item if $(p, c) \in \intP{\intA{\dom}}$ then $p \in \Delta_{P}$ and $c \in \Delta_{C}$
 \item if $(p, c) \in \intP{\intA{\range}}$ then $p \in \Delta_{P}$ and $c \in \Delta_{C}$
 \item if $(x, c) \in \intP{\intA{\type}}$ then $c \in \Delta_{C}$
 \end{enumerate}
 \end{description}
 }

\nd Entailment among ground graphs $G$ and $H$ is as usual.
Now, $G \models H$, where $G$ and $H$ may contain blank nodes, \iff for any grounding $G'$ of $G$ there is a grounding $H'$ of $H$ such that $G' \models H'$.\footnote{A grounding $G'$ of graph $G$ is obtained, as usual, by replacing variables in $G$ with terms in \AUL.}

\begin{remark}\label{rem1}
  In~\cite{Munoz07}, the authors define two variants of the semantics: the default one includes reflexivity of
  $\intP{\int{ \spp}}$ (resp. $\intC{\int{\subclass}}$) over $\Delta_{P}$ (resp. $\Delta_{C}$) but we are only
  considering the alternative semantics presented in~\cite[Definition 4]{Munoz07} which omits this requirement.  Thus,
  we do not support an inference such as $G \models \triple{a, \subclass, a}$, which anyway are of marginal interest.
\end{remark}

\begin{remark}\label{rem3}
  In a First-Order Logic (FOL) setting, we may interpret classes as unary
  predicates, and (RDF) predicates as binary predicates. Then
  \begin{enumerate}
  \item a subclass relation between class $c$ and $d$ may be encoded as the formula $\forall x. c(x) \implies d(x)$
  \item a subproperty relation between property $p$ and $q$ may be encoded as $\forall x\forall y. p(x,y) \implies
    q(x,y)$
  \item domain and range properties may be represented as: $\forall x \forall y. p(x,y) \implies c(x)$ and $\forall x
    \forall y. p(x,y) \implies c(y)$
  \item the transitivity of a property can be represented as $\forall x\forall y\exists z. (p(x,z) \land p(z,y))
    \implies p(x,y)$
  \end{enumerate}

\nd  Although this remark is trivial, we will see that it will play an important
  role in the formalisation of annotated RDFS.
\end{remark}

\subsection{Deductive system} 
In what follows, we provide the sound and complete deductive system for our language derived from~\cite{Munoz07}.  %
The system is arranged in groups of rules
that captures the semantic conditions of models. In every rule, $A,B,C,X$, and $Y$ are meta-variables representing
elements in $\AUBL$ and $D,E$ represent elements in $\AUL$.  The rules are as follows:
{\small
 \setlength{\leftmargini}{5mm}
  \begin{enumerate}
  \item Simple: \\[0.5em]
    \begin{tabular}{llll}
      $(a)$ & $\frac{G}{G'}$ for a map $\mu:G' \to G$ & $(b)$ & $\frac{G}{G'}$ for  $G' \subseteq G$ 
    \end{tabular}
  \item Subproperty: \\[0.5em]
    \begin{tabular}{llll}
      $(a)$ & $\frac{\triple{A,  \spp, B},  \triple{B,  \spp, C}}{\triple{A,  \spp, C}}$ & $(b)$ & $\frac{\triple{D,  \spp, E},  \triple{X, D, Y}}{\triple{X, E, Y}}$ 
    \end{tabular}
  \item Subclass: \\[0.5em]
    \begin{tabular}{llll}
      $(a)$ & $\frac{\triple{A, \subclass, B},  \triple{B, \subclass, C}}{\triple{A, \subclass, C}}$ & $(b)$ & $\frac{\triple{A, \subclass, B},  \triple{X, \type, A}}{\triple{X, \type, B}}$ 
    \end{tabular}
  \item Typing: \\[0.5em]
    \begin{tabular}{llll}
      $(a)$ & $\frac{\triple{D, \dom, B},  \triple{X, D, Y}}{\triple{X, \type, B}}$ & $(b)$ & $\frac{\triple{D, \range, B},  \triple{X, D, Y}}{\triple{Y, \type, B}}$ 
    \end{tabular}
  \item Implicit Typing: \\[0.5em]
    \begin{tabular}{llll}
      $(a)$ & $\frac{\triple{A, \dom, B},  \triple{D,  \spp, A}, \triple{X, D, Y}}{\triple{X, \type, B}}$  &
      $(b)$ & $\frac{\triple{A, \range, B},  \triple{D,  \spp, A}, \triple{X, D, Y}}{\triple{Y, \type, B}}$
    \end{tabular}
  \end{enumerate}
}
\noindent A reader familiar with~\cite{Munoz07} will notice that these rules are as rules 1-5 of~\cite{Munoz07}
(which has 7 rules). We excluded the rules handling reflexivity (rules 6-7) which are not needed to answer queries.
Furthermore, as noted in~\cite{Munoz07}, the ``Implicit Typing'' rules are a necessary addition to the rules presented
in~\cite{haye-04} for complete RDFS entailment.  These represent the case when variable $A$ in $\triple{D, \spp, A}$ and
$\triple{A, \dom, B}$ or $\triple{A, \range, B}$, is a property implicitly represented by a blank node. 

We denote with  $\{\tau_{1}, \ldots, \tau_{n}\} \vdash_{\mathsf{RDFS}} \tau$ that the consequence $\tau $ is obtained from the premise $\tau_{1}, \ldots, \tau_{n}$ by applying one of the inference rules 2-5 above. Note that $n \in\{2,3\}$.  $\vdash_{\mathsf{RDFS}}$ is extended to the set of all RDFS rules as well, in which case $n\in \{1,2,3\}$.

If a graph $G'$ can be obtained by recursively applying rules 1-5 from a graph $G$, the sequence of applied rules is
called a \emph{proof}, denoted $G \vdash G'$, of $G'$ from $G$.  The following proposition shows that our proof
mechanism is sound and complete \wrt the \rhodf semantics:

\begin{proposition}[Soundness and completeness~\cite{Munoz07}]
  Inference $\vdash$ based on rules 1-5 as of~\cite{Munoz07} and applied to our semantics defined above is sound and
  complete for $\models$, that is, $G \vdash G'$ \iff $G \models G'$.
\end{proposition}

\begin{proposition}[\cite{Munoz07}] \label{pp1}
 Assume  $G \vdash G'$ then there is a proof of $G'$ from $G$ where the rule $(1a)$ is used at most once and at the end.
\end{proposition}

\nd Finally, the \emph{closure} of a graph $G$ is defined as $cl(G) = \{\tau \mid G \vdash^{*}~\tau\}$, where $\vdash^{*}$ is as
$\vdash$ except that rule $(1a)$ is excluded. Note that the size of the closure of $G$ is polynomial in the size of $G$
and that the closure is \emph{unique}. Now we can prove that:

\begin{proposition} \label{pp2}
 $G \vdash G'$ \iff $G' \subseteq cl(G)$ or $G'$ is obtained from $cl(G)$ by applying rule $(1a)$.
\end{proposition}

\subsection{Query Answering}
Concerning query answering, we are inspired by~\cite{Gutierrez04} and the Logic Programming setting and we assume that a
RDF graph $G$ is \emph{ground}, that is blank nodes have been skolemised, \ie replaced with terms in $\AUL$.

A \emph{query} is of the rule-like form $$q(\vec{x}) \leftarrow \exists \vec{y}.\varphi(\vec{x},\vec{y})$$ where
$q(\vec{x})$ is the \emph{head} and $\exists \vec{y}.\varphi(\vec{x},\vec{y})$ is the \emph{body} of the query, which is
a conjunction (we use the symbol $``,''$ to denote conjunction in the rule body) of triples $\tau_{i}$ ($1 \leq i \leq n$). $\vec{x}$ is a vector of variables occurring in
the body, called the \emph{distinguished variables}, $\vec{y}$ are so-called \emph{non-distinguished variables} and are
distinct from the variables in $\vec{x}$, each variable occurring in $\tau_{i}$ is either a distinguished or a
non-distinguished variable. If clear from the context, we may omit the existential quantification $\exists \vec{y}$.  

In a query, we allow built-in triples of the form $\triple{s,p,o}$, where $p$ is a built-in predicate taken from a reserved vocabulary and having a \emph{fixed
  interpretation}. We generalise the built-ins to any $n$-ary predicate $p$, where
$p$'s arguments may be $\rhodf$ variables, values from $\AUL$, and $p$ has a
fixed interpretation. We will assume that the evaluation of the predicate can be decided in finite time.  For
convenience, we write ``functional predicates''\footnote{A predicate $p(\vec{x},y)$ is functional if for any $\vec{t}$ there is  \emph{unique} $t'$ for which $p(\vec{t},t')$ is true.} as \emph{assignments} of the form $x\assign f(\vec{z})$ and assume that the function $f(\vec{z})$ is safe. We also assume that a non functional built-in predicate $p(\vec{z})$ should be safe as well.

A query example is:
\begin{tabbing}
 $q(x,y) \leftarrow$ \= \triple{y, \term{created}, x}, \triple{y, \type, \term{Italian}},\\
                   \> \triple{x, \term{exhibitedAt}, \term{Uffizi}}
\end{tabbing} 
\nd  having intended meaning to retrieve all the artefacts $x$ created by Italian artists $y$, being exhibited at Uffizi Gallery.

In order to define an \emph{answer} to a query we introduce the following:

\begin{definition}[Query instantiation]
Given a vector $\vec{x} = \tuple{x_1, \dots, x_k}$ of variables, a \emph{substitution} over
  $\vec{x}$ is a vector of terms $\vec{t}$ replacing variables in $\vec{x}$ with terms of $\AUBL$. Then,
given a query $q(\vec{x}) \leftarrow \exists \vec{y}.\varphi(\vec{x},\vec{y})$, and two substitutions $\vec{t},
  \vec{t'}$ over $\vec{x}$ and $\vec{y}$, respectively, the \emph{query instantiation} $\varphi(\vec{t}, \vec{t}')$ is derived from $\varphi(\vec{x},
  \vec{y})$ by replacing  $\vec{x}$ and $\vec{y}$ with $\vec{t}$ and
  $\vec{t}'$, respectively.
\end{definition}

\nd Note that a query instantiation is an RDF graph.

\begin{definition}[Entailment]
  \label{def:entailment}
  Given a graph $G$, a query $q(\vec{x}) \leftarrow \exists \vec{y}.\varphi(\vec{x},\vec{y})$, and a vector $\vec{t}$ of
  terms in $\universe(G)$, we say that $q(\vec{t})$ is \emph{entailed} by $G$, denoted $G \models q(\vec{t})$, \iff in
  any model $\I$ of $G$, there is a vector $\vec{t}'$ of terms in $\universe(G)$ such that $\I$ is a model of the query
  instantiation $\varphi(\vec{t}, \vec{t}')$.
\end{definition}

\begin{definition}
  If $G \models q(\vec{t})$ then $\vec{t}$ is called an \emph{answer} to $q$. The \emph{answer set} of $q$ \wrt~$G$ is
  defined as $ans(G,q) = \{\vec{t} \mid G \models q(\vec{t}) \}$.
\end{definition}
We next show how to compute the answer set.  
The following can be shown:
\begin{proposition}\label{prop:rdf_closure}
  Given a graph $G$, $\vec{t}$ is an \emph{answer} to $q$ \iff there exists an instantiation $\varphi(\vec{t},\vec{t}')$
  that is true in the closure of $G$~(\ie all triples in $\varphi(\vec{t},\vec{t}')$ are in $cl(G)$).  
\end{proposition}
Therefore, we have a simple method to determine $ans(G, q)$. Compute the closure $cl(G)$ of $G$ and store it
into a database, \eg using the method~\cite{iann09}. It is easily verified that any query can be mapped into an SQL
query over the underlying database schema. Hence, $ans(G, q)$ can be determined by issuing such an SQL query to the
database.
\section{RDFS with Annotations}
\label{sec:annotated-rdfs}
This section presents the extension to RDF towards generic annotations.  Throughout this paper we will use an RDF
dataset describing companies, acquisitions between companies and employment history.  This dataset is partially
presented in Figure~\ref{fig:dataset-example}. We consider this data to be annotated with the temporal 
domain, which intuitively means that the annotated triple is valid in dates contained in the annotation interval (the
exact meaning of the annotations will be explained later).  Also, the information in this example can be derived from
Wikipedia and thus we can consider this data also annotated with the provenance domain (although not explicitly
represented in the example). %
We follow the modelling of employment records proposed by DBpedia, for instance a list of employees of Google is
available as members of the class \texttt{\url{http://dbpedia.org/class/yago/GoogleEmployees}}.  For presentation
purposes we use the shorter name \uri{googleEmp}. We also introduce \uri{SkypeCollab} (resp. \uri{EbayCollab}) to represent Skype's (resp. Ebay's) collaborators.
\begin{figure}[t]
  \centering
\begin{small}
\begin{tabular}{|l|}
\hline
{\tiny\ }\\
\fuzzyg{\triple{\uri{youtubeEmp}, \subclass, \uri{googleEmp}}}{[2006,2011]} \\
\fuzzyg{\triple{\uri{steveChen}, \type, \uri{youtubeEmp}}}{[2005,2011]} \\
\fuzzyg{\triple{\uri{chadHurley}, \type, \uri{youtubeEmp}}}{[2005,2010]} \\
\fuzzyg{\triple{\uri{jawedKarim}, \type, \uri{youtubeEmp}}}{[2005,2011]} \\
\fuzzyg{\triple{\uri{jawedKarim}, \type, \uri{paypalEmp}}}{[2000,2005]} \\
\fuzzyg{\triple{\uri{paypalEmp}, \subclass, \uri{ebayEmp}}}{[2002,2011]} \\
\fuzzyg{\triple{\uri{chadHurley}, \type, \uri{paypalEmp}}}{[2002,2005]} \\
\fuzzyg{\triple{\uri{skypeEmp}, \subclass, \uri{ebayEmp}}}{[2005,2011]} \\
\fuzzyg{\triple{\uri{SkypeCollab}, \subclass, \uri{EbayCollab}}}{[2005,2009]} \\
\fuzzyg{\triple{\uri{SkypeCollab}, \subclass, \uri{EbayCollab}}}{[2009,2011]} \\
\fuzzyg{\triple{\uri{niklasZennstrom}, \uri{ceo}, \uri{skype}}}{[2003,2007]} \\
\fuzzyg{\triple{\uri{ceo}, \spp, \uri{worksFor}}}{[-\infty,+\infty]} \\
\fuzzyg{\triple{\uri{larryPage}, \uri{worksFor}, \uri{google}}}{[1998,2011]} \\
\fuzzyg{\triple{\uri{sergeyBrin}, \uri{worksFor}, \uri{google}}}{[1998,2011]} \\
{\tiny\ }\\
\hline
\end{tabular}
\end{small}
  \caption{Company acquisition dataset example}
  \label{fig:dataset-example}
\end{figure}

\subsection{Syntax}

Our approach is to extend triples with annotations, where an annotation is taken from a specific domain.\footnote{The
  readers familiar with the annotated logic programming framework~\cite{Kifer92}, will notice the similarity of the
  approaches.}
  
An \emph{annotated triple} is an expression $\fuzzyg{\tau}{\lambda}$, where $\tau$ is a triple and $\lambda$ is an
\emph{annotation value} (defined below). An \emph{annotated graph} is a finite set of annotated triples.  The intended
semantics of annotated triples depends of course on the meaning we associate to the annotation values. For instance, in
a temporal setting~\cite{gutierrez-etal-2007}, 
\[
\fuzzyg{\triple{\uri{niklasZennstrom}, \uri{ceoOf},    \uri{skype}}}{[2003,2007]}
\]
\nd has intended meaning ``Niklas was CEO of Skype during the period $2003$ to $2007$'', while in the fuzzy
setting~\cite{Straccia09f} $\fuzzyg{\triple{\uri{skype}, \uri{ownedBy}, \uri{bigCompany}}}{0.3}$ has intended meaning
``Skype is owned by a big company to a degree not less than 0.3''.

\subsection{RDFS Annotation Domains}
\label{sec:rdfs-annot-doma}

To start with, let us consider a non-empty set $L$. Elements in $L$ are our annotation values. For example,
in a fuzzy setting, $L = [0,1]$, while in a typical temporal setting, $L$ may be time points or time intervals. 
In our annotation framework, an interpretation will map statements to elements of the annotation domain.  Our semantics
generalises the formulae in Remark~\ref{rem3} by using a well known algebraic structure.

We say that an \emph{annotation domain} for RDFS is an idempotent, commutative semi-ring 
\[
D = \tuple{L,  \oplus, \otimes, \bot, \top} \ ,
\]
\nd where $\oplus$ is $\top$-annihilating~\cite{Buneman10}. That is, for $\lambda, \lambda_{i} \in L$

\begin{enumerate}
\item $\oplus$ is idempotent, commutative, associative;
\item $\otimes$ is commutative and associative;
\item $\bot \oplus \lambda = \lambda$, $\top \otimes \lambda = \lambda$, $\bot \otimes \lambda = \bot$, and $\top \oplus \lambda = \top$; 
\item $\otimes$ is  distributive over $\oplus$, \ie 
$\lambda_{1} \otimes (\lambda_{2} \oplus \lambda_{3}) =  (\lambda_{1} \otimes \lambda_{2}) \oplus (\lambda_{1} \otimes \lambda_{3})$; 
\end{enumerate}

\nd It is well-known that there is a natural partial order on any idempotent semi-ring: an annotation domain 
$D = \tuple{L,  \oplus, \otimes, \bot, \top}$ induces a partial order $\preceq$ over $L$ defined as:

\[
\lambda_{1} \preceq \lambda_{2} \mbox{ \ \iff  \ }  \lambda_{1} \oplus \lambda_{2} = \lambda_{2} \ .
\]

\nd The order $\preceq$ is used to express redundant/entailed/subsumed information. For instance, for temporal intervals, an
annotated triple $\fuzzyg{\triple{s,p,o}}{[2000,2006]}$ entails $\fuzzyg{\triple{s,p,o}}{[2003,2004]}$, as $[2003,2004]
\subseteq [2000,2006]$ (here, $\subseteq$ plays the role of $\preceq$).

\begin{remark}\label{remRes}
In previous work~\cite{Straccia10a,Lopes10a}, an annotation domain was assumed to be a more specific structure, namely a  residuated bounded lattice $D = \tuple{L, \preceq, \land, \lor, \otimes,  \implies, \bot, \top}$. That is, 
\begin{enumerate}
\item  $\tuple{L, \preceq,\land, \lor, \bot, \top}$ is a bounded lattice, where $\bot$ and $\top$ are bottom and top elements, and $\land$ and $\lor$ are meet and join operators;

\item $\tuple{L, \otimes, \top}$ is a commutative monoid.

\item $\implies$ is the so-called residuum of $\otimes$, \ie for all $\lambda_{1},\lambda_{2}, \lambda_{3}$, $\lambda_{1} \otimes \lambda_{3}\preceq \lambda_{2}$ \iff $\lambda_{3} \preceq (\lambda_{1} \implies \lambda_{2})$. 
\end{enumerate}

\nd Note that any bounded residuated lattice satisfies the conditions of an annotation domain. 
In~\cite{Buneman10} it was shown that we may use a slightly weaker structure than residuated lattices for annotation domains.
\end{remark}

\begin{remark}\label{remsemi}
Observe that $\tuple{L,  \preceq, \oplus, \bot, \top}$ is a bounded join semi-lattice.
\end{remark}

\begin{remark}\label{rem:crisp}
  Note that the domain $D_{01} = \tuple{\{0,1\}, \max, \min, 0, 1}$ corresponds to the boolean case. In fact, in this
  case annotated RDFS will turn out to be the same as classical RDFS.
\end{remark}

\begin{remark}\label{rem4}
  We use $\oplus$ to combine information about the same statement.  For instance, in temporal logic, from
  $\fuzzyg{\tau}{[2000,2006]}$ and $\fuzzyg{\tau}{[2003,2008]}$, we infer $\fuzzyg{\tau}{[2000,2008]}$, as $[2000,2008] =
  [2000,$ $2006]\cup [2003,2008]$; here, $\cup$ plays the role of $\oplus$. In the fuzzy context, from $\fuzzyg{\tau}{0.7}$ and
  $\fuzzyg{\tau}{0.6}$, we infer $\fuzzyg{\tau}{0.7}$, as $0.7 = \max(0.7, 0.6)$ (here, $\max$ plays the role of
  $\oplus$).
\end{remark}

\begin{remark}\label{rem5}
  We use $\otimes$ to model the ``conjunction'' of information. In fact, a $\otimes$ is a generalisation of
  boolean conjunction to the many-valued case. In fact, $\otimes$ satisfies also that
 \begin{enumerate}
 \item $\otimes$ is bounded: \ie $\lambda_{1} \otimes \lambda_{2} \preceq \lambda_{1}$.   

 \item $\otimes$ is $\preceq$-monotone, \ie~for $\lambda_{1} \preceq \lambda_{2}$, $\lambda \otimes \lambda_{1} \preceq \lambda \otimes \lambda_{2}$
 
\end{enumerate}
\nd   For instance, on interval-valued temporal logic, from $\fuzzyg{\triple{a,
      \subclass, b}}{[2000,2006]}$ and $\fuzzyg{\triple{b, \subclass, c}}{[2003,2008]}$, we will infer
  $\fuzzyg{\triple{a, \subclass, c}}{[2003,2006]}$, as $[2003,2006] = [2000,2006] \cap [2003,2008]$; here, $\cap$ plays
  the role of $\otimes$.\footnote{As we will see, $\oplus$ and $\otimes$ may be more involved.} In the fuzzy context, one
  may chose any t-norm~\cite{HajekP98,Klement00}, \eg product, and, thus, from $\fuzzyg{\triple{a, \subclass, b}}{0.7}$ and $\fuzzyg{\triple{b,
      \subclass, c}}{0.6}$, we will infer $\fuzzyg{\triple{a, \subclass, c}}{0.42}$, as $0.42 = 0.7 \cdot 0.6)$ (here,
  $\cdot $ plays the role of $\otimes$). 
\end{remark}

\begin{remark}\label{remdistr}
Observe that the distributivity condition is used to guarantee that \eg we obtain the same annotation $\lambda \otimes (\lambda_{2} \oplus \lambda_{3}) = (\lambda_{1} \otimes \lambda_{2}) \oplus (\lambda_{1} \otimes \lambda_{3})$ of the triple 
$\triple{a, \subclass, c}$ that can be inferred from triples 
$\fuzzyg{\triple{a, \subclass, b}}{\lambda_{1}}$, 
  $\fuzzyg{\triple{b, \subclass, c}}{\lambda_{2}}$ and 
  $\fuzzyg{\triple{b, \subclass,  c}}{\lambda_{3}}$.

\end{remark}

\nd Finally, note that, conceptually, in order to build an annotation domain, one has to:
\begin{enumerate}
\item determine the set of annotation values $L$ (typically a countable set\footnote{Note that one may use XML decimals in $[0,1]$ in place of real numbers for the fuzzy domain.}), 
identify the top and bottom elements;

\item define a suitable operations $\otimes$ and $\oplus$ that acts as ``conjunction'' and ``disjunction'' function, to support the intended inference over
  schema axioms, such as 
 \begin{quote}
  ``from $\fuzzyg{\triple{a, \subclass, b}}{\lambda}$ and $\fuzzyg{\triple{b, \subclass,
      c}}{\lambda'}$ infer $\fuzzyg{\triple{a, \subclass, c}}{\lambda \otimes \lambda'}$''
\end{quote}       

\nd and
 \begin{quote}
  ``from $\fuzzyg{\tau}{\lambda}$ and $\fuzzyg{\tau}{\lambda'}$ infer $\fuzzyg{\tau}{\lambda \oplus \lambda'}$''
\end{quote}       

\end{enumerate}

\subsection{Semantics}
Fix an annotation domain $D = \tuple{L,  \oplus, \otimes, \bot, \top}$. 
Informally, an interpretation $\I$ will assign to a triple
$\tau$ an element of the annotation domain $\lambda \in L$.
Formally, an \emph{annotated interpretation} $\I$ over a vocabulary $V$ is a tuple
$$\I=\tuple{\Delta_{R}, \Delta_{P}, \Delta_{C}, \Delta_{L}, \intP{\cdot}, \intC{\cdot}, \int{\cdot}}$$ where $\Delta_{R},
\Delta_{P}, \Delta_{C}, \Delta_{L}$ are interpretation domains of $\I$ and $\intP{\cdot}, \intC{\cdot}, \int{\cdot}$ are
interpretation functions of $\I$.

They have to satisfy:
\begin{enumerate}
\item $\Delta_{R}$ is a nonempty finite set of resources, called the domain or universe of $\I$;

\item $\Delta_{P}$ is a finite set of property names (not necessarily disjoint from $\Delta_{R}$);

\item $\Delta_{C} \subseteq \Delta_{R}$ is a distinguished subset of $\Delta_{R}$ identifying if a resource denotes a
  class of resources;

\item $\Delta_{L} \subseteq \Delta_{R}$, the set of literal values, $\Delta_{L}$ contains all plain literals in $\AL
  \cap V$;

\item\label{item5} $\intP{\cdot}$ maps each property name $p \in \Delta_{P}$ into a function $\intP{p}: \Delta_{R} \times
  \Delta_{R}\to L$, \ie assigns an annotation value to each pair of resources;

\item\label{item6} $\intC{\cdot}$ maps each class $c \in \Delta_{C}$ into a function $\intC{c} : \Delta_{R} \to L$, \ie assigns an
  annotation value representing class membership in $c$ to every resource;

\item $\int{\cdot}$ maps each $t \in \AUL \cap V$ into a value $\int{t} \in \Delta_{R} \cup \Delta_{P}$ and such that
  $\int{\cdot}$ is the identity for plain literals and assigns an element in $\Delta_{R}$ to each element in $\AL$.
\end{enumerate}

\nd An interpretation $\I$ is a \emph{model} of an annotated ground graph $G$, denoted $\I \models G$, \iff $\I$ is an interpretation over
the vocabulary $\rhodf \cup \universe(G)$ that satisfies the following conditions:

{\small
 \setlength{\leftmargini}{3mm}
\begin{description}
\item[Simple:] \
\begin{enumerate}
\setlength{\itemindent}{-5mm}
\item $\fuzzyg{\triple{s, p, o}}{\lambda} \in G$ implies $\int{p} \in\Delta_{P}$ and $\intP{\int{p}}(\int{s}, \int{o})
  \succeq \lambda$;
\end{enumerate}
\item[Subproperty:] \
\begin{enumerate}
\setlength{\itemindent}{-5mm}
\item $\intP{\int{ \spp}}(p, q) \otimes \intP{\int{ \spp}}(q, r) \preceq \intP{\int{ \spp}}(p, r)$;  
\item $\intP{\int{ p}}(x, y) \otimes \intP{\int{ \spp}}(p, q) \preceq \intP{\int{ q}}(x, y)$;
\end{enumerate}
\item[Subclass:] \
\begin{enumerate}
\setlength{\itemindent}{-5mm}
\item $\intP{\int{ \subclass}}(c, d) \otimes \intP{\int{ \subclass}}(d, e) \preceq \intP{\int{ \subclass}}(c, e)$;  
\item $\intC{\int{c}}(x) \otimes \intP{\int{ \subclass}}(c, d) \preceq \intP{\int{ d}}(x)$;  
\end{enumerate}
\item[Typing I:] \
\begin{enumerate}
\setlength{\itemindent}{-5mm}
\item $\intC{c}(x) =  \intP{\int{\type}}(x,c)$;
\item $\intP{\int{\dom}}(p, c) \otimes  \intP{p}(x,y) \preceq  \intC{c}(x)$;
\item $\intP{\int{\range}}(p, c) \otimes  \intP{p}(x,y) \preceq  \intC{c}(y)$;
\end{enumerate}
\item[Typing II:] \
\begin{enumerate}
\setlength{\itemindent}{-5mm}
\item For each $\eee \in \rhodf$, $\int{\eee} \in \Delta_{P}$;
\item $\intP{\int{\spp}}(p, q)$ is defined only for $p,q \in \Delta_{P}$;
\item $\intC{\int{\subclass}}(c, d)$ is defined only for $c,d \in \Delta_{C}$;
\item $\intP{\int{\dom}}(p, c)$ is defined only for $p \in \Delta_{P}$ and $c \in \Delta_{C}$;
\item  $\intP{\int{\range}}(p, c)$ is defined only for $p \in \Delta_{P}$ and $c \in \Delta_{C}$;
\item $\intP{\int{\type}}(s, c)$ is defined only for $c \in \Delta_{C}$.
\end{enumerate}
\end{description}
}

\nd Intuitively, a triple $\fuzzyg{\triple{s, p, o}}{\lambda}$ is %
satisfied by %
$\I$ if $(s, o)$ belongs to the extension of $p$ to a ``wider'' extent than $\lambda$.
Note that the major differences from the classical setting relies on items~\ref{item5} and~\ref{item6}. 

We further note
that the classical setting is as the case in which the annotation domain is $D_{01}$ where $L=\{0,1\}$.

Finally, entailment among annotated ground graphs $G$ and $H$ is as usual.
Now, $G \models H$, where $G$ and $H$ may contain blank nodes, \iff for any grounding $G'$ of $G$ there is a grounding $H'$ of $H$ such that $G' \models H'$.

\begin{remark}\label{remBot}
  Note that we always have that $G \models \fuzzyg{\tau}{\bot}$. Clearly, triples of the form $\fuzzyg{\tau}{\bot}$ are
  uninteresting and, thus, in the following we do not consider them as part of the language.
\end{remark}

\nd As for the crisp case, it can be shown that:
\begin{proposition}
Any annotated RDFS graph has a finite model.
\end{proposition}

\begin{proof}
  Let $G$ be an annotated graph over domain $D$.  Let $Lit = L \cap
  \mathit{universe}(G)$ be the set of literals present in $G$ and $l_0 \in Lit$.  We define the interpretation
  $\mathcal{I}$ over $V$ as follows: {\small
    \begin{enumerate}
    \item $\Delta_{R} = \Delta_{P} = \Delta_{C} = Lit = \Delta_{L} = Lit$;
    \item $\forall x,y,p\ \intP{p}(x,y) \mapsto \top$;
    \item $\forall x,c\ \intC{c}(x) \mapsto \top$;
    \item
      \begin{enumerate}
      \item $\forall l \in L, \int{l} = l$
      \item $\forall x \in V, \int{l} = l_0$
      \end{enumerate}
    \end{enumerate}
  } 
  \nd It is easy to see that $\mathcal{I}$ satisfies all the conditions of RDF-satisfiability and thus is a model
  of $G$.
\end{proof}
\noindent Therefore, we do not have to care about consistency.

\subsection{Examples of primitive domains}
\label{sec:domain-examples}

\noindent To demonstrate the power of our approach, we illustrate its application to some domains:
fuzzy~\cite{Straccia09f}, temporal~\cite{gutierrez-etal-2007} and
provenance.%

\subsubsection{The fuzzy domain}
\noindent To model fuzzy RDFS~\cite{Straccia09f} we may define the annotation domain as $D_{[0,1]} = \tuple{[0,1], \max,
  \otimes, 0, 1}$ where $\otimes$ is any continuous t-norm on $[0,1]$.
\begin{example}
  \label{exUSf}
  Adapting our example of employment records to the fuzzy domain we can state the following: Skype collaborators are also
  Ebay collaborators to some degree since Ebay possesses 30\% of Skype's shares, and also that Toivo is a part-time Skype collaborator: {\small
\[
\begin{array}{l}
  \fuzzyg{\triple{\uri{SkypeCollab}, \subclass, \uri{EbayCollab}}}{0.3} \\
  \fuzzyg{\triple{\uri{toivo}, \type, \uri{SkypeCollab}}}{0.5} 
\end{array}
\]
}
\nd Then, \eg~under the product t-norm $\otimes$, we can infer the following triple:
  {\small
    \[
    \begin{array}{l}
      \fuzzyg{\triple{\uri{toivo}, \type, \uri{EbayCollab}}}{0.15} 
    \end{array}
    \]
  } 

\end{example}
\subsubsection{The temporal domain}
\label{sec:temporal} 

Most of the semantic information on the Web deals with time in an implicit or explicit way. Social relation graphs,
personal profiles, information about various entities continuously evolve and do not remain static. This dynamism can
take various forms: certain information is only valid in a specific time interval (\eg somebody's address), some data
talks about events that took place at a specific time point in the past (\eg beginning of a conference), some data
describe eternal truth (\eg tigers are mammals), or truth that is valid from a certain point of time onwards forever
(e.g Elvis is dead), or creation or change dates of online information items (\eg the edit history of a wiki
page).
We believe that treating web data in a time-sensitive way is one of the biggest steps towards turning the Semantic Web
idea into reality.

\paragraph*{Precise temporal information}
For our representation of the temporal domain we aim at using non-discrete time as it is necessary to model temporal
intervals with any precision.  
however, for presentation purposes we will show the dates as years only.

\paragraph{Modelling the temporal domain}
\label{sec:temporal-modelling}

To start with, \emph{time points} are elements of 
the value space~$\mathbb{Q} \cup\ \{-\infty,
+\infty\}$.
A \emph{temporal interval} is a non-empty interval $[\alpha_1, \alpha_2]$, where
$\alpha_{i}$ are time points.  An empty interval is denoted as $\emptyset$.
We define a partial order on intervals as $I_{1} \leq I_{2} \mbox{ \iff\ } I_{1} \subseteq I_{2}$.  The intuition here is
that if a triple is true at time points in $I_{2}$ and $I_{1} \leq I_{2}$ then, in particular, it is true at any time
point in $I_{1} \neq \emptyset$.

Now, apparently the set of intervals would be a candidate for $L$, which however is not the case. The reason is that,
\eg in order to represent the upper bound interval of $\fuzzyg{\tau}{[1,5]}$ and $\fuzzyg{\tau}{[8,9]}$ we rather need
the union of intervals, denoted $\{[1,5], [8,9]\}$, meaning that a triple is true both in the former as well as
in the latter interval. %
Now, we define $L$ as (where $\bot =
\{\emptyset\}, \top = \{[-\infty,+\infty]\}$) {\small
\[
  L = \{t \mid t \mbox{\small \ is a finite set of disjoint temporal intervals} \}
  \cup \{\bot,\top\} \ .
\]
}
\noindent Therefore, a \emph{temporal term} is an element $t \in L$, \ie a set of pairwise disjoint time intervals. We
allow to write $[\alpha]$ as a shorthand for $[\alpha, \alpha]$, $\fuzzyg{\tau}{\alpha}$ as a shorthand of
$\fuzzyg{\tau}{\{[\alpha]\}}$ and $\fuzzyg{\tau}{[\alpha, \alpha']}$ as a shorthand of $\fuzzyg{\tau}{\{[\alpha,
  \alpha']\}}$. Furthermore, on $L$ we define the following partial order:
\[
t_1 \preceq t_2 \mbox{ \iff } \forall I_{1} \in t_1 \exists I_2 \in t_2, \mbox{ such that } I_1 \leq I_2 \ .
\]
\noindent Please note that $\preceq$ is the Hoare order on power sets~\cite{Abramsky94}, which is a pre-order. For the
anti-symmetry property, assume that $t_1 \preceq t_2$ and $t_2 \preceq t_1$: so for $I_{1} \in t_{1}$, there is $I_{2}
\in t_{2}$ for which there is $I_{3} \in t_{1}$ such that $I_{1} \subseteq I_{2} \subseteq I_{3}$. But, $t_{1}$ is
maximal and, thus, $I_{1} = I_{3} = I_{2}$. So, $t_{1} = t_{2}$ and, thus, $\preceq$ is a partial order. Similarly as
for time intervals, the intuition for $\preceq$ is that if a triple is true at time points in intervals in $t_{2}$ and
$t_{1} \preceq t_{2}$, then, in particular, it is true at any time point in intervals in $t_{1}$. Essentially, if $t_{1}
\preceq t_{2}$ then a temporal triple $\fuzzyg{\tau_{2}}{t_{2}}$ is true to a larger ``temporal extent'' than the
temporal triple $\fuzzyg{\tau_{1}}{t_{1}}$. It can also be verified that $\tuple{L, \preceq, \bot, \top}$ is a bounded
lattice. Indeed, to what concerns us, the partial order $\preceq$ induces the following join ($\oplus$) operation on
$L$. Intuitively, if a triple is true at $t_{1}$ and also true at $t_{2}$ then it will be true also for time points
specified by $t_{1} \oplus t_{2}$ (a kind of union of time points). As an example, if $\fuzzyg{\tau}{\{[2,5], [8,12]\}}$ and
$\fuzzyg{\tau}{\{[4,6], [9,15]\}}$ are true then we expect that this is the same as saying that $\fuzzyg{\tau}{\{[2,6],
  [8,15]\}}$ is true. The join operator will be defined in such way that $\{[2,5], [8,12]\} \oplus \{[4,6], [9,15]\} =
\{[2,6], [8,15]\}$. Operationally, this means that $t_{1} \oplus t_{2}$ will be obtained as follows: \ii{i} take the union
of the sets of intervals $t = t_{1} \cup t_{2}$; and \ii{ii} join overlapping intervals in $t$ until no more overlapping
intervals can be obtained. Formally, {\small
\[
t_1 \oplus t_2 =  \inf \{t \mid t \succeq t_{i}, i=1,2 \} \ .
\]
}
\nd It remains to define the meet $\otimes$ over sets of intervals. Intuitively, we would like to support
inferences such as ``from $\fuzzyg{\triple{a, \subclass, b}}{\{[2,5], [8,12]\}}$ and $\fuzzyg{\triple{b, \subclass, c}}{
  \{[4,6], [9,15]\}}$ infer $\fuzzyg{\triple{a, \subclass, b}}{\{[4,5], [9,12]\}}$'', where $\{[2,5], [8,12]\} \otimes
\{[4,6]$, $[9,15]\}= \{[4,5], [9,12]\}$.  We get it by means of {\small
\[
t_1 \otimes t_2 = \sup \{t \mid t \preceq t_{i}, i=1,2 \}  \ .
\]
}
\nd Note that here the t-norm used for modelling ``conjunction'' coincides with the lattice meet operator.

\begin{example}\label{exU1}
  Using the data from our running example, we can infer that
  $$\fuzzyg{\triple{\uri{chadHurley}, \type, \uri{googleEmp}}}{[2006, 2010]}$$
  where 
  $$\{[2005,2010]\} \otimes \{[2006,2011]\} = \{[2006,2010]\}$$
\end{example}

\noindent In~\cite{gutierrez-etal-2007} are described some further features such as a ``Now'' time point (which is just
a defined time point in $D_{T}$) and anonymous time points, allowing to state that a triple is true at some
point. Adding anonymous time points would require us to extend the lattice by appropriate operators, \eg $[4, T] \oplus
[T, 8] = [4,8]$ (where $T$ is an anonymous time point), etc. %

\subsubsection{Provenance domain}
\label{sec:provenance-domain}

Identifying provenance of triples is regarded as an important issue for dealing with the heterogeneity of Web Data, and several proposals have been made to model provenance~\cite{Ding_Finin_Peng_PinheirodaSilva_McGuinness:05,Carroll_Bizer_Hayes_Stickler:05,Flouris_Fundulaki_Pediaditis_Theoharis_Christophides:09,Hartig:09}. Typically, provenance is identified by a URI, usually the URI of the document in which the triples are defined or possibly a URI identifying a name graph. However, provenance of inferred triples is an issue that have been little tackled in the literature~\cite{Delbru_Polleres_Tummarello_Decker:08,Flouris_Fundulaki_Pediaditis_Theoharis_Christophides:09}. We propose to address this issue by introducing an annotation domain for provenance.

The intuition behind our approach is similar to the one of~\cite{Delbru_Polleres_Tummarello_Decker:08} and~\cite{Flouris_Fundulaki_Pediaditis_Theoharis_Christophides:09} where provenance of an inferred triple is defined as the aggregation of provenances of documents that allow to infer that triple. For instance, if a document $d_1$ defines $\triple{\uri{youtubeEmp}, \subclass, \uri{googleEmp}}{:}d_1$ and a second document $d_2$ defines $\triple{\uri{chadHurley}, \type, \uri{youtubeEmp}}{:}d_2$, then we can infer $\triple{\uri{chadHurley}, \type, \uri{googleEmp}}{:}d_1\land d_2$.

Such a mechanism makes sense and would fit well as a meet operator, but these approaches do not address the join operation which should take place when identical triples are annotated differently. We improve this with the following formalisation.

\paragraph{Modelling the provenance domain}
\label{sec:provenance-modelling}

We start from a countably infinite set of \emph{atomic provenances} \AP which, in practice, can be represented by URIs. We consider the propositional formulae made from symbols in \AP (atomic propositions), logical \textsl{or} ($\lor$) and logical \textsl{and} ($\land$), for which we have the standard entailment $\models$. A \emph{provenance value} is an equivalent class for the logical equivalence relation, \ie~the set of annotation values is the quotient set of $\AP$ by the logical equivalence. The order relation is $\models$, $\otimes$ and $\oplus$ are $\land$ and $\lor$ respectively. We set $\top$ to \textsl{true} and $\bot$ to \textsl{false}.

\begin{example}
  Consider the following data:
\[
\begin{array}{l}
\fuzzyg{\triple{\uri{chadHurley}, \uri{worksFor}, \uri{youtube}}}{\texttt{\footnotesize chad}}\\
\fuzzyg{\triple{\uri{chadHurley}, \type, \uri{Person}}}{\texttt{\footnotesize chad}}\\
\fuzzyg{\triple{\uri{youtube}, \type, \uri{Company}}}{\texttt{\footnotesize chad}}\\
\fuzzyg{\triple{\uri{Person}, \subclass, \uri{Agent}}}{\texttt{\footnotesize foaf}} \\
\fuzzyg{\triple{\uri{worksFor}, \dom, \uri{Person}}}{\texttt{\footnotesize workont}} \\
\fuzzyg{\triple{\uri{worksFor}, \range, \uri{Company}}}{\texttt{\footnotesize workont}} 
\end{array}
\]
We can deduce that \uri{chadHurley} is an \uri{Agent} in two different ways: using the first, fourth and fifth statement
or using the second and fourth statement. So, it is possible to infer the following annotated triple:
\begin{tabbing}
 \triple{\uri{chadHurley}, \type, \uri{Agent}}:\=(\texttt{\footnotesize chad}$\land$\texttt{\footnotesize foaf}$\land$\texttt{\footnotesize workont})\\
 \>$\lor$(\texttt{\footnotesize chad}$\land$\texttt{\footnotesize foaf})
\end{tabbing}
However, since $(\texttt{\footnotesize chad}\land\texttt{\footnotesize foaf}\land\texttt{\footnotesize workont})\lor(\texttt{\footnotesize chad}\land\texttt{\footnotesize foaf})$ is logically
equivalent to $\texttt{\footnotesize chad}\land\texttt{\footnotesize foaf}$, the aggregated inference can be collapsed into:
\[
\fuzzyg{\triple{\uri{chadHurley}, \type, \uri{Agent}}}{\texttt{\footnotesize chad}\land\texttt{\footnotesize foaf}}
\]
\end{example}

\nd Intuitively, a URI denoting a provenance can also denote a RDF graph, either by using a named graph approach, or implicitly by getting a RDF document by dereferencing the URI. In this case, we can see the conjunction operation as a union of graphs and disjunction as an intersection of graphs.

\paragraph{Comparison with other approaches}

\cite{Delbru_Polleres_Tummarello_Decker:08} does not formalise the semantics and properties of his aggregation operation (simply denoted by $\land$) nor the exact rules that should be applied to correctly and completely reason with provenance. Query answering is not tackled either.

The authors of~\cite{Flouris_Fundulaki_Pediaditis_Theoharis_Christophides:09} are providing more insight on the formalisation and actually detail the rules by reusing (tacitly)~\cite{Munoz07}. They also provide a formalisation of a simple query language. However, the semantics they define is based on a strong restriction of $\rhodf$\footnote{Remember that \rhodf is already a restriction of RDFS.}. 

As an example, they define the answers to the query $(?x, \type, ?y, ?c)$ as the tuples $(X,Y,C)$ such that there is a
triple $(X, \type, Y, C)$ which can be inferred from only the application of rules (3a) and (3b). This means that a
domain or range assertion would not provide additional answers to that type of query.

Finally, none of those papers discuss the possibility of universally true statements (the $\top$ provenance) or the statements from unknown provenance ($\bot$). They also do not consider mixing non-annotated triples with annotated ones as we do in Section~\ref{sec:classical_triples}.

\subsection{Deductive system}
An important feature of our framework is that we are able to provide a deductive system in the style of the one for
classical RDFS.  Moreover, \emph{the schemata of the rules are the same for any annotation domain} (only support for the
domain dependent $\otimes$ and $\oplus$ operations has to be provided) and, thus, are amenable to an easy implementation
on top of existing systems.  The rules are arranged in groups that capture the semantic conditions of models, $A,B,C,X$
and $Y$ are meta-variables representing elements in $\AUBL$ and $D,E$ represent elements in $\AUL$.  The rule set
contains two rules, $(1a)$ and $(1b)$, that are the same as for the crisp case, while rules $(2a)$ to $(5b)$ are the
annotated rules homologous to the crisp ones. Finally, rule $(6)$ is specific to the annotated case.

Please note that
rule $(6)$ is destructive \ie~this rule removes the premises as the conclusion is inferred.  We also assume that a rule
is not applied if the consequence is of the form $\fuzzyg{\tau}{\bot}$ (see Remark~\ref{remBot}).  It can be shown that:

{\scriptsize
\renewcommand{\arraystretch}{2}
  \begin{enumerate}
  \item Simple: \\[0.5em]
    \begin{tabular}{ll}
      $(a)$ & $\frac{G}{G'}$ for a map $\mu:G' \to G$ \\ $(b)$ & $\frac{G}{G'}$ for  $G' \subseteq G$ 
    \end{tabular}
  \item Subproperty: \\[0.5em]
    \begin{tabular}{ll}
      $(a)$ & $\frac{\fuzzyg{\triple{A,  \spp, B}}{\lambda_{1}},  \fuzzyg{\triple{B,  \spp, C}}{\lambda_{2}}}{\fuzzyg{\triple{A,  \spp, C}}{\lambda_{1}\otimes \lambda_{2}}}$ \\
      $(b)$ & $\frac{\fuzzyg{\triple{D, \spp, E}}{\lambda_{1}},  \fuzzyg{\triple{X, D, Y}}{\lambda_{2}}}{\fuzzyg{\triple{X, E, Y}}{\lambda_{1} \otimes \lambda_{2}}}$ 
    \end{tabular}
  \item Subclass: \\[0.5em]
    \begin{tabular}{ll}
      $(a)$ & $\frac{\fuzzyg{\triple{A, \subclass, B}}{\lambda_{1}},  \fuzzyg{\triple{B, \subclass, C}}{\lambda_{2}}}{\fuzzyg{\triple{A, \subclass, C}}{\lambda_{1} \otimes \lambda_{2}}}$ \\
      $(b)$ & $\frac{\fuzzyg{\triple{A, \subclass, B}}{\lambda_{1}},  \fuzzyg{\triple{X, \type, A}}{\lambda_{2}}}{\fuzzyg{\triple{X, \type, B}}{\lambda_{1} \otimes \lambda_{2}}}$ 
    \end{tabular}
  \item Typing: \\[0.5em]
    \begin{tabular}{ll}
      $(a)$ & $\frac{\fuzzyg{\triple{D, \dom, B}}{\lambda_{1}},  \fuzzyg{\triple{X, D, Y}}{\lambda_{2}}}{\fuzzyg{\triple{X, \type, B}}{\lambda_{1} \otimes \lambda_{2}}}$ \\
  $(b)$ & $\frac{\fuzzyg{\triple{D, \range, B}}{\lambda_{1}},  \fuzzyg{\triple{X, D, Y}}{\lambda_{2}}}{\fuzzyg{\triple{Y, \type, B}}{\lambda_{1} \otimes \lambda_{2}}}$
    \end{tabular}
  \item Implicit Typing: \\[0.5em]
    \begin{tabular}{ll}
      $(a)$ & $\frac{\fuzzyg{\triple{A, \dom, B}}{\lambda_{1}},  \fuzzyg{\triple{D,  \spp, A}}{\lambda_{2}},
        \fuzzyg{\triple{X, D, Y}}{\lambda_{3}}}{\fuzzyg{\triple{X, \type, B}}{\lambda_{1} \otimes \lambda_{2} \otimes
          \lambda_{3}}}$ \\
      $(b)$ & $\frac{\fuzzyg{\triple{A, \range, B}}{\lambda_{1}},  \fuzzyg{\triple{D,  \spp, A}}{\lambda_{2}}, \fuzzyg{\triple{X, D, Y}}{\lambda_{3}}}{\fuzzyg{\triple{Y, \type, B}}{\lambda_{1} \otimes \lambda_{2} \otimes \lambda_{3}}}$ 
    \end{tabular}
  \item  Generalisation:\\[0.5em]
    \begin{tabular}{l}
      $\frac{\fuzzyg{\triple{X, A, Y}}{\lambda_{1}},  \fuzzyg{\triple{X, A, Y}}{\lambda_{2}}}{\fuzzyg{\triple{X, A, Y}}{\lambda_{1} \oplus \lambda_{2}}}$
    \end{tabular}
  \end{enumerate}
}

\begin{proposition}[Soundness and completeness]
  For an annotated graph, the proof system $\vdash$ is sound and complete for $\models$, that is, (1) if $G \vdash
  \fuzzyg{\tau}{\lambda}$ then $G \models  \fuzzyg{\tau}{\lambda}$ and (2) if $G \models  \fuzzyg{\tau}{\lambda}$ then there is $\lambda' \succeq   \lambda$ with $G \vdash  \fuzzyg{\tau}{\lambda'}$.
\end{proposition}

\nd We point out that rules $2 - 5$ can be represented concisely using the following inference rule:
\begin{center}
  \begin{tabular}{cl}
    $(AG)$ & $\frac{
      \fuzzyg{\tau_{1}}{\lambda_{1}},\ \ldots,\ \fuzzyg{\tau_{n}}{\lambda_{n}, \{\tau_{1}, \ldots \tau_{n}\} \vdash_{\mathsf{RDFS}} \tau}
    }
    {
      \fuzzyg{\tau}{\bigotimes_{i} \lambda_{i}}
    }$
  \end{tabular}
\end{center}
\nd Essentially, this rule says that if a classical RDFS triple $\tau$ can be inferred by applying a classical RDFS
inference rule to triples $\tau_{1}, \ldots \tau_{n}$ (denoted $\{\tau_{1}, \ldots, \tau_{n}\} \vdash_{\mathsf{RDFS}}
\tau$), then the annotation term of $\tau$ will be $\bigotimes_{i} \lambda_{i}$, where $\lambda_{i}$ is the annotation
of triple $\tau_{i}$.  It follows immediately that, using rule $(AG)$, in addition to rules $(1)$ and $(6)$ from
the deductive system above, it is easy to extend these rules to cover complete RDFS. %
Finally, like for the classical case, the \emph{closure} is defined as $cl(G) = \{\fuzzyg{\tau}{\lambda} \mid
G \vdash^{*} \fuzzyg{\tau}{\lambda} \}$, where $\vdash^{*}$ is as $\vdash$ without rule $(1a)$.
Note again that the size of the closure of $G$ is polynomial in $|G|$ and can be computed in
polynomial time, provided that the computational complexity of operations $\otimes$ and $\oplus$ are polynomially bounded
(from a computational complexity point of view, it is as for the classical case, plus the cost of the operations
$\otimes$ and $\oplus$ in $L$). Eventually, similar propositions as Propositions~\ref{pp1} and \ref{pp2} hold.

\begin{example}
  As an example, consider the following triples from
  Figure~\ref{fig:dataset-example}:
  \begin{small}
    \[
    \begin{array}{l}
      \fuzzyg{\triple{\uri{youtubeEmp}, \subclass, \uri{googleEmp}}}{[2006,2011]} \\
      \fuzzyg{\triple{\uri{chadHurley}, \uri{worksFor}, \uri{youtubeEmp}}}{[2005,2010]}
    \end{array}
    \]
  \end{small}
  we infer the following triple:
  \begin{small}
    \[
    \begin{array}{l}
      \fuzzyg{\triple{\uri{chadHurley}, \type, \uri{googleEmp}}}{[2006,2010]}
    \end{array}
    \]
  \end{small}
  
\end{example}

\subsection{Query Answering} \label{aqa} Informally, queries %
are as for the classical case where triples are replaced with annotated triples in which \emph{annotation variables}
(taken from an appropriate alphabet and denoted $\Lambda$) may occur.  We allow built-in triples of the form
$\triple{s,p,o}$, where $p$ is a built-in predicate taken from a reserved vocabulary and having a \emph{fixed
  interpretation} on the annotation domain $D$, such as $\triple{\lambda, \preceq , l}$ stating that the value of
$\lambda$ has to be $\preceq$ than the value $l \in L$. We generalise the built-ins to any $n$-ary predicate $p$, where
$p$'s arguments may be annotation variables, $\rhodf$ variables, domain values of $D$, values from $\AUL$, and $p$ has a
fixed interpretation. We will assume that the evaluation of the predicate can be decided in finite time. As for the crisp case, for
convenience, we write ``functional predicates'' as \emph{assignements} of the form $x\assign f(\vec{z})$ and assume that
the function $f(\vec{z})$ is safe. We also assume that a non functional built-in predicate $p(\vec{z})$ should be safe as well.

For instance, informally for a given time interval $[t_{1}, t_{2}]$, we may define $x\assign length([t_{1}, t_{2}])$ as true \iff
the value of $x$ is $t_{2} - t_{1}$.

\begin{example}\label{exUs4}
  \noindent Considering our dataset from Figure~\ref{fig:dataset-example} as input and the query asking for people that work for Google between 2002 and 2011 and the temporal term at which this was true:
  \begin{align*}
    q(x, \Lambda) \leftarrow \fuzzyg{\triple{x, \uri{worksFor}, \uri{google}}}{\Lambda'},&\\
    \Lambda\assign (\Lambda'& \land [2002,2011])
  \end{align*}
\noindent will get the following answers:
\[
\begin{array}{l}
 \tuple{\uri{steveChen}, [2006, 2011]}\\
 \tuple{\uri{chadHurley}, [2006, 2010]}\\
 \tuple{\uri{jawedKarim}, [2006, 2011]}\\
 \tuple{\uri{larryPage}, [2002, 2011]}\\
 \tuple{\uri{sergeyBrin}, [2002, 2011]}.
 \end{array}
 \]
\end{example}

\noindent Formally, an \emph{annotated query} is of the form
$$q(\vec{x},\vec{\Lambda}) \leftarrow \exists \vec{y}\exists\mathbf{\Lambda}'.\varphi(\vec{x}, \vec{\Lambda},\vec{y},\vec{\Lambda}')$$
in which $\varphi(\vec{x}, \vec{\Lambda},\vec{y},\vec{\Lambda}')$ is a conjunction (as for the crisp case, we use ``,'' as conjunction symbol) of annotated triples and built-in
predicates, $\vec{x}$ and $ \vec{\Lambda}$ are the distinguished variables, $\vec{y}$ and $\vec{\Lambda}'$ are the
vectors of \emph{non-distinguished variables} (existential quantified variables), and $\vec{x}$, $\vec{\Lambda}$,
$\vec{y}$ and $\vec{\Lambda}'$ are pairwise disjoint. Variables in $\vec{\Lambda}$ and $\vec{\Lambda}'$ can only appear
in annotations or built-in predicates. The query head contains at least one variable.

Given an annotated graph $G$, a query $q(\vec{x}, \vec{\Lambda}) \leftarrow \exists
\vec{y}\exists\mathbf{\Lambda}'.\varphi(\vec{x}, \vec{\Lambda}, \vec{y},\vec{\Lambda}')$, a vector $\vec{t}$ of terms in
$uni\-verse(G)$ and a vector $\vec{\lambda}$ of annotated terms in $L$, we say that $q(\vec{t}, \vec{\lambda})$ is
\emph{entailed} by $G$, denoted $G \models q(\vec{t}, \vec{\lambda})$, \iff in any model $\I$ of $G$, there is a vector
$\vec{t}'$ of terms in $ \universe(G)$ and a vector $\vec{\lambda}'$ of annotation values in $L$ such that $\I$ is a
model of $\varphi(\vec{t}, \vec{\lambda}, \vec{t}', \vec{\lambda}')$. If $G \models q(\vec{t}, \vec{\lambda})$ then
$\tuple{\vec{t}, \vec{\lambda}}$ is called an \emph{answer} to $q$. The \emph{answer set} of $q$ \wrt~$G$ is ($\preceq$
extends to vectors point-wise) 
\[
\begin{array}{ll}
ans(G, q) = \{ \tuple{\vec{t}, \vec{\lambda}} \mid G \models  q(\vec{t}, \vec{\lambda}), \vec{\lambda} \neq \vec{\bot} 
\mbox { and  } \\
\hspace*{0.5cm} \mbox { for any }  \vec{\lambda}' \neq \vec{\lambda} 
   \mbox { such that } G \models  q(\vec{t}, \vec{\lambda}'), \vec{\lambda}' \preceq \vec{\lambda}  \mbox { holds} \} \ .
\end{array}
\]
\noindent That is, for any tuple $\vec{t}$, the vector of annotation values $\vec{\lambda}$ is as large as possible. This is to avoid that redundant/subsumed answers occur in the answer set. %
The following can be shown:
\begin{proposition}\label{propU2}
  Given a graph $G$, $\tuple{\vec{t}, \vec{\lambda}}$ is an \emph{answer} to $q$ \iff $\exists
  \vec{y}\exists\mathbf{\Lambda}'.\varphi(\vec{t}, \vec{\lambda}, \vec{y},\vec{\Lambda}')$ is true in the closure of $G$
  and $\lambda$ is $\preceq$-maximal.\footnote{$\exists \vec{y}\exists\mathbf{\Lambda}'.\varphi(\vec{t}, \vec{\lambda},
    \vec{y},\vec{\Lambda}')$ is true in the closure of $G$ \iff for some $\vec{t}'$, $\vec{\lambda}'$ for all triples in
    $\varphi(\vec{t},\vec{\lambda},\vec{t}',\vec{\lambda}')$ there is a triple in $cl(G)$ that subsumes it and the
    built-in predicates are true, where an annotated triple $\fuzzyg{\tau}{\lambda_{1}}$ subsumes
    $\fuzzyg{\tau}{\lambda_{2}}$ \iff $\lambda_{2} \preceq \lambda_{1}$.}
\end{proposition}

\noindent Therefore, we may devise a similar query answering method as for the crisp case by computing the closure,
store it into a database and then using SQL queries with the appropriate support of built-in predicates and domain
operations.

\subsection{Queries with aggregates} \label{sec:aggr}

As next, we extend the query language by allowing so-called aggregates to occur in a query. Essentially, aggregates may
be like the usual SQL aggregate functions such as $\mathsf{SUM}, \mathsf{AVG}, \mathsf{MAX}, \mathsf{MIN}$. But, we have
also domain specific aggregates such as $\oplus$ and $\otimes$.  

The following examples present some queries that can be expressed with the use of built-in queries and aggregates.

\begin{example}
  Using a built-in aggregate we can pose a query that, for each employee, retrieves his maximal time of employment for
  any company in the following way:
  \[
  \begin{array}{lcl}
  q(x, \mathit{maxL}) & \leftarrow & \fuzzyg{\triple{x, \uri{worksFor}, y}}{\lambda},\\
                      &            & \mathit{maxL} \assign \mathit{maxlength}(\lambda)
  \end{array}
  \]
  \nd Here, the $\mathit{maxlength}$ built-in predicate returns, given a set of temporal intervals, the maximal interval
  in the set.
\end{example}

\begin{example}\label{exAA}
  Suppose we are looking for employees that work for some companies for a certain time period. We would like to know the
  average length of their employment. Then such a query will be expressed as
\[
\begin{array}{lcl}
q(x, avgL) & \leftarrow & \fuzzyg{\triple{x, \uri{worksFor}, y}}{\lambda},\\
           &            & \mathsf{GroupedBy}(x),\\
           &            & avgL \assign \mathsf{AVG}[length(\lambda)]
\end{array}
\]
\nd Essentially, we group by the employee, compute for each employee the time he worked for a company by means of the
built-in function $length$, and compute the average value for each group. That is, $g = \{\tuple{t, t_{1}},\ldots,
\tuple{t, t_{n}}\}$ is a group of tuples with the same value $t$ for employee $x$, and value $t_{i}$ for $y$, where each
length of employment for $t_{i}$ is $l_{i}$ (computed as $length(\cdot)$), then the value of $avgL$ for the group $g$ is
$(\sum_{i} l_{i})/n$.
\end{example}

\nd Formally, let $\aggr$ be an aggregate function with $\aggr \in \{\mathsf{SUM}, \mathsf{AVG}, \mathsf{MAX},
\mathsf{MIN}, \mathsf{COUNT}, \oplus, \otimes\}$ then a query with aggregates is of the form
\[
\begin{array}{lcl}
  q(\vec{x}, \vec{\Lambda},\alpha) & \leftarrow & \exists \vec{y}\exists\mathbf{\Lambda}'.\varphi(\vec{x}, \vec{\Lambda}, \vec{y},\vec{\Lambda}'),\\
                                   &            & \mathsf{GroupedBy(\vec{w})},\\
                                   &            &  \alpha \assign\aggr[f(\vec{z})]
\end{array}
\]

\nd where $\vec{w}$ are variables in $\vec{x}$, $\vec{y}$ or $\vec{\Lambda}$ and each variable in $\vec{x}$ and $ \vec{\Lambda}$ occurs in $\vec{w}$ and any variable in $\vec{z}$ occurs in $\vec{y}$ or $\vec{\Lambda'}$.

From a semantics point of view, we say that $\I$  \emph{is a model of} (\emph{satisfies}) $q(\vec{t},\vec{\lambda}, a)$, denoted 
$\I \models q(\vec{t},\vec{\lambda}, a)$ \iff
\[
\begin{array}{l}
a =   \aggr [a_{1}, \ldots, a_{k}]  \mbox{ where }g = \{ \tuple{\vec{t}, \vec{\lambda}, \vec{t}'_{1},\vec{\lambda}'_{1}}, \ldots , \tuple{\vec{t}, \vec{\lambda}, \vec{t}'_{k},\vec{\lambda}_{k}'} \}, \\
\hspace{1.3cm}\mbox{is a group of $k$ tuples with identical projection}\\
\hspace{1.3cm}\mbox{on the variables in } \vec{w}, \varphi(\vec{t}, \vec{\lambda}, \vec{t}'_{r},\vec{\lambda}'_{r}) \mbox{ is true in } \I \\
\hspace{1.3cm}\mbox {and } a_{r} =f(\vec{\vec{t}}) \mbox{ where } \vec{\vec{t}} \mbox{ is the projection of }  \tuple{\vec{t}'_{r}, \vec{\lambda}'_{r}}\\
\hspace{1.3cm}\mbox{on the variables } \vec{z} \ . \\
\end{array}
\]

\nd Now, the notion of  $G \models q(\vec{t},\vec{\lambda}, a)$ is as usual: any model of $G$ is a model of $q(\vec{t},\vec{\lambda}, a)$.

Eventually, we further allow to order answers according to some ordering functions. 

\begin{example} \label{exx}
Consider Example~\ref{exAA}. We additionally would like to order the employee according to the average length of employment. 
Then such a query will be expressed as
\[
\begin{array}{lcl}
q(x,avgL) & \leftarrow & \fuzzyg{\triple{x, \uri{worksFor}, y}}{\lambda},\\
          &            & \mathsf{GroupedBy}(x),\\
          &            & avgL \assign \mathsf{AVG}[length(\lambda)],\\
          &            & \mathsf{OrderBy}(avgL)
\end{array}
\]

\end{example}

\nd Formally, a query with ordering is of the form
\[
\begin{array}{lcl}
q(\vec{x}, \vec{\Lambda}, z) & \leftarrow & \exists \vec{y}\exists\mathbf{\Lambda}'.\varphi(\vec{x}, \vec{\Lambda}, \vec{y},\vec{\Lambda}'), \mathsf{OrderBy}(z)
\end{array}
\]

\nd or, in case grouping is allowed as well, it is of the form
\[
\begin{array}{lcl}
  q(\vec{x}, \vec{\Lambda},z, \alpha) & \leftarrow & \exists \vec{y}\exists\mathbf{\Lambda}'.\varphi(\vec{x}, \vec{\Lambda}, \vec{y},\vec{\Lambda}'),\\
                                      &            & \mathsf{GroupedBy(\vec{w})},\\
                                      &            & \alpha \assign\aggr[f(\vec{z})],\\
                                      &            & \mathsf{OrderBy}(z)
\end{array}
\]

\nd From a semantics point of view, the notion of $G \models q(\vec{t},\vec{\lambda}, z, a)$ is as before, but the
notion of answer set has to be enforced with the fact that the answers are now ordered according to the assignment to
the variable $z$. Of course, we require that the set of values over which $z$ ranges can be ordered (like string,
integers, reals). In case the variable $z$ is an annotation variable, the order is induced by $\preceq$. In case,
$\preceq$ is a partial order then we may use some linearisation method for posets, such as~\cite{Labrador10}.
Finally, note that the additional of the SQL-like statement $\mathsf{LIMIT}(k)$ can be added straightforwardly.
\section{\aSPARQL: Annotated SPARQL}
\label{sec:annotated-sparql}

Our introduced query language so far allows for conjunctive queries. Languages like SQL and SPARQL allow to pose more
complex queries including built-in predicates to filter solutions, advanced features such as negation or aggregates.
In this section we will present an extension of the SPARQL~\cite{sparql} query language, called \aSPARQL, that enables
querying annotated graphs.  We will begin by presenting some preliminaries on SPARQL.

\subsection{SPARQL}
\label{sec:class-SPARQL}

SPARQL~\cite{sparql} is the W3C recommended query language for RDF. A \emph{SPARQL query} is defined by a triple
$Q=(P,G,V)$, where $P$ is a \emph{graph pattern} and the \emph{dataset} $G$ is an RDF graph and $V$ is the \emph{result
  form}.  We will restrict ourselves to \textsf{SELECT} queries in this work so it is sufficient to consider the result
form $V$ as a list of variables.
\begin{remark}
  \label{fn:nonamedgraphs}
  Note that, for presentation purposes, we simplify the notion of datasets by excluding named graphs and thus GRAPH
  queries. Our definitions can be straightforwardly extended to named graphs and we refer the reader to the SPARQL W3C
  specification~\cite{sparql} for details.
\end{remark}
We base our semantics of SPARQL on the semantics presented by P\'erez \etal~\cite{DBLP:journals/tods/PerezAG09},
extending the multiset semantics to lists, which are considered a multiset with ``default'' ordering.
RDF triples, possibly with variables in subject, predicate or object positions, are called \emph{triple patterns}. In
the basic case, graph patterns are sets of triple patterns, also called \emph{basic graph patterns} (BGP).  Let $\AU$,
$\AB$, $\AL$ be defined as before and let $\AV$ denote a set of variables, disjoint from $\AUBL$.  We further
denote by $\mathit{var}(P)$ the set of variables present in a graph pattern $P$.
\begin{definition}[Solution~{\cite[Section\ 12.3.1]{sparql}}]
  Given a graph $G$ and a BGP $P$, a \emph{solution} $\theta$ for $P$ over $G$ is a mapping over a subset $V$ of
  $\mathit{var}(P)$, \ie $\theta: V \to \mathit{term}(G)$ \st $G \models P\theta$ where $P\theta$ represents the
  triples obtained by replacing the variables in graph pattern $P$ according to $\theta$, and where $G \models P\theta$ means that any triple in $P\theta$ is entailed by $G$.  We call $V$ the
  \emph{domain} of $\theta$, denoted by $\mathit{dom}(\theta)$. For convenience, sometimes we will use the notation $\theta = \{x_{1}/t_{1}, \ldots, x_{n}/t_{n}\}$ to indicate that $\theta(x_{i}) = t_{i}$, \ie variable $x_{i}$ is assigned to term $t_{i}$.
\end{definition}
Two mappings $\theta_1$ and $\theta_2$ are considered \emph{compatible} if for all $x \in \mathit{dom}(\theta_1) \cap
\mathit{dom}(\theta_2), \theta_1(x) = \theta_2(x)$.
We call the evaluation of a BGP $P$ over a graph $G$, denoted $\eval{P}_{G}$, the set of solutions.  
\begin{remark}\label{rem:no-vars}
  Note that variables in the domain of $\theta$ play the role of distinguished variables in conjunctive queries and
  there are no non-distinguished variables.
\end{remark}

\nd The notion of solution for BGPs is the same as the notion of answers for conjunctive queries:

\begin{proposition} \label{pp}
Given a graph $G$ and a BGP $P$, then the solutions of $P$ are the same as the answers of the query
$q(var(P)) \leftarrow P$ (where $var(P)$ is the vector of variables in $P$), \ie~$ans(G,q) = \eval{P}_{G}$.
\end{proposition}

\nd We present the syntax of SPARQL based on~\cite{DBLP:journals/tods/PerezAG09} and present \emph{graph patterns}
similarly.  A \emph{triple pattern}~$\triple{s,p,o}$ is a graph pattern where $s, o \in \AULV$ and $p \in
\AUV$.\footnote{We do not consider blank nodes in triple patterns since they can be considered as variables.}  Sets of
triple patterns are called \emph{Basic Graph Patterns (BGP)}.  A generic \emph{graph pattern} is defined in a recursive
manner: any BGP is a graph pattern; if $P$ and $P'$ are graph patterns, $R$ is a filter expression~(see~\cite{sparql}),
then $(P\ \mathsf{AND}\ P')$, $(P\ \mathsf{OPTIONAL}\ P')$, $(P\ \mathsf{UNION}\ P')$, $(P\ \mathsf{FILTER}\ R)$ are
graph patterns.  As noted in Remark~\ref{fn:nonamedgraphs} we do not consider $\mathsf{GRAPH}$ patterns.

Evaluations of more complex patterns including FILTERs, OPTIONAL patterns, AND patterns, UNION patterns, etc. are
defined by an algebra that is built on top of this \emph{basic graph pattern
  matching}~(see~\cite{sparql,DBLP:journals/tods/PerezAG09}).

\begin{definition}[SPARQL Relational Algebra]
 \label{def:sparql-relalg}
 Let $\Omega_1$ and $\Omega_2$ be sets of mappings:
 \begin{tabbing}
     $\Omega_1 \louter \Omega_2$\ \=$=$\ \=\kill
     $\Omega_1 \bowtie \Omega_2$\>$=$\ \>$\{\theta_1 \cup \theta_2 \mid \theta_1 \in \Omega_1, \theta_2 \in \Omega_2,
     \theta_1\textrm{ and }\theta_2\textrm{ compatible}\}$\\
     $\Omega_1 \doublecup \Omega_2$\>$=$\>$\{\theta \mid \theta \in \Omega_1\textrm{ or } \theta \in \Omega_2\}$\\
     $\Omega_1 - \Omega_2$\>$=$\>$\{\theta_1 \in \Omega_1 \mid \textrm{ for all }\theta_2 \in \Omega_2,\
     \theta_1\textrm{ and }\theta_2\textrm{ not compatible} \}$\\
     $\Omega_1 \louter \Omega_2$\>$=$\>$(\Omega_1 \bowtie \Omega_2) \doublecup (\Omega_1 - \Omega_2)$
   \end{tabbing}
\end{definition}

\begin{definition}[Evaluation {\cite[Definition 2.2]{DBLP:journals/tods/PerezAG09}}]
\label{def:semantics-sparql}
Let $\tau = \triple{s,p,o}$ be a triple pattern, $P, P_1, P_2$ graph patterns and $G$ an RDF graph, then the evaluation
$\eval{\cdot}_{G}$ is recursively defined as follows: {
\begin{tabbing}
$\eval{t}_{G}\ =\ \{\theta \mid \mathit{dom}(\theta)=\mathit{var}(P)\textrm{ and }G \models \tau\theta\}$\\
$\eval{P_1\ \mathsf{FILTER}\ P_2}_{G}\quad\quad$\=$=$\ \=$\eval{P_1}_{G}$\ \=$\louter$\ \=$\eval{P_2}_{G}$\kill
$\eval{P_1\ \mathsf{AND}\ P_2}_{G}$\>$=$\>$\eval{P_1}_{G}$\>$\ \bowtie$\>$\eval{P_2}_{G}$\\
$\eval{P_1\ \mathsf{UNION}\ P_2}_{G}$\>$=$\>$\eval{P_1}_{G}$\>$\ \doublecup$\>$\eval{P_2}_{G}$\\
$\eval{P_1\ \mathsf{OPTIONAL}\ P_2}_{G}$\>$=$\>$\eval{P_1}_{G}$\>$\louter$\>$\eval{P_2}_{G}$\\
$\eval{P\ \mathsf{FILTER}\ R}_{G}$\>$=$\>$\{ \theta \in \eval{P}_{G} \mid R\theta \mbox{ is true } \}$
\end{tabbing}}

Let $R$ be a $\mathsf{FILTER}$\footnote{For simplicity, we will omit from the presentation $\mathsf{FILTER}{s}$ such as
  comparison operators (`$<$', `$>$',`$\leq$',`$\geq$'), data type conversion and string functions and refer the reader
  to~\cite[Section 11.3]{sparql} for details.} expression, $u,v \in \AV \cup \AUBL$. The valuation of $R$ on a
substitution $\theta$, written $R\theta$, is \emph{true}
if:%
{
\begin{tabbing}
(1) $R = \mathsf{BOUND}(v)$ with $v \in \mathit{dom}(\theta)$;\\
(2) $R = \mathsf{isBLANK}(v)$ with $v \in \mathit{dom}(\theta)$ and $\theta(v) \in \AB$;\\
(3) $R = \mathsf{isIRI}(v)$ with $v \in \mathit{dom}(\theta)$ and $\theta(v) \in \AU$;\\
(4) $R = \mathsf{isLITERAL}(v)$ with $v \in \mathit{dom}(\theta)$ and $\theta(v) \in \AL$;\\
(5) $R = (u = v)$ with $u,v \in \mathit{dom}(\theta)\cup \AUBL \wedge\ \theta(u) = \theta(v)$;\\
(6) $R = (\neg R_1)$ with $R_1\theta\mbox{ is false}$;\\ 
(7) $R = (R_1 \vee R_2 )$ with $R_1\theta\mbox{ is true}$ or $R_2\theta\mbox{ is true}$;\\
(8) $R = (R_1 \wedge R_2)$ with $R_1\theta\mbox{ is true}$ and $R_2\theta\mbox{ is true}$.
\end{tabbing}}

\noindent$R\theta$ yields an error (denoted $\varepsilon$), if: {
\begin{enumerate}[(1)]
\setlength{\itemsep}{-1.5mm}
\item $R = \mathsf{isBLANK}(v)$,$R = \mathsf{isIRI}(v)$, or $R = \mathsf{isLITERAL}(v) $ and $v \not\in \mathit{dom}(\theta)\cup T$;
\item $R = (u = v)$ with $u \not\in \mathit{dom}(\theta)\cup T$ or $v\not\in \mathit{dom}(\theta)\cup T$;
\item $R = (\neg R_1)$ and $R_1\theta=\varepsilon$;
\item $R = (R_1 \vee R_2 )$ and $(R_1\theta\not=\top$ and $R_2\theta\not=\top)$ and $(R_1\theta=\varepsilon$ or $R_2\theta=\varepsilon)$;
\item $R = (R1 \wedge R2)$ and $R_1\theta=\varepsilon$ or $R_2\theta=\varepsilon$.
\end{enumerate}}
\noindent Otherwise $R\theta$ is \emph{false}.
\end{definition}

\nd In order to make the presented semantics compliant with the SPARQL specification~\cite{sparql}, we need to introduce an
extension to consider unsafe $\mathsf{FILTER}{s}$ (also presented in~\cite{DBLP:conf/semweb/AnglesG08}):
\begin{definition}[\textsf{OPTIONAL} with \textsf{FILTER} Evaluation]
\label{def:optsemantics}
Let $P_1, P_2$ be graph patterns $R$ a \textsf{FILTER} expression. A mapping $\theta$ is in $\eval{P_1\ \textsf{OPTIONAL}\
(P_2\ \textsf{FILTER}\ R)}_{DS}$ \iff:
\begin{itemize} 
\item $\theta = \theta_1\cup \theta_2$, s.t. $\theta_1 \in \eval{P_1}_{G}$, $\theta_2 \in \eval{P_2 }_{G}$ are compatible and
  $R\theta$ is true, or
\item $\theta \in \eval{P_1}_{G}$ and $\forall \theta_2 \in \eval{P_2}_{G}$,  $\theta$ and $\theta_2$ are not compatible, or
\item $\theta \in \eval{P_1}_{G}$ and $\forall \theta_2 \in \eval{P_2}_{G}$ s.t. $\theta$ and $\theta_2$ are compatible, and
  $R\theta_{3}$ is false for $\theta_{3} = \theta \cup \theta_2$.
\end{itemize}
\end{definition}

\subsection{\aSPARQL}
\label{sec:annotated-sparql2}

We are now ready to extend SPARQL for querying annotated RDF. We call the novel query language \aSPARQL.  For the rest
of this Section we fix a specific annotation domain, $D = \tuple{L, \oplus, \otimes, \bot, \top}$, as defined in
Section~\ref{sec:rdfs-annot-doma}.

\subsubsection{Syntax}
\label{sec:aSPARQL-syntax}

We take inspiration on the notion of conjunctive annotated queries discussed in Section~\ref{aqa}.
A \emph{simple \aSPARQL query} is defined -- analogously to a SPARQL query -- as a triple $Q=(P,G,V, A)$ with the differences
that \begin{inparaenum}[(1)]
\item $G$ is an annotated RDF graph;
\item we allow annotated graph patterns as presented in Definition~\ref{def:pattern} and
\item $A$ is the set of annotation variables taken from an infinite set $\AA$ (distinct from $\AV$).
\end{inparaenum}
We further denote by $\mathit{avar}(P)$ the set of annotation variables present in a graph pattern $P$.

\begin{definition}[Annotated Graph Pattern] 
  \label{def:pattern}
  Let $\lambda$ be an annotation value from $L$ or an annotation variable from $\AA$.  We call $\lambda$ an
  \emph{annotation label}.
  Triple patterns in annotated \aSPARQL are defined the same way as in SPARQL.  For a triple pattern $\tau$, we call
  $\fuzzyg{\tau}{\lambda}$ an \emph{annotated triple pattern} and sets of annotated triple patterns are called \emph{basic
    annotated patterns} (BAP).
  A generic \emph{annotated graph pattern} is defined in a recursive manner: any BAP is an annotated graph pattern; if $P$ and $P'$ are annotated graph
  patterns, $R$ is a filter expression (see~\cite{sparql}), then $(P\ \mathsf{AND}\ P')$,
  $(P\ \mathsf{OPTIONAL}\ P')$, $(P\ \mathsf{UNION}\ P')$, $(P\ \mathsf{FILTER}\ R)$
  are annotated graph patterns.
\end{definition}

\begin{example} \label{exx1} Suppose we are looking for Ebay employees during some time period and that optionally
  owned a car during that period.  This query can be posed as follows:
\begin{Verbatim}
SELECT ?p ?l ?c WHERE {
   (?p type ebayEmp):?l
   OPTIONAL{(?p hasCar ?c):?l}
}
\end{Verbatim}
Assuming our example dataset from Figure~\ref{fig:dataset-example} extended with the following triples:
{\small
  \[
  \begin{array}{l}
    \fuzzyg{\triple{\uri{toivo}, \type, \uri{paypalEmp}}}{[2000,2009]} \\
    \fuzzyg{\triple{\uri{toivo}, \uri{hasCar}, \uri{peugeot}}}{[1999,2005]}\\
    \fuzzyg{\triple{\uri{toivo}, \uri{hasCar}, \uri{renault}}}{[2005,2010]}\\
  \end{array}
  \]
}
we will get the following answers:
{\small
\[
\begin{array}{lcl}
  \theta_{1} & = & \{?p/\uri{toivo}, ?l/[2002, 2009]\} \\
  \theta_{2} & = & \{?p/\uri{toivo}, ?l/[2002, 2005], ?c/\uri{peugeot} \} \\
  \theta_{3} & = &  \{?p/\uri{toivo}, ?l/[2005, 2009], ?c/\uri{renault} \} \  .
\end{array}
\]
}
The first answer corresponds to the answer in which the $\mathsf{OPTIONAL}$ pattern is not satisfied, so we get the
annotation value $[2002, 2009]$ that corresponds to the time \uri{toivo} is an Ebay employee.  In the second and third answers,
the $\mathsf{OPTIONAL}$ pattern is also matched and, in this case, the annotation value is restricted to the time when
Toivo is employed by Paypal and has a car.  %
\end{example}
Note that -- as we will see -- this first query will return as a result for the annotation variable the periods where a
car was owned. %

\begin{example} \label{exx2} A slightly different query can be the employees of Ebay during some time period and
  optionally owned a car at some point during their stay.  This query -- which will rather return the time periods of
  employment -- can be written as follows:
\begin{Verbatim}[commandchars=+\[\]]
SELECT ?p ?l ?c WHERE {
   (?p type ebayEmp):?l 
   OPTIONAL {(?p hasCar ?c):?l2
   FILTER (?l2 +(+preceq+) ?l)}
}
\end{Verbatim}
Using the input data from  Example~\ref{exx1}, we obtain the following answers:
    \[
    \begin{array}{lcl}
  \theta_{1} & = &  \{ ?p/\uri{toivo}, ?l/[2002, 2009]\} \\
  \theta_{2} & = &  \{ ?p/\uri{toivo}, ?l/[2002, 2009], ?c/renault \}
    \end{array}
    \]
    In this example the $\mathsf{FILTER}$ behaves as in SPARQL by removing from the answer set the mappings that do not
    make the $\mathsf{FILTER}$ expression true. 
\end{example}

\nd This query also exposes the issue of unsafe filters, noted in~\cite{DBLP:conf/semweb/AnglesG08} and we presented the
semantics to deal with this issue in Definition~\ref{def:optsemantics}.

\subsubsection{Semantics}
\label{sec:aSPARQL-semantics}

We are thus ready to define the semantics of \aSPARQL queries by extending the notion of SPARQL BGP matching.  As for
the SPARQL query language, we are going to define the notion of solutions for BAP as the equivalent notion of answers
set of annotated conjunctive queries.  Just as matching BGPs against RDF graphs is at the core of SPARQL semantics,
matching BAPs against annotated RDF graphs is the heart of the evaluation semantics of \aSPARQL.

We extend the notion of \emph{substitution} to include a substitution of annotation variables in which we do not allow
any assignment of an annotation variable to $\bot$ (of the domain $D$).  An annotation value of $\bot$, although it is a
valid answer for any triple, does not provide any additional information and thus is of minor interest. Furthermore this
would contribute to increasing the number of answers unnecessarily.

\begin{definition}[BAP evaluation]
  \label{def:btpmatch}
  Let $P$ be a BAP and $G$ an annotated RDF graph. We define \emph{evaluation} $\eval{P}_G$ as the list %
  of substitutions that are \emph{solutions} of $P$, \ie $\eval{P}_G = \{\theta \mid G \models \theta(P) \}$, and where $G
  \models \theta(P)$ means that any annotated triple in $\theta(P)$ is entailed by $G$.
\end{definition}

\nd As for SPARQL, we have:
\begin{proposition} \label{ppp}
Given an annotated graph $G$ and a BAP $P$, the solutions of $P$ are the same as the answers of the annotated query
$q(var(P)) \leftarrow P$ (where $var(P)$ is the vector of variables in $P$), \ie~$ans(G,q) = \eval{P}_{G}$.
\end{proposition}

\nd For the extension of the SPARQL relational algebra to the annotated case we introduce -- inspired by the definitions in~\cite{DBLP:journals/tods/PerezAG09} -- definitions of compatibility and union of substitutions: 

\begin{definition}[$\otimes$-compatibility] 
  Two substitutions $\theta_1$ and $\theta_2$ are \emph{$\otimes$-compatible} \iff \ii{i} $\theta_1$ and $\theta_2$ are
  compatible for all the non-annotation variables, \ie $\theta_1(x) = \theta_2(x)$ for any non-annotation variable $x\in
  \mathit{dom}(\theta_1) \cap \mathit{dom}(\theta_2)$; and \ii{ii} $\theta_1(\lambda) \otimes \theta_2(\lambda) \neq \bot$ for any
  annotation variable $\lambda \in \mathit{dom}(\theta_1) \cap \mathit{dom}(\theta_2)$.
\end{definition}

\begin{definition}[$\otimes$-union of substitutions]
  Given two $\otimes$-compatible substitutions $\theta_1$ and $\theta_2$, the \emph{$\otimes$-union} of $\theta_1$ and
  $\theta_2$, denoted $\theta_1 \otimes \theta_2$, is as $\theta_{1} \cup \theta_{2}$, with the exception that any
  annotation variable $\lambda \in \mathit{dom}(\theta_1) \cap \mathit{dom}(\theta_2)$ is mapped to $\theta_1(\lambda) \otimes
  \theta_2(\lambda)$.
\end{definition}

We now present the notion of evaluation for generic \aSPARQL graph patterns.  This consists of an extension of
Definition~\ref{def:semantics-sparql}:

\def\killedline{$\eval{P_1\ \mathsf{FILTER}\ P_2}\quad\quad\quad\quad$\=$=$\ \ \ \ \=$\eval{P_1}$\ \ \ \=$\louter$\ \=$\eval{P_2}$\kill}

\begin{definition}[Evaluation, extends {\cite[Definition 2]{DBLP:journals/tods/PerezAG09}}]
\label{def:semantics-anql}
Let $P$ be a BAP, $P_{1}, P_{2}$ annotated graph patterns, $G$ an annotated graph and $R$ a filter expression, then the
evaluation $\eval{\cdot}_{G}$, \ie~set of \emph{answers},\footnote{\label{fn:answersets} Strictly speaking, we consider
  \emph{sequences} of answers -- note that SPARQL allows duplicates and imposes an order on solutions,
  cf. Section~\ref{sec:asparql-aggregates} below for more discussion -- but we stick with set \emph{notation}
  representation here for illustration. Whenever we mean ``real'' sets where duplicates are removed we write
  $\{\ldots\}_{\mathsf{DISTINCT}}$.} is recursively defined as:
\begin{center}
{\small
\begin{itemize}[]
\item $\eval{P}_{G}$  =
	$\left\{\theta \mid \mathit{dom}(\theta)=\mathit{var}(P)\textrm{ and }G\models \theta(P)\right\}$
\item $\eval{P_1\ \mathsf{AND}\ P_2}_{G}$ = 
	$\left\{\theta_1\otimes\theta_2 \mid \theta_1 \in \eval{P_1}_{G}, \theta_2 \in \eval{P_2}_{G}, \theta_1\right.$ and $\left.\theta_2\textrm{ $\otimes$-compatible}\right\}$
\item $\eval{P_1\ \mathsf{UNION}\ P_2}_{G}$ = 
	$\eval{P_1}_{G} \cup \eval{P_2}_{G}$
\item $\eval{P_1\ \mathsf{FILTER}\ R}_{G}$ =  
	$\left\{\theta \mid \theta \in \eval{P_1}_{G} \mbox{ and } R\theta \textrm{ is true} \right\}$
\item $\eval{P_1\ \mathsf{OPTIONAL}\ P_2[R]}_{G}$ =\\
	\hspace{4mm}$\{ \theta \mid$ and $\theta$ meets one of the following conditions:
{\begin{enumerate}
	\setlength{\itemsep}{-1.5mm}
	\item $\theta = \theta_1 \otimes \theta_2$ if $\theta_1 \in \eval{P_1}_{G}, \theta_2 \in \eval{P_2}_{G}, \theta_1$ and $\theta_2 \otimes$-compatible, and $R\theta$ is true;
	\item $\theta = \theta_1 \in \eval{P_1}_{G}$ and $\forall\theta_2 \in \eval{P_2}_{G}$ such that $\theta_1$ and $\theta_2$ $\otimes$-compatible, $R(\theta_1\otimes\theta_2)$ is true, and for all annotation variables $\lambda\in \mathit{dom}(\theta_1) \cap \mathit{dom}(\theta_2)$, $\theta_2(\lambda) \prec \theta_1(\lambda)$;
	\item $\theta = \theta_1 \in \eval{P_1}_{G}$ and $\forall\theta_2 \in \eval{P_2}_{G}$ such that $\theta_1$ and $\theta_2$ $\otimes$-compatible, $R(\theta_1\otimes\theta_2)$ is false $\}$
	\end{enumerate}}
\end{itemize}}
\end{center}
Let $R$ be a $\mathsf{FILTER}$ expression and $x,y \in \AA \cup L$, in addition to the $\mathsf{FILTER}$ expressions presented in
Definition~\ref{def:semantics-sparql} we further allow the expressions presented next.  The valuation of $R$ on a
substitution $\theta$, denoted $R\theta$ is true if:\footnote{We consider
  a simple evaluation of filter expressions where the ``error'' result is ignored, see~\cite[Section 11.3]{sparql} for details.}

\nd(9) $R = (x \preceq y)$ with $x,y \in \mathit{dom}(\theta)\cup L \wedge\ \theta(x) \preceq \theta(y)$;\\
(10) $R = p(\vec{z})$ with $p(\vec{z})\theta=\mbox{ true \iff }\ p(\theta(\vec{z})) =\textrm{ true}$, where $p$ is a built-in predicate.

\noindent Otherwise $R\theta$ is false.
\end{definition}

\nd In the $\mathsf{FILTER}$ expressions above, a built-in predicate $p$ is any $n$-ary predicate $p$, where $p$'s
arguments may be variables (annotation and non-annotation ones), domain values of $D$, values from $\AUL$, $p$ has a
fixed interpretation and we assume that the evaluation of the predicate can be decided in finite time.
Annotation domains may define their own built-in predicates that range over annotation values as in the
following query:
\begin{example}
  Consider our example dataset from Figure~\ref{fig:dataset-example} and that we want to know where \uri{chadHurley} was
  working before 2005. This query can be expressed in the following way:
\begin{Verbatim}
SELECT ?city WHERE {
    (chadHurley worksFor ?comp):?l
    FILTER(before(?l, [2005]))
}
\end{Verbatim}
\end{example}

\begin{remark} \label{remsubs}
\nd For practical convenience, we retain in $\eval{\cdot}_{G}$ only ``domain maximal answers''. That is, let us define $\theta' \preceq \theta$ \iff 
\ii{i} $\theta' \neq \theta$; \ii{ii} $\mathit{dom}(\theta) = \mathit{dom}(\theta')$;  \ii{iii} $\theta(x) = \theta'(x)$ for any non-annotation variable $x$; and \ii{iv} $\theta'(\lambda) \preceq \theta(\lambda)$ for any annotation variable $\lambda$. Then, for any $\theta \in  \eval{P}_{G}$ we remove any $\theta'  \in  \eval{P}_{G}$ such that $\theta' \preceq \theta$.
 \end{remark}
 
\begin{remark}
  Please note that the cases for the evaluation of the $\mathsf{OPTIONAL}$ are compliant with the SPAR\-QL
  specification~\cite{sparql}, covering the notion of unsafe $\mathsf{FILTER}{s}$ as presented
  in~\cite{DBLP:conf/semweb/AnglesG08}.  However, there are some peculiarities inherent to the annotated case. More
  specifically case 2.) introduces the side effect that annotation variables that are compatible between the mappings
  may have different values in the answer depending if the $\mathsf{OPTIONAL}$ is matched or not.  This is the behaviour
  demonstrated in Example~\ref{exx1}.
\end{remark}

\nd The following proposition shows that we have a conservative extension of SPARQL:
\begin{proposition} \label{prop:saparql2sparql}
  Let $Q = (P,G,V)$ be a SPARQL query over an RDF graph $G$. Let $G'$ be obtained from $G$ by annotating triples  with $\top$. Then $\eval{P}_G$ under SPARQL semantics is in one-to-one correspondence to $\eval{P}_{G'}$ under \aSPARQL semantics such that for any $\theta \in \eval{P}_G$ there is a   $\theta' \in \eval{P}_{G'}$ with $\theta$ and $\theta'$ coinciding on $var(P)$.
\end{proposition}

\subsubsection{Further Extensions of \aSPARQL}
\label{sec:asparql-aggregates}

In this section we will present extensions of Definition~\ref{def:semantics-anql} to include variable assignments,
aggregates and solution modifiers.  These are extensions similar to the ones presented in Section~\ref{sec:aggr}.

\begin{definition}
  \label{def:anql-assign}
  Let $P$ be an annotated graph pattern and $G$ an annotated graph, the evaluation of an $\mathsf{ASSIGN}$ statement is
  defined as:
  \begin{align*}
    \eval{P\ \mathsf{ASSIGN}\ f(\vec{z})\ AS\ z}_{G} = \{\theta \mid \theta_1 \in \eval{P}_{G},&\\
    \theta =\theta_1&[z/f(\theta_{1}(\vec{z}))]\}
  \end{align*}
\nd where 
\[
\theta[z/t] = \left\{
  \begin{array}{ll}
    \theta \cup \{z/t\} &\textrm{ if } z \not\in \mathit{dom}(\theta)\\
    \left(\theta \setminus \{z/t'\}\right) \cup \{z/t\} &\textrm{ otherwise} \ .
  \end{array}
  \right.
\]

\end{definition}

\nd Essentially, we assign to the variable $z$ the value $f(\theta_{1}(\vec{z}))$, which is the evaluation of the function $f(\vec{z})$ with respect to a substitution $\theta_{1} \in \eval{P}_{G}$.

\begin{example} \label{exx3} Using a built-in function we can retrieve for each employee the length of employment for
  any company:
\begin{Verbatim}
SELECT ?x ?y ?z WHERE {
    (?x worksFor ?y):?l
    ASSIGN length(?l) AS ?z
}
\end{Verbatim}
\nd Here, the $\mathit{length}$ built-in predicate returns, given a set of temporal intervals, the overall total length
of the intervals.  %
\end{example}

\begin{remark}
  Note that this definition is more general than ``\texttt{SELECT $expr$ AS $?var$}'' project expressions in current
  SPARQL~1.1~\cite{sparql11} due to not requiring that the assigned variable be unbound.
\end{remark}

\nd We introduce the $\mathsf{ORDERBY}$ clause where the evaluation of a $\eval{P\ \mathsf{ORDERBY}\ ?x}_{G}$ statement
is defined as %
the ordering of the solutions -- for any $\theta \in \eval{P}_G$ -- according to the values of $\theta(?x)$. Ordering
for non-annotation variables follows the rules in~\cite[Section 9.1]{sparql}. %

Similarly to ordering in the query answering setting, we require that the set of values over which $x$ ranges can be
ordered and some linearisation method for posets may be applied if necessary, such as~\cite{Labrador10}.
We can further extend the evaluation of \aSPARQL queries with aggregate functions 
\[
\aggr \in \{\mathsf{SUM},\mathsf{AVG}, \mathsf{MAX}, \mathsf{MIN}, \mathsf{COUNT}, \oplus, \otimes\}
\]
\nd as follows:

\begin{definition}
  \label{def:groupby}
  The evaluation of a $\mathsf{GROUPBY}$ statement is defined as:\footnote{In the expression,
    $\vec{\aggr}\vec{f}(\vec{z})\ AS\ \vec{\alpha}$ is a concise representation of $n$ aggregations of the form
    $\aggr_{i} f_{i}(\vec{z}_{i})\ AS\ \alpha_{i}$.}  {\small
    \begin{align*}
      \eval{P\ \mathsf{GROUPBY}(\vec{w})\ \vec{\aggr}\vec{f}(\vec{z})\ AS\ \vec{\alpha}}_{G} = \{ \theta \mid
      \theta_1\mbox{ in }\eval{P}_{G},&\\ 
      \theta = \theta_{1}|_{\vec{w}}[\alpha_i/\aggr_{i}{}f_i(\theta_i(\vec{z}_{i}))]&
      \}_{\mathsf{DISTINCT}}
    \end{align*}
} \nd where the variables $\alpha_i \not\in var(P)$, $\vec{z}_i \in var(P)$ and none of the $\mathsf{GROUPBY}$ variables
$\vec{w}$ are included in the aggregation function variables $\vec{z}_i$.  Here, we denote by $\theta|_{\vec{w}}$ the
restriction of variables in $\theta$ to variables in $\vec{w}$. Using this notation, we can also straightforwardly
introduce projection, \ie sub-SELECTs as an algebraic operator in the language covering another new feature of
SPARQL~1.1: {\small
\begin{eqnarray*}
\eval{\mathsf{SELECT}\ \vec{V}\ \{ P \} }_{G} & = & \{ \theta \mid \theta_1\mbox{ in }\eval{P}_{G}, \theta = \theta_{1}|_{\vec{v}} \} \ .
\end{eqnarray*}
}

\end{definition}
\begin{remark}
  Please note that the aggregator functions have a domain of definition and thus can only be applied to values of their
  respective domain.  For example, $\mathsf{SUM}$ and $\mathsf{AVG}$ can only be used on numeric values, while
  $\mathsf{MAX}, \mathsf{MIN}$ are applicable to any total
  order. Resolution of type mismatches for aggregates is currently
  being defined in SPARQL~1.1~\cite{sparql11} and we aim to follow those, as soon
  as the language is stable. The
  $\mathsf{COUNT}$ aggregator can be used for any finite set of
  values. The last two aggregation functions, 
  namely $\oplus$ and $\otimes$, are defined by the annotation domain and thus can be used on any annotation variable. 
\end{remark}

\begin{remark}
  Please note that, unlike the current SPARQL 1.1 syntax, assignment, solution
  modifiers (ORDER BY, LIMIT) and aggregation are stand-alone
  operators in our language and do not need to be tied to a
  sub-SELECT but can occur nested within any pattern. This may be viewed
  as syntactic sugar allowing for more concise writing than the
  current SPARQL~1.1~\cite{sparql11} draft.
\end{remark}

\begin{example}\label{ex:AA}
  Suppose we want to know, for each employee, the average length of
  their employments with different employers. Then such a query will be
  expressed as:
\begin{Verbatim}
SELECT ?x ?avgL WHERE {
    (?x worksFor ?y):?l
    GROUPBY(?x)
    AVG(length(?l)) AS ?avgL
}
\end{Verbatim} 

\nd Essentially, we group by the employee, compute for each employee the time he worked for a company by means of the
built-in function $length$, and compute the average value for each group. That is, if $g = \{\tuple{t, t_{1}},\ldots,
\tuple{t, t_{n}}\}$ is a group of tuples with the same value $t$ for employee $x$, and value $t_{i}$ for $y$, where each
length of employment for $t_{i}$ is $l_{i}$ (computed as $length(\cdot)$), then the value of $avgL$ for the group $g$ is
$(\sum_{i} l_{i})/n$.  %
\end{example}

\begin{proposition}
  Assuming the built-in predicates are computable in finite time, the answer set of any \aSPARQL is finite and can also
  be computed in finite time.
\end{proposition}

\nd This proposition can be demonstrated by induction over all the constructs we allow in \aSPARQL.

\subsection{Constraints vs Filters}\label{sec:constr-vs-filt}

Please note that $\mathsf{FILTER}{s}$ do not act as \emph{constraints} over the query.  Given the data from our dataset
example and for the following query:
\begin{Verbatim}
SELECT ?l1 ?l2 WHERE {
    (?p type youtubeEmp):?l1 . 
    (steveChen type youtubeEmp):?l2
}
\end{Verbatim}
with an additional \emph{constraint} that requires $?l1$ to be ``before'' $?l2$,  we could expect the answer 
$$\{?l1/[2005, 2010], ?l2/[2011, 2011]\}.$$ This answer matches the following triples of our dataset:
\begin{small}
\[
\begin{array}{l}
\fuzzyg{\triple{\uri{steveChen}, \type, \uri{youtubeEmp}}}{[2005,2011]} \\
\fuzzyg{\triple{\uri{chadHurley}, \type, \uri{youtubeEmp}}}{[2005,2010]} \\
\end{array}
\]
\end{small}
\nd and satisfies the proposed \emph{constraint}.  However, we require maximality of the annotation values in the answers,
which in general, do not exist in presence of \emph{constraints}. For this reason, we do not allow general
\emph{constraints}.

\subsection{Union of annotations}\label{sec:union-annotations}

The SPARQL $\mathsf{UNION}$ operator may also introduce some discussion when considering shared annotations between
graph patterns. Take for example the following query:
\begin{Verbatim}
SELECT ?l WHERE {
    {(chadHurley type youtubeEmp):?l} 
    UNION
    {(chadHurley type paypalEmp):?l}
}
\end{Verbatim}
\nd and assume our dataset from Figure~\ref{fig:dataset-example} as input.  Considering the temporal domain, the intuitive meaning of the query is ``retrieve all time periods when \uri{chadHurley}
was an employee of Youtube or PayPal''.  In the case of $\mathsf{UNION}$ patterns the two instances of the variable $?l$
are treated as two different variables.  If the intended query would rather require treating both instances of the
variable $?l$ as the same, for instance to retrieve the time periods when \uri{chadHurley} was an employee of either
Youtube or PayPal but assuming we may not have information for one of the patterns, the query should rather look like:
\begin{Verbatim}[commandchars=+\[\]]
SELECT ?l WHERE {
    {(chadHurley type youtubeEmp):?l1} 
    UNION
    {(chadHurley type paypalEmp):?l2}
    ASSIGN ?l1 +(+lor+) ?l2 as ?l
}
\end{Verbatim}
where $\lor$ represents the domain specific built-in predicate for union of annotations.

\section{On primitive domains and their combinations} \label{sec:twists-asparql}

In this section we discuss some practical issues related to
(i) the representation of the temporal domain (Section~\ref{sec:doma-spec-issu});
(ii) the combination of several domains into one compound domain (Section~\ref{sec:extens-mult-doma});
(iii) the integration of differently annotated triples or non-annotated triples in the data or query (Section~\ref{sec:classical_triples}).

\subsection{Temporal issues}\label{sec:doma-spec-issu}

Let us highlight some specific issues inherent to the temporal domain.  Considering queries using Allen's temporal
relations~\cite{DBLP:journals/cacm/Allen83} (before, after, overlaps, etc.) as allowed in~\cite{tapp-bern-09}, we can
pose queries like ``find persons who were employees of PayPal before \uri{toivo}''.  This query raises some
ambiguity when considering that persons may have been employed by the same company at different disjoint intervals.
We can model such situations -- relying on sets of temporal intervals modelling the temporal domain. Consider our
dataset triples from Figure~\ref{fig:dataset-example} extended with the following triple:
\[
\begin{array}{l}
  \fuzzyg{\triple{\uri{toivo}, \type, \uri{paypalEmp}}}{\{[1999,2004], [2006,2008]\}}
\end{array}
\]
Tappolet and Bernstein~\cite{tapp-bern-09} consider this triple as two triples with disjoint intervals as
annotations. For the following query in their language $\tau$SPARQL:
\begin{Verbatim}
SELECT ?p WHERE {
    [?s1,?e1] ?p type youtubeEmp .
    [?s2,?e2] chadHurley type youtubeEmp .
    [?s1,?e1] time:intervalBefore [?s2,?e2]
}
\end{Verbatim}
we would get $\uri{chadHurley}$ as an answer although \uri{toivo} was already working for PayPal when \uri{chadHurley}
started.  This is one possible interpretation of ``before'' over a set of intervals. In \aSPARQL we could add
different domain specific built-in predicates, representing different interpretations of ``before''. For instance, we
could define binary built-ins (i) $\mathsf{beforeAny}(?A1, ?A2)$ which is true if there exists \emph{any} interval in
annotation $?A1$ before an interval in $?A2$, or, respectively, a different built-in $\mathsf{beforeAll}(?A1, ?A2)$
which is only true if \emph{all} intervals in annotation $?A1$ are before any interval in $?A2$. Using the latter, an
\aSPARQL query would look as follows:
\begin{Verbatim}
SELECT ?p WHERE {
    (?p type youtubeEmp):?l1 . 
    (toivo type youtubeEmp):?l2 .
    FILTER(beforeAll(?l1,?l2))
}
\end{Verbatim}
This latter query gives no result, which might comply with people's understanding of ``before'' in some cases, while we
also have the choice to adopt the behaviour of~\cite{tapp-bern-09} by use of $\mathsf{beforeAny}$ instead. 

More formally, if we consider an Allen relation $\mathbf{r}$ that holds between individual intervals, we can define a relation $\overline{\mathbf{r}}$ over sets of intervals in five different ways:
\begin{definition}\label{def:temporal_relations}
Let $T_1$ and $T_2$ be two non-empty sets of disjoint intervals. We define the following relations:
{\small\setlength{\leftmargini}{5mm}\begin{itemize}
 \item $\overline{\mathbf{r}}_{\exists\exists}=\{\tuple{T_1,T_2} \mid \exists t_1\in T_1,\exists t_2\in T_2\text{ such that }\tuple{t_1,t_2}\in\mathbf{r}\}$;
 \item $\overline{\mathbf{r}}_{\exists\forall}=\{\tuple{T_1,T_2} \mid \exists t_1\in T_1,\forall t_2\in T_2\text{ such that }\tuple{t_1,t_2}\in\mathbf{r}\}$;
 \item $\overline{\mathbf{r}}_{\forall\exists}=\{\tuple{T_1,T_2} \mid \forall t_1\in T_1,\exists t_2\in T_2\text{ such that }\tuple{t_1,t_2}\in\mathbf{r}\}$;
 \item $\overline{\mathbf{r}}_{\exists\forall\land\forall\exists}=\overline{\mathbf{r}}_{\exists\forall}\cap\overline{\mathbf{r}}_{\forall\exists}$;
 \item $\overline{\mathbf{r}}_{\forall\forall}=\{\tuple{T_1,T_2} \mid \forall t_1\in T_1,\forall t_2\in T_2\text{ such that }\tuple{t_1,t_2}\in\mathbf{r}\}$.
\end{itemize}}
\end{definition}

\nd These relations are illustrated by the following examples, taking the Allen relation $\mathsf{before}$:

\begin{example}
  Figure~\ref{fig:temporal_relations} is an example of time intervals that make each of the relations
  introduced in Definition~\ref{def:temporal_relations} true for the \emph{before} Allen relation.
\begin{figure*}
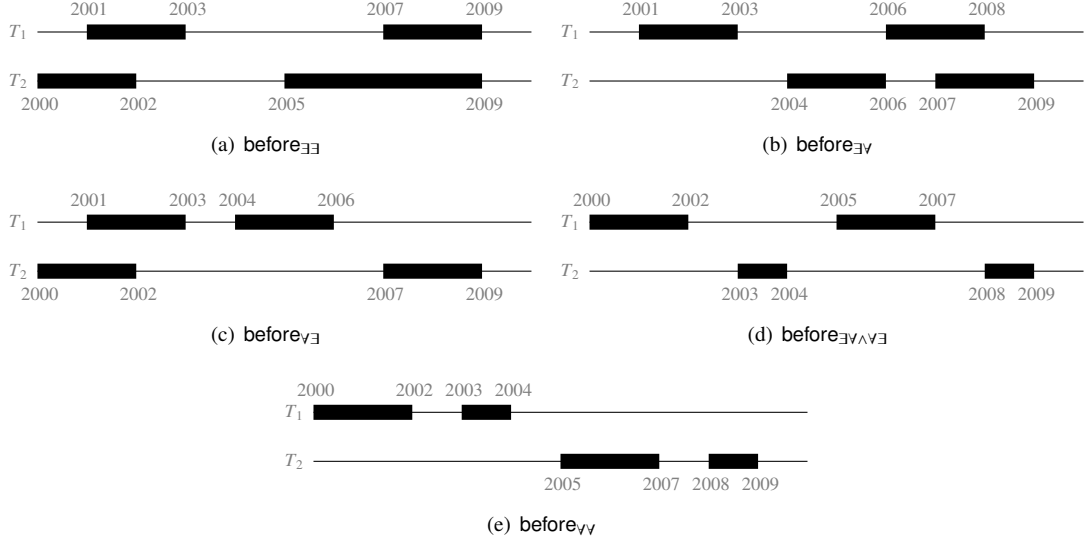

\centering
\subfigure[$\mathsf{before}_{\exists\exists}$]{\makegraph{{1/3,7/9}}{{0/2, 5/9}}\label{fig:tr1}} 
\subfigure[$\mathsf{before}_{\exists\forall}$]{\makegraph{{1/3,6/8}}{{4/6, 7/9}}\label{fig:tr2}}
\subfigure[$\mathsf{before}_{\forall\exists}$]{\makegraph{{1/3,4/6}}{{0/2, 7/9}}\label{fig:tr3}}
\subfigure[$\mathsf{before}_{\exists\forall\land\forall\exists}$]{\makegraph{{0/2,5/7}}{{3/4, 8/9}}\label{fig:tr4}}
\subfigure[$\mathsf{before}_{\forall\forall}$]{\makegraph{{0/2,3/4}}{{5/7, 8/9}}\label{fig:tr5}}
\caption{Temporal relations}\label{fig:temporal_relations}
\end{figure*}

\end{example}

It should be noticed that if one stick to one choice of quantifier, the resulting set of relations does not form a proper relation algebra. Indeed, it is easy to see that, in the first 3 cases, the relations are not disjoint. For instance, two sets of intervals can be involved in both a $\overline{\mathsf{before}}_{\exists\exists}$ and an $\overline{\mathsf{after}}_{\exists\exists}$ relation.
On the other hand, the last 4 cases are incomplete, that is, there are pairs of sets of intervals that cannot be related with any of the $\overline{\mathbf{r}}_{\exists\forall}$, $\overline{\mathbf{r}}_{\forall\exists}$, $\overline{\mathbf{r}}_{\forall\forall}$ or $\overline{\mathbf{r}}_{\exists\forall\land\forall\exists}$.

\begin{figure}[ht]
\centering
\begin{tikzpicture}[scale=0.7]
 \node (EE) at (0,0) {$\overline{\mathbf{r}}_{\exists\exists}$};
 \node (EA) at (-2,-2) {$\overline{\mathbf{r}}_{\exists\forall}$};
 \node (AE) at (2,-2) {$\overline{\mathbf{r}}_{\forall\exists}$};
 \node (EAAE) at (0,-4) {$\overline{\mathbf{r}}_{\exists\forall\land\forall\exists}$};
 \node (AA) at (0,-6) {$\overline{\mathbf{r}}_{\forall\forall}$};
 \draw[->] (AA) -- node[sloped,above,minimum width=0.1]{{\scriptsize $\subseteq$}} (EAAE);
 \draw[->] (EAAE) -- node[sloped,above,minimum width=0.1]{{\scriptsize $\subseteq$}} (AE);
 \draw[->] (EAAE) -- node[sloped,above,minimum width=0.1]{{\scriptsize $\supseteq$}} (EA);
 \draw[->] (AE) -- node[sloped,above,minimum width=0.1]{{\scriptsize $\supseteq$}} (EE);
 \draw[->] (EA) -- node[sloped,above,minimum width=0.1]{{\scriptsize $\subseteq$}} (EE);
\end{tikzpicture}
\caption{Hierarchy of relations.}\label{fig:hierarchy}
\end{figure}
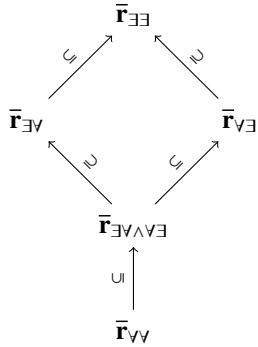

\subsection{Extensions to multiple domains}
\label{sec:extens-mult-doma}
Since annotations in our framework can range over different domains in different applications, one may be interested in combining several annotation domains such as annotating triples with a temporal term and a truth degree or degree of trust, etc.
In~\cite{Straccia10a}, we proposed an approach for easily combining multiple domains, based on the pointwise extension of domain operators to a product of domains. Here, we criticise this approach and propose a revised approach that better fits the intuition.

\subsubsection{Former approach and criticism}

The approach described in~\cite{Straccia10a} is the following. In general, assuming having domains $D_{1},
\ldots, D_{n}$, where $D_{i}= \tuple{L_{i}, \oplus_{i}, \otimes_{i}, \bot_{i}, \top_{i}}$, we may build the domain $D =
D_{1} \times \ldots \times D_{n} = \tuple{L, \oplus, \otimes, \bot, \top}$, where $L = L_{1} \times \ldots \times
L_{n}$, $\bot = \tuple{\bot_{1}, \ldots, \bot_{n}}$, $\top = \tuple{\top_{1}, \ldots, \top_{n}}$ and the
meet and join operations $\otimes$ and $\oplus$ are extended pointwise to $L$, \eg
$\tuple{\lambda_{1}, \ldots, \lambda_{n}} \otimes \tuple{\lambda_{1}', \ldots, \lambda_{n}}' = \tuple{\lambda_{1}
  \otimes \lambda'_{1}, \ldots, \lambda_{n} \otimes \lambda'_{n}}$.  For instance, 
\[
\fuzzyg{\triple{\uri{SkypeCollab}, \subclass, \uri{EbayCollab}}}{\tuple{[2009,2011], 0.3}}
\]
\noindent may indicate that during 2009-2011, the collaborators of Skype were also considered collaborators of Ebay to degree
$0.3$ (here we combine a temporal domain and a fuzzy domain).  The interesting point of our approach is that the rules
of the deductive systems need not be changed, nor the query answering mechanism (except to provide the support to
compute $\otimes$ and $\oplus$ accordingly).

The problem with this approach is that the annotations are dealt with independently from each others. As a result, \eg~the truth value $0.3$ does not apply to the time range $[2009,2011]$. This problem is made very apparent when one observes the unexpected consequences of our $\oplus$ operator on such a combination:

\[
    \begin{array}{l}
      \fuzzyg{\triple{\uri{SkypeCollab}, \subclass, \uri{EbayCollab}}}{\tuple{[2005, 2009], 1}}\\
      \fuzzyg{\triple{\uri{SkypeCollab}, \subclass, \uri{EbayCollab}}}{\tuple{[2009,2011], 0.3}}
		\end{array}
\]
\nd Applying the point-wise operation $\oplus$, this leads to the conclusion:
\[
\fuzzyg{\triple{\uri{SkypeCollab}, \subclass, \uri{EbayCollab}}}{\tuple{[2005, 2011], 1}}
\]

\nd This defies the intuition that, between 2005 and 2009, Skype collaborators where also Ebay employees (collaborate to degree $1$), but from 2009 to 2011 Skype collaborators were Ebay collaborators to the degree $0.3$. The pointwise aggregation does not follow this intuition and levels up everything. In the example above, we would like to say that the fuzzy value itself has a duration, so that the temporal interval corresponds more to an annotation of a quadruple.
Note that this problem is not specific to the combination of time and fuzziness. We observe a similar issue when combining provenance, for instance, with other domains:
\[
    \begin{array}{l}
    \fuzzyg{\triple{\uri{skypeEmp}, \subclass, \uri{ebayEmp}}}{\tuple{[2005, 2009], {\texttt{\footnotesize wikipedia}}}}\\
    \fuzzyg{\triple{\uri{skypeEmp}, \subclass, \uri{ebayEmp}}}{\tuple{[1958, 2012], {\texttt{\footnotesize wrong}}}}
    \end{array}
\]

\nd Using a point-wise aggregation method, the result would be:
\[
\begin{array}{rl}
\triple{\uri{skypeEmp}, \subclass, \uri{ebayEmp}}:& \langle[1958, 2012],\\
                                                  & \texttt{\footnotesize wikipedia}\lor\texttt{\footnotesize wrong}\rangle
\end{array}
\]
\nd which entails:
\[
\fuzzyg{\triple{\uri{skypeEmp}, \subclass, \uri{ebayEmp}}}{\tuple{[1958, 2012], {\texttt{\footnotesize wikipedia}}}}
\]

\nd Again, the problem is that provenance here does not define the provenance of the temporal annotation and the temporal annotation is not local to a certain provenance.
  
In order to match the intuition, we devise a systematic construction that defines a new compound domain out of two existing domains.

\subsubsection{Improved Formalisation} \label{sec:formalisation}

In this section, we propose a generic construction that builds an annotation domain by combining two predefined domains in a systematic way. To achieve this, we will assume the existence of two annotation domains $D_1=\tuple{L_1,\oplus_1,\otimes_1,\bot_1,\top_1}$ and $D_2=\tuple{L_2,\oplus_2,\otimes_2,\bot_2,\top_2}$ which will be instantiated in examples with the temporal domain for $D_1$ (abbreviated $D_\mathsf{t}$) and either the fuzzy domain ($D_\mathsf{f}$) or the provenance domain ($D_\mathsf{p}$) for $D_2$. We denote the temporal and fuzzy combination \textit{time+fuzzy}, and the temporal and provenance combination \textit{time+provenance}.

\paragraph{Intuition and desired properties}

In our former approach, we remarked that some information is lost in the join operation. Considering time+fuzzy, we see that the join should represent temporary \emph{changes} in the degree of truth of the triple. Yet, it is clear that representing such changes cannot be done with a simple pair (intervals,value). So, as a first extension of our previous na\"ive solution, we suggest using sets of pairs of primitive annotations, as exemplified below.
\[
\begin{array}{ll}
\triple{\uri{SkypeCollab}, \subclass, \uri{EbayCollab}}: & \{\tuple{[2005, 2009], 1},\\
                                                   & \ \tuple{[2009,2011], 0.3}\}
\end{array}
\]

\nd Starting from this, we devise an annotation domain that correctly matches the intuitive meaning of the compound annotations. The annotated triple above can be interpreted as follows: for each pair in the annotation, for each time point in the temporal value of the pair, the triple holds to at least the degree given by the fuzzy value of the pair. The time+provenance combination is interpreted analogously, except that the triple holds (at least) in the context given by the provenance value of the pair.

This interpretation of the compound annotations implies that multiple sets of pairs can convey the exact same information. For example, the following time+fuzzy annotated triples are equivalent:
\[
    \begin{array}{ll}
    \triple{\uri{SkypeCollab}, \subclass, \uri{EbayCollab}}:&\hspace{-3mm}\{\tuple{[2005,2009], 1}\}\\
    \triple{\uri{SkypeCollab}, \subclass, \uri{EbayCollab}}:&\hspace{-3mm}\{\tuple{[2005,2009], 0.3},\\
    &\hspace{-3mm}\tuple{[2005,2009], 1}\}
  \end{array}
\]

\nd From this observation, we postulate the following desired property:

\begin{property}\label{prop:functionality}
For all $x\in L_1, y,y'\in L_2$ and for all \rhodf triples $\tau$, $\tau:\{\tuple{x,y},\tuple{x,y'}\}$ is semantically equivalent to $\tau:\{\tuple{x,y\oplus_2 y'}\}$.
\end{property}

\nd Consequently, it is always possible to assign a unique element $y\in L_2$ to a given element of $L_1$. Thus, an arbitrary set of pairs in $L_1\times L_2$ is equivalently representable as a partial mapping from $L_1$ to $L_2$. Additionally, given a certain time interval, we can easily compute the maximum known degree to which a time+fuzzy annotated triple holds. For instance, with the annotation $\{\tuple{[2005,2009], 0.3},\tuple{[2008,2011], 1}\}$, we can assign the degree $1$ to any subset of $[2008,2011]$; the degree $0.3$ to any subset of $[2005,2009]$ which is not contained in $[2008,2009]$; the degree $0$ to any other temporal value.

This remark justifies that we can consider a compound annotation $A$ as a \emph{total} function from $L_1$ and $L_2$. From now on, whenever $A$ is a finite set of pairs, we will denote by $\overline{A}$ the function that maps elements of $L_1$ to an element of $L_2$ that, informally, minimally satisfies the constraints imposed by the pairs in $A$. This is formalised below. For instance, if $A=\{\tuple{[2005,2009], 0.3},\tuple{[2008,2011], 1}\}$, then:
\begin{displaymath}
  \overline{A}(x) = \left\{ \begin{array}{ll}
   1   & \text{if } x\subseteq[2008,2011] \\
   0.3 & \text{if } x\subseteq[2005,2009] \text{ and } x\not\subseteq[2008,2009] \\
   0   & \text{otherwise}
  \end{array} \right.
\end{displaymath}

Whereas in this example, for the time+fuzzy domain the value of $\overline{A}$ for a particular interval seems to follow quite intuitively, let us next turn to the less obvious combination of time+provenance. 
Here, we postulate that the following triples
\[
\begin{array}{rl}
  \tau: \{& \tuple{[2005,2009], \texttt{\footnotesize wikipedia}},\\
          & \tuple{[2008,2011], \texttt{\footnotesize wrong}}\}\\
  \tau: \{& \tuple{[2005,2007], \texttt{\footnotesize wikipedia}},\\
          & \tuple{[2007,2009], \texttt{\footnotesize wikipedia}},\\
          & \tuple{[2008,2011], \texttt{\footnotesize wrong}}\}\\
  \tau: \{& \tuple{[2005,2008], \texttt{\footnotesize wikipedia}},\\
          & \tuple{[2008,2009], \texttt{\footnotesize wikipedia}\lor\texttt{\footnotesize wrong}},\\
          & \tuple{[2009,2011], \texttt{\footnotesize wrong}}\}
\end{array}
\]

\noindent represent in fact equivalent annotations. 
Let us check the intuition behind this on a particular interval, $[2005,2009]$, which for the first triple has unambigously associated the
provenance value \texttt{\footnotesize wikipedia}.
Considering the second annotated triple, we observe that the provenance \texttt{\footnotesize wikipedia} can likewise be associated with the interval $[2005,2009]$ because this provenance is associated with two intervals that -- when joined -- cover the time span $[2005,2009]$. In the case of the last annotated triple, the provenance \texttt{\footnotesize wikipedia}$\lor$\texttt{\footnotesize wrong} means that the triple holds in \texttt{\footnotesize wikipedia} as well as in \texttt{\footnotesize wrong} (notice that $x\lor y$ means that the assertion holds in $x$ and in $y$ likewise, see Section~\ref{sec:provenance-domain} for details). Intuitively, we expect for the last triple that the provenance associated with the joined interval $[2005,2009]$ is obtained from applying the meet operator over the respective provenance annotations \texttt{\footnotesize wikipedia} (for the partial interval $[2005,2008]$) and $\texttt{\footnotesize wikipedia}\lor\texttt{\footnotesize wrong}$ (for the partial interval $[2008,2009]$), \ie $(\texttt{\footnotesize wikipedia}\lor\texttt{\footnotesize wrong}) \land \texttt{\footnotesize wikipedia}$ 
 which -- again -- is equivalent to \texttt{\footnotesize wikipedia} in the provenance domain. Besides, considering now the interval $[2005,2011]$, the triple is true in either \texttt{\footnotesize wikipedia.org} or \texttt{\footnotesize wrong}, which is modelled as $(\texttt{\footnotesize wikipedia}\land\texttt{\footnotesize wrong})$ in the provenance domain.
Let us cast this intuition into another property we want to ensure on the function $\overline{A}$:

\begin{property}\label{prop:overlap}
 Given a set of annotation pairs $A$, for all $x_0\in L_1$ whenever $\exists J\subseteq A$ with
               $x_0\preceq_1 \bigoplusone{\tuple{x,y}\in J}x$, we have
$\overline{A}(x_0)\succeq_2 \bigotimestwo{\tuple{x,y}\in J}y$.
\end{property}

Our goal in what follows is to characterise the set of functions associated with a finite set of pairs, that is $\{\overline{A}\mid A\subseteq L_1\times L_2\}$, in a manner such that Property~\ref{prop:functionality} and Property~\ref{prop:overlap} are satisfied.

\paragraph{Formalisation}

As mentioned before, a compound annotation can be seen as a function that maps values of the first domain to values of the second domain. In order to get the desired properties above established, we restrict this function to a particular type of functions that we call \emph{quasihomomorphism} because it closely resembles a semiring homomorphism.

\begin{definition}[Quasihomomorphism]
Let $f$ be a function from $D_1 = \tuple{L_1,\oplus_1,\otimes_1,\bot_1,\top_1}$ to $D_2 = \tuple{L_2,\oplus_2,\otimes_2,\bot_2,\top_2}$. $f$ is a \emph{quasihomomorphism} of domains \iff for all $x,y\in L_1$: (i) $f(x\oplus_1 y)\succeq_2 f(x)\otimes_2 f(y)$ and (ii) $f(x\otimes_1 y) \succeq_2 f(x)\oplus_2 f(y)$. %
\end{definition}

\nd We now use quasihomomorphisms to define -- on an abstract level -- a compound domain of annotations.

\begin{definition}[Compound annotation domain]
Given two primitive annotation domains $D_1$ and $D_2$, the \emph{compound annotation domain} of $D_1$ and $D_2$ is the tuple $\tuple{L_{12},\oplus_{12},\otimes_{12},\bot_{12},\top_{12}}$ defined as follows:

\begin{itemize}
 \item $L_{12}$ is the set of quasihomomorphisms from $D_1$ to $D_2$;
 \item $\bot_{12}$ is the function defined such that for all $x\in L_1$, $\bot_{12}(x) = \bot_2$;
 \item $\top_{12}$ is the function defined such that for all $x\in L_1$, $\top_{12}(x) = \top_2$;
 \item for all $\lambda,\mu\in L_{12}$, for all $x\in L_1$, $(\lambda\oplus_{12}\mu)(x)=\lambda(x)\oplus_2\mu(x)$;
 \item for all $\lambda,\mu\in L_{12}$, for all $x\in L_1$, $(\lambda\otimes_{12}\mu)(x)=\lambda(x)\otimes_2\mu(x)$;
\end{itemize}
\end{definition}

\nd This definition yields again a valid RDF annotation domain, as stated in the following proposition:

\begin{proposition}
 $\tuple{L_{12},\oplus_{12},\otimes_{12},\bot_{12},\top_{12}}$ is an idempotent, commutative semiring and $\oplus_{12}$ is $\top_{12}$-annihilating.
\end{proposition}

\nd Quasihomomorphisms are abstract values that may not be representable syntactially. By analogy with XML datatypes~\cite{xmlschema11-2}, we can say that they represent the \emph{value space} of the compound domain. In the following, we want to propose a finite representation of some of these functions.
Indeed, as we have seen in the examples above, we intend to represent compound annotations just as finite sets of pairs of primitive annotations. Thus, continuing the analogy, the \emph{lexical space} is merely containing finite sets of pairs of primitive annotation values. To complete the definition, we just have to define a mapping from such finite representation to a corresponding quasihomomorphism. That is, we have to define the \emph{lexical-to-value mapping}.

Consider again the (primitive) domains $D_1$ and $D_2$ and let $A\subseteq L_1\times L_2$ be a finite set of pairs of primitive annotations. We define the function $\overline{A}: D_1\to D_2$ as follows:\footnote{Note that as $D_{2}$ is an annotation domain, the $\mathbf{lub}$ operation is well defined.}
\[
\forall z\in L_1, \overline{A}(z) = \mathbf{lub}\{\bigotimestwo{\tuple{x,y}\in J}y \mid J\subseteq A \text{ and }z\preceq_1\bigoplusone{\tuple{x,y}\in J}x\} \ .
\]

\begin{theorem}\label{thm:strongquasi}
If $A\subseteq L_1\times L_2$ is a finite set of pairs of primitive annotations, then $\overline{A}$ is a quasihomomorphism.
\end{theorem}

\nd The proof is mostly a sequence of manipulation of notations with little subtlety, so we refer the reader to Appendix~\ref{sec:proofs} for details.

Now, we know that we can translate an arbitrary finite set of pairs of primitive annotations into a compound annotation. However, using arbitrary sets of pairs is problematic in practice for two reasons: (1) several sets of pairs have equivalent meaning,\footnote{Particularly, we note that there can still be an infinite set of finite representations of the same compound annotation.} that is, the function induced by the two sets are identical; (2) the approach does neither gives a programmatic way of computing the operations ($\otimes_{12}$,$\oplus_{12}$) on compound annotations, nor gives us a tool to finitely represent the results of these operations.

Thus, we next turn towards how to choose a canonical finite representative for a finite set annotation pairs. To this end, we need a normalising function $N:\pow{L_1\times L_2}\to\pow{L_1\times L_2}$ such that for all $A,A'\subseteq L_1\times L_2, \overline{A}=\overline{A'}$ \iff $N(A)=N(A')$.  This will in turn also allow us to define the operations $\oplus_{12}$ and $\otimes_{12}$ over the set of normalised annotations.

\paragraph{Normalisation} We propose a normalisation algorithm based on two main operations:
\begin{description}
 \item[$\mathsf{Saturate}$:] informally, the saturate function increases the size of a set of pairs of annotations by adding any redundant pairs that ``result from the application of $\otimes$ and $\oplus$ to values existing in the initial pairs'';
 
 \item[$\mathsf{Reduce}$:] takes the output of the saturation step and removes ``subsumed'' pairs.
\end{description}

\nd In particular, the $\mathsf{Saturate}$ algorithm is adding pairs of annotations to the input such that in the end, all primitive annotations that can be produced by the use of existing values and operators $\otimes$ and $\oplus$ appear in the output.
The algorithm for $\mathsf{Saturate}$, $\mathsf{Reduce}$ and $\mathsf{Normalise}$ are given in Algorithm~\ref{alg:saturate}, Algorithm~\ref{alg:reduce} and Algorithm~\ref{alg:normalise} respectively.

\begin{algorithm}[ht!]
\caption{$\mathsf{Saturate}(A)$}\label{alg:saturate}
\begin{algorithmic}
 \REQUIRE $A\subseteq L_1\times L_2$ finite
 \ENSURE $\mathsf{Saturate}(A)$
 \STATE $R := \emptyset$;
 \FORALL {$X\subseteq \pow{A}$}
	  \STATE $R := R \cup \{\tuple{\bigoplusone{J\in X}\bigotimesone{\tuple{x,y}\in J}x,\bigotimestwo{J\in X}\bigoplustwo{\tuple{x,y}\in J}y}\}$;
	  \STATE $R := R \cup \{\tuple{\bigotimesone{J\in X}\bigoplusone{\tuple{x,y}\in J}x,\bigoplustwo{J\in X}\bigotimestwo{\tuple{x,y}\in J}y}\}$;
 \ENDFOR
 \RETURN $R$;
\end{algorithmic}
\end{algorithm}

\nd If the operations $\otimes_1$ and $\otimes_2$ are idempotent, Algorithm~\ref{alg:saturate} ensures that given a value $x\in L_1$ that is the result of using operators $\otimes_1$ and $\oplus_1$ on any number of primitive annotations of $L_1$ appearing in $A$, then there exists $y\in L_2$ such that $\tuple{x,y}$ exists in the output of $\mathsf{Saturate}$. Similarly, given $y\in L_2$ that can be obtained from combinations of values of $L_2$ appearing in $A$ and operators $\otimes_2$ and $\oplus_2$, then there exists $x\in L_2$ such that $\tuple{x,y}$ exists in the output of $\mathsf{Saturate}$.

\begin{example}\label{exEE}
Consider the following time+fuzzy annotation:
\[
    \{\tuple{[2000,2005], 0.7},\tuple{[2002,2008], 0.5}\}
\]
Application of the function saturate gives the following result:
\begin{align*}
    &\{\tuple{[2000,2005], 0.7},\tuple{[2002,2008], 0.5},\\
    &  \tuple{[2000,2008], 0.35},\tuple{[2002,2005], 0.7},\\
    &  \tuple{[2000,2005], 0.49},\tuple{[2002,2005], 0.49}\}
\end{align*}
\end{example}

\nd Now we notice that this can introduce redundant information, which should be eliminated. This is the goal of the function $\mathsf{Reduce}$ which is defined by Algorithm~\ref{alg:reduce}.

\begin{algorithm}[ht!]
\caption{$\mathsf{Reduce}(A)$}\label{alg:reduce}
\begin{algorithmic}
 \REQUIRE $A\subseteq L_1\times L_2$ finite and saturated
 \ENSURE $\mathsf{Reduce}(A)$
 \WHILE {$\exists\tuple{x,y}\in A,\exists\tuple{x',y'}\in A\setminus\{\tuple{x,y}\}$ such that $x\preceq_1 x'$ \AND $y \preceq_2 y'$}
   \STATE $R := R\setminus\{\tuple{x,y}\}$;
 \ENDWHILE
 \WHILE {$\exists\tuple{x,y}\in A$ such that $x = \bot_1$ \OR $y = \bot_2$}
   \STATE $R := R\setminus\{\tuple{x,y}\}$;
 \ENDWHILE
 \RETURN $R$;
\end{algorithmic}
\end{algorithm}

\begin{example}
Considering Example~\ref{exEE} the output of the $\mathsf{Saturate}$ algorithm above, the $\mathsf{Reduce}$ function gives the following result::
\begin{align*}
    &\{\tuple{[2000,2005], 0.7},\tuple{[2002,2008], 0.5},\tuple{[2000,2008], 0.35}\}
\end{align*}
\end{example}

\begin{algorithm}[ht!]
\caption{$\mathsf{Normalise}(A)$}\label{alg:normalise}
\begin{algorithmic}
 \REQUIRE $A\subseteq L_1\times L_2$ finite
 \ENSURE $\mathsf{Normalise}(A)$
 \RETURN $\mathsf{Reduce}(\mathsf{Saturate}(A))$;
\end{algorithmic}
\end{algorithm}

\begin{example}
Consider the following time+provenance annotation:
\[
    \{\tuple{[1998,2006],\texttt{\footnotesize wikipedia}},\tuple{[2001,2011],\texttt{\footnotesize wrong}}\}
\]
which normalises to:
\begin{align*}
    &\{\tuple{[1998,2011],\texttt{\footnotesize wikipedia}\land\texttt{\footnotesize wrong}},\tuple{[1998,2006],\texttt{\footnotesize wikipedia}},\\&\tuple{[2001,2011],\texttt{\footnotesize wrong}},\tuple{[2001,2006],\texttt{\footnotesize wikipedia}\lor\texttt{\footnotesize wrong}}\}
\end{align*}
Note that the pair $\tuple{[2001,2006],\texttt{\footnotesize wikipedia}\lor\texttt{\footnotesize wrong}}$ is introduced by Line~6 of Algorithm~\ref{alg:saturate} and is not discarded during the reduction phase.
\end{example}

\nd The following property can be shown:

\begin{proposition}\label{prop:normal}
 If $D_1=\tuple{L_1,\oplus_1,\otimes_1,\bot_1,\top_1}$ is a lattice then, for all $A\subseteq L_1\times L_2$ finite, $\overline{A}=\overline{\mathsf{Normalise}(A)}$.
\end{proposition}

\nd Notice that we must impose that the first primitive domain of annotation is a lattice for the normalisation to work, that is, we need that $z\preceq_1 x$ and $z\preceq_1 y$ iff $z\preceq_1 x\otimes_1 y$. Details of the proof can be found in Appendix~\ref{sec:proofs}.

The following theorem shows that the normalisation is actually unique up to equivalence of the corresponding functions.

\begin{theorem}\label{thm:normalisation}
 If $D_1=\tuple{L_1,\oplus_1,\otimes_1,\bot_1,\top_1}$ is a lattice then, for all $A,B\subseteq L_1\times L_2$ finite then $\overline{A}=\overline{B}\Leftrightarrow\mathsf{Normalise}(A)=\mathsf{Normalise}(B)$.
\end{theorem}

\nd Again, to improve readability, we put the proof in Appendix~\ref{sec:proofs}.

\paragraph{Operations on normalised annotations}
We can now present the operations $\oplus_{12}$ and $\otimes_{12}$ on normalised finite sets of pairs.
\begin{itemize}
 \item $A\oplus_{12} B=\mathsf{Normalise}(A\cup B)$;
 \item $A\otimes_{12} B=\mathsf{Normalise}(\{\tuple{x\otimes_1 x',y\otimes_2 y'}\mid\tuple{x,y},\tuple{x',y'}\in A\times B\})$.
\end{itemize}

\nd Finally, with the proposed representation and operations, we devised a systematic approach to compute combination of domains using existing primitive domains. This implies that an implementation would not need to include operators that are specific to a given combination, as long as programmatic modules exist for the primitive annotation domains.

\subsubsection{Discussion} \label{sec:combinediscuss}

Our definition of a compound annotation domain is, to the best of our knowledge, a novelty in settings involving annotations: previous work on annotated RDF~\cite{Udrea10,Straccia10a}, annotated logic programmes~\cite{Kifer92} or annotated database relations~\cite{Green_Karvounarakis_Tannen:07} have not addressed this issue.  We present in this section some considerations with respect to the chosen approach.

\begin{enumerate}
 \item The normalisation algorithm is not optimised and would prove inefficient if directly implemented ``as is''. In this part, we have provided a working solution for normalising compound annotation as a mere proof of existence of such a solution.  By observing the examples that we provide for the time+fuzzy domain, it seems that the cost of normalising can be reduced significantly with appropriate strategies.

 \item As indicated by Theorem~\ref{thm:normalisation}, we only ensure that the normalisation is feasible for a combination of annotation domains where at least one is a lattice. Whether a normalisation function exists in the more general case of two commutative, idempotent, $\top$-annihilating semirings is an open question. %

 \item The method we provide defines a new domain of annotation in function of existing domains, such that it is possible to reason and to query triples annotated with pairs of values. This does not mean that it is possible to reason with a combination of triples annotated with the values of the first domain, and triples annotated with values of the second domain. For instance, reasoning with a combination of temporally annoted triples and fuzzy annotated triples does not boil down to reasonning over time+fuzzy-annotated triples. The next section discusses this issues and how non-annotated triples can be combined with annotated triples.
\end{enumerate}

\subsection{Integrating differently annotated triples in data and queries}\label{sec:classical_triples}
While our approach conservatively extends RDFS, we would like to be able to seamlessly reason with and query together
annotated triples and non-annotated triples. Since non-annotated triples can be seen as triples annotated with boolean
values, we can generalise this issue to reasoning and querying graphs annotated with distinct domains. For instinct, let
us assume that a dataset provides temporally annotated triples, another one contains fuzzy-annotated triples and yet
another is a standard RDF dataset. We want to provide a uniform treatment of all these datasets and even handle the merge
of differently annotated triples. Moreover, we expect to allow multiple annotation domains in AnQL queries.

\subsubsection{Multiple annotation domains in the data}

Consider the following example:
\[
\begin{array}{l}
  \fuzzyg{\triple{\uri{chadHurley}, \type, \uri{googleEmp}}}{[2006,2010]}\\
  \fuzzyg{\triple{\uri{chadHurley}, \type, \uri{googleEmp}}}{0.7}\\
  \fuzzyg{\triple{\uri{googleEmp}, \subclass, \uri{Person}}}{0.97}
\end{array}
\]
We can assume that the subclass relation has been determined by ontology matching algorithms, which typically return
confidence measures in the form of a number between $0$ and $1$. Consider as well the following example queries:

\begin{example}\label{sec:example-query01}
\  \begin{Verbatim} 
SELECT ?a WHERE {
    (chadHurley type googleEmp):?a
}
  \end{Verbatim}
\end{example}

\begin{example}\label{sec:example-query02}
\  \begin{Verbatim} 
SELECT ?a WHERE {
    (chadHurley type Person):?a
}
  \end{Verbatim}
\end{example}
We propose two alternative approaches to deal with multiple annotation domains. The first one simply seggregates the domains
of annotations, such that no inferences are made across differently annotated triples. The second one takes advantage of
the compound domain approach defined in Section~\ref{sec:extens-mult-doma}.

\paragraph{Seggregation of domains}
With this approach, distinct domains are not combined during reasoning, such that the first annotated triple together
with the third one would not produce new results. The query from Example~\ref{sec:example-query01} would have the
following answers: $\{?a/[2006,2010]\},\{?a/0.7\}$.  The query from Example~\ref{sec:example-query02} would have the answer $\{?a/0.679\}$ (under product t-norm $\otimes$).

The main advantage is that query answering is kept very straightforward. Moreover, it is possible to combine different annotation domains within the query by simply joining results from the seggregated datasets. The drawback is that reasoning would not complement non-annotated knowledge with annotated one and vice versa.

\paragraph{Using compound domains}
The principle of this approach is to assume that two primitive annotation values from distinct domains actually represent a pair with an implicit default value for the second element. The default value can be domain dependent or generic, such as using $\top$ or $\bot$ systematically. An example of domain specific default is found in~\cite{gutierrez-etal-2007} where the value $[-\infty,\mathsf{Now}]$ is used to fill the missing annotations in standard, non-annotated RDF. It can be noticed that using $\bot$ as a default would boil down to having seggregated datasets, as in the previous approach. The use of $\top$ has the advantage of being generic and allows one to combine knowledge from differently annotated sources in inferences. So, the query in
Example~\ref{sec:example-query01} has answer $\{?a/\{\tuple{[2006,2010],1},\tuple{[-\infty,+\infty],0.7}\}\}$, while Example~\ref{sec:example-query02} has the answer $\{?a/\{\tuple{[2006,2010],0.97},\tuple{[-\infty,+\infty], 0.679}\}\}$.

The main advantage is the possibility to infer new statements by combining various annotated or non-annotated triples. The drawbacks are that (i) our combination approach is, so far, limited to the case where one domain is a lattice; (ii) if triples with a new annotation domain are added, then it adds one dimension to the answers, which obliges to recompute existing answers; (iii) the combination of more than two domains may be particularly complex and possibly non-commutative.

\subsubsection{Multiple annotation domains in the query}
\label{sec:uniformquery}
When dealing with multiple domains in the query, we face a similar choice as in the data, but we are also offered the option to replace the default value with a variable. If seggregation of domains has been chosen, then distinct domains in the query are only used to match the corresponding data, but it is still possible to combine the results from differently annotated sources.
For instance:

\begin{Verbatim}[commandchars=\\|~,codes={\catcode`\$=3\catcode`_=8}]
SELECT ?e ?c ?t ?f WHERE {
  {(?c sc ebayEmp):?f}
 UNION
  {(?e type ?c):?t.}
 FILTER{?t $\preceq_t$ [2005,2011] AND ?f $\preceq_f$ 0.5}
}\end{Verbatim}

This query can be executed even on a dataset that does not include fuzzy value. The fuzzy-annotated triple pattern can be simply ignored and the temporally annotated pattern evaluated.

In the case of the second approach using compound domains, the choices are as follow:
\begin{enumerate}
\item add a single fresh annotation variable for all triples in the query that are missing a value for an annotation domain; or
\item add a different fresh annotation variable for each triple in the query; or
\item add a constant annotation such as $\top$ to all missing annotation values.
\end{enumerate}
In later discussions, we will use the meta-variable $\mathrm{\Theta}_{D}$ to represent the default value of domain $D$ assigned to annotations in the query triples.
\begin{example} \label{non-annotated} For instance, if we again consider the query (excluding the annotation variables)
  and input data from Example~\ref{exx1}, the query would look like:
\begin{Verbatim}
SELECT ?p ?c WHERE {
    (?p type :ebayEmp)
    OPTIONAL{(?p :hasCar ?c)}
}
\end{Verbatim}
\end{example}
Now, given the above three approaches for transforming this query we would get the following answers:
{\small
\renewcommand{\tabcolsep}{0.5cm}
\begin{center}
  \begin{tabular}{|c|cc|}
    \hline
    \multirow{3}{*}{Approach 1} & $?p/\uri{toivo}$ &  - \\
    & $?p/\uri{toivo}$ & $?c/\uri{peugeot}$\\
    & $?p/\uri{toivo}$ & $?c/\uri{renault}$\\
    \hline
    \multirow{2}{*}{Approach 2} & $?p/\uri{toivo}$ & $?c/\uri{peugeot}$\\
    & $?p/\uri{toivo}$ & $?c/\uri{renault}$\\
    \hline
    Approach 3 & \multicolumn{2}{c|}{$\emptyset$}\\
    \hline
  \end{tabular}
\end{center}
}

\subsubsection{Querying multi-dimensional domains}\label{sec:query-multi-dimens}

Similarly to the discussion in the previous subsection, we can encounter mismatches between the Annotated RDF dataset
and the \aSPARQL query.  In case the \aSPARQL query contains only variables for the annotations, the query can be
answered on any Annotated RDF dataset.  From a user perspective, the expected answers may differ from the actual
annotation domain in the dataset, \eg the user may be expecting temporal intervals in the answers when the answers
actually contain a fuzzy value.  For this reason some built-in predicates to determine the type of annotation should be
introduced, like $\mathsf{isTEMPORAL}$, $\mathsf{isFUZZY}$, etc.

If the \aSPARQL query contains annotation values and the Annotated RDF dataset contains annotations from a different
domain, one option is to not provide any answers.  Alternatively, we can consider combining the domain of the query with
the domain of the annotation into a multi-dimensional domain, as illustrated in the next example.

\begin{example} \label{exx1C}
 Assuming the following input data:
\vspace{-2mm}
   \[
     \fuzzyg{\triple{\uri{chadHurley}, \type, \uri{youtubeEmp}}}{{\texttt{\footnotesize chad}}}
     \vspace{-2mm}\]
 When performing the following query:
\begin{Verbatim}
SELECT ?p ?c WHERE {
    (?p type ?c):[2009, 2010]
}
\end{Verbatim}
we would interpret the data to the form: %
\vspace{-2mm}  \[
    \fuzzyg{\triple{\uri{chadhurley}, \type, \uri{youtubeEmp}}}{\tuple{{\texttt{\footnotesize chad}}, \mathrm{\Omega}_{\mathrm{temporal}}}}
\vspace{-2mm}  \]
while the query would be interpreted as:
\begin{Verbatim}[commandchars=\\|~,codes={\catcode`\$=3\catcode`_=8}]
SELECT ?p ?c WHERE {
    (?p type ?c):$\langle\Theta_\mathrm|provenance~$, [2009, 2010]$\rangle$
}
\end{Verbatim}
\nd where $\mathrm{\Omega}_{\mathrm{temporal}}$ and $\mathrm{\Theta}_{\mathrm{provenance}}$ are annotations corresponding to the default values of their respective domains, as discussed
in Section~\ref{sec:uniformquery}. The semantics of combining different domains into one multi-dimensional domain has
been discussed in Section~\ref{sec:extens-mult-doma}.
\end{example}

\section{Implementation Notes}
\label{sec:implementation-notes}
Our prototype implementation is split into two distinct modules: one for that implements the Annotated RDFS inferencing
and the second module is an implementation of the \aSPARQL query language that relies on the first module to retrieve
the data.  Our prototype implementation is based on SWI-Prolog's Semantic Web
library~\cite{WielemakerHuangvan-der-Meij:2008aa} and we present the architecture of the implementation in
Figure~\ref{fig:aRDF-schema}.

Our Annotated RDFS module consists of a bottom-up reasoner used to calculate the closure of a given RDF dataset
\textbf{1)}. The variable components comprise \textbf{2)} the specification of the given annotation domain; and
\textbf{3)} the ruleset describing the inference rules and the way the annotation values should be propagated.  For
\textbf{1)} we do not suggest a special RDF serialisation for temporal triples but rely on existing proposals using
reification~\cite{gutierrez-etal-2007}.  Annotation domains in \textbf{2)} are to be specified by appropriate lattice
operations and describing default annotations for non-annotated triples.

The rules in \textbf{3)} are specified using a high-level language to specify domain independent rules that abstracts
from peculiarities of the reification syntax. For example the following rule provides \textit{subclass inference} in the
RDFS ruleset:
{\small
\begin{verbatim}
rdf(O, rdf:type, C2, V) <== 
         rdf(O, rdf:type, C1, V1), 
         rdf(C1, rdfs:subClassOf, C2, V2), 
         infimum(V1, V2, V). 
\end{verbatim}}
\noindent \textbf{2)} and \textbf{3)} are independent of each other: it is possible to combine arbitrary rulesets and
domains (see above). %

The \aSPARQL module also implemented in Prolog relies on the SPARQL implementation provided by the ClioPatria Semantic
Web Server.\footnote{\url{http://www.swi-prolog.org/web/ClioPatria/}} For the \aSPARQL implementation, the domain
specification needs to be extended with the grammar rules to parse an annotation value and any built-in functions
specific to the domain.

More information and downloads of the prototype implementation can be found at {\small\url{http://anql.deri.org/}}.

\begin{figure}[t]
  \centering
  \includegraphics[scale=0.45]{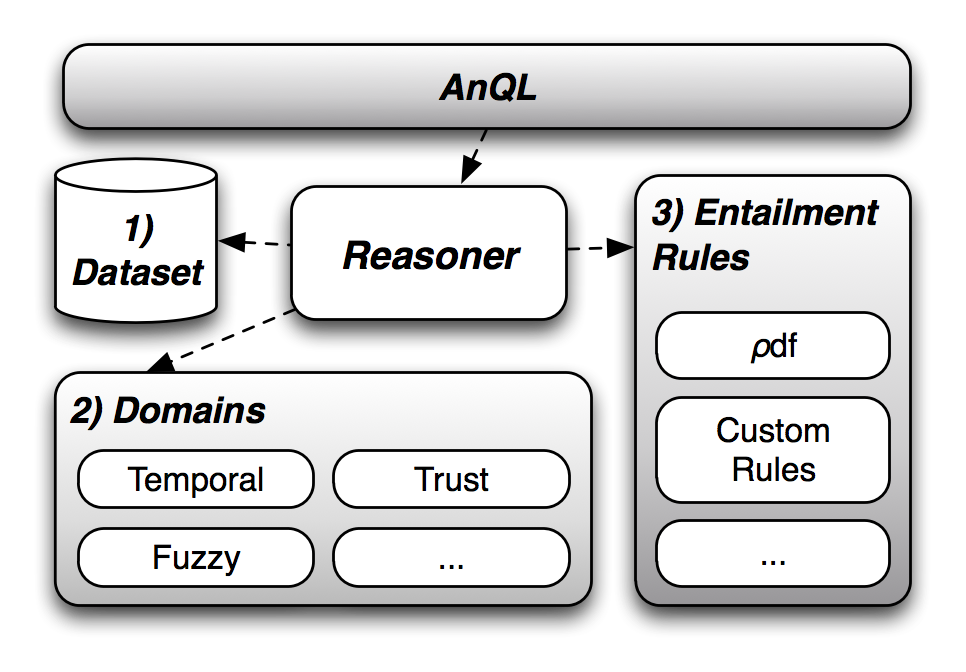}

  \caption{Annotated RDF implementation schema}

  \label{fig:aRDF-schema}
\end{figure}

\subsection{Implementation of specific domains}
\label{sec:impl-spec-doma}

For example, for the fuzzy domain the default value is considered to be $1$ and the $\otimes$ and $\oplus$ operations are,
respectively, the \textit{min} and \textit{max} operations.  The \aSPARQL grammar rules consist simply of calling the
parser predicate that parses a decimal value.

As for the temporal domain, we are representing triple annotations as \emph{ordered} list of disjoint time intervals.
This implies some additional care in the construction of the $\otimes$ and $\oplus$ operations.  For the representation
of $-\infty$ and $+\infty$ we are using the \stt{inf} and \stt{sup} Prolog atoms, respectively. Concrete time points are
represented as integers and we use a standard constraint solver over finite domains (CLPFD) in the $\otimes$ and
$\oplus$ operations. The default value for non-annotated triples is %
\stt{[inf,sup]}. The $\otimes$ operation is implemented as the recursive intersection of all the elements of the
annotation values, \ie temporal intervals.  The $\oplus$ operation is handled by constructing CLPFD expressions that
evaluate the union of all the temporal intervals.%
Again, the \aSPARQL grammar rules take care of adapting the parser to the specific domain and we have defined the domain
built-in operations described in Section~\ref{sec:doma-spec-issu}.

\subsection{Use-case example: Sensor Data}
\label{sec:sensor-data}
As a use-case for Annotated RDF and \aSPARQL, we present the scenario of exposing sensor readings as RDF data.
Representing sensor data as RDF, more specifically as Annotated RDF, enables not only a precise and correct
representation for sensor data but also the possibility of interlinking the data with other existing sources on the Web.

Consider the scenario in which each person is assigned a sensor tag (mode) to use in a building that is equipped with
several sensor base-stations (that will be responsible for recording the presence of tags).  Whenever sensor modes are
detected in the proximity of a base-station, sensor readings are created.  Normally this sensor reading will contain
the time of the reading, the identifier of the base-station and the tag.  
For our example we used datasets publicly available, that represent movements of persons in a conference.  For our test
purposes we used a subset of the dataset available at \url{http://people.openpcd.org/meri/openbeacon/sputnik/data/24c3/}
with a one hour time frame.

For the specific Annotated RDF domain, we can take as starting point the temporal domain, where each triple is annotated
with a temporal validity.  Conceptually, a temporally annotated triple would look like the following:
\begin{small}
\begin{Verbatim}
(tag4302 locatedIn room103):
  [2010-07-28T16:52:00Z,2010-07-28T14:59:00Z]
\end{Verbatim}
\end{small}
stating that the tag represented by the URI \uri{tag4302} was in the room identified by \uri{room103} during the specified time
period.  For the URIs we can define a domain vocabulary or rely on an already existing vocabulary.  

Since a sensor mode can, at any time, be discovered by several base-stations the issue arrises of how to detect which
base-station it is closer to.  This can be viewed as a data cleanup process that can be achieved as a post-processing
step over the stored data.  In our specific experiment, the sensor readings were of the following format:
  \begin{Verbatim}
2010-10-11 14:57:51  10.254.2.15  4302  83
2010-10-11 14:57:51  10.254.3.1   4302  83
2010-10-11 14:57:51  10.254.2.6   4302  83
\end{Verbatim}    
where the columns represent respectively:
\begin{inparaenum}
\item[1)] \texttt{timestamp} when the record was created;
\item[2)] \texttt{ip} address of the base station;
\item[3)] \texttt{tag} identifier and;
\item[4)] \texttt{ssi}. 
\end{inparaenum}
The \texttt{ssi} represents the signal strength of the response from the tag.  Each base station registers each tag at the same
timestamp with different signal strengths, which can be interpreted as the lower the signal strength value is, the
closer the tag is to the base station.  This value can then be used in the data cleanup process to discard the base
station records in which the tag is furthest from. 

In the data cleanup process we start by grouping all the \texttt{ips} (with the lowest \texttt{ssi}) for a given
\texttt{timestamp} and \texttt{tag}.  After this step we can merge all records that share the \texttt{tag} and
\texttt{ip} and have consecutive \texttt{timestamp} into a single interval.

\section{Conclusion}
\label{sec:conclusion}

In this paper we have presented a generalised RDF annotation framework that conservatively extends the RDFS semantics,
along with an extension of the SPARQL query language to query annotated data. 
The framework presented here is generic enough to cover other proposals for RDF annotations and their query languages.
Our approach extends the classical case of RDFS reasoning with features of different annotation domains, such as
temporality, fuzzyness, trust, etc.~and presents a uniform and programatic way to combine any annotation domains.

Furthermore, we presented a semantics for an extension of the SPARQL query language, \aSPARQL, that enables querying RDF
with annotations.  %
Queries exemplified in related literature for specific extensions of SPARQL can be expressed in
\aSPARQL.
Noticeably, our semantics goes beyond the expressivity of the current SPARQL specification and includes some features
from SPARQL~1.1 such as aggregates, variable assignments and sub-queries.
We also described our implementation of AnQL based on constraint logic programming techniques along with a practical
experiment for representing sensor data as Annotated RDF.

\section*{Acknowledgement}
We would like to thank Gergely Luk\'acsy for his participation in the development of this report. The work presented in
this report has been funded in part by Science Foundation Ireland under Grant No.\ SFI/08/CE/I1380
(L\'ion-2) and supported by COST Action IC0801 on Agreement Technologies.

\bibliography{references}

\begin{thebibliography}{10}
\expandafter\ifx\csname url\endcsname\relax
  \def\url#1{\texttt{#1}}\fi
\expandafter\ifx\csname urlprefix\endcsname\relax\def\urlprefix{URL }\fi
\expandafter\ifx\csname href\endcsname\relax
  \def\href#1#2{#2} \def\path#1{#1}\fi

\bibitem{rdf}
F.~Manola, E.~Miller, {RDF Primer}, {W3C Recommendation}, World Wide Web
  consortium, available at \url{http://www.w3.org/TR/rdf-primer/} (Feb.~10
  2004).

\bibitem{gutierrez-etal-2007}
C.~Guti{\'e}rrez, C.~A. Hurtado, A.~A. Vaisman, {Introducing Time into RDF},
  IEEE Transactions on Knowledge and Data Engineering 19~(2) (2007) 207--218.

\bibitem{DBLP:conf/www/PuglieseUS08}
A.~Pugliese, O.~Udrea, V.~S. Subrahmanian, {Scaling RDF with time}, in:
  J.~Huai, R.~Chen, H.-W. Hon, Y.~Liu, W.-Y. Ma, A.~Tomkins, X.~Zhang (Eds.),
  {Proceedings of the 17th International Conference on World Wide Web, WWW
  2008, Beijing, China, April 21-25, 2008}, ACM, 2008, pp. 605--614.

\bibitem{tapp-bern-09}
J.~Tappolet, A.~Bernstein, {Applied Temporal RDF: Efficient Temporal Querying
  of RDF Data with SPARQL}, in: Aroyo et~al.  \cite{ESWC:2009}, pp. 308--322.

\bibitem{Mazzieri08}
M.~Mazzieri, A.~F. Dragoni, {A Fuzzy Semantics for the Resource Description
  Framework}, in: Uncertainty Reasoning for the Semantic Web I, ISWC
  International Workshops, URSW 2005-2007, Revised Selected and Invited Papers,
  no. 5327 in Lecture Notes in Computer Science, Springer, 2008, pp. 244--261.

\bibitem{Straccia09f}
U.~Straccia, {A Minimal Deductive System for General Fuzzy RDF}, in:
  A.~Polleres, T.~Swift (Eds.), {Web Reasoning and Rule Systems, Third
  International Conference, RR 2009, Chantilly, VA, USA, October 25-26, 2009,
  Proceedings}, Vol. 5837 of Lecture Notes in Computer Science, Springer, 2009,
  pp. 166--181.

\bibitem{hartig-09}
O.~Hartig, {Querying Trust in RDF Data with tSPARQL}, in: Aroyo et~al.
  \cite{ESWC:2009}, pp. 5--20.

\bibitem{Schenk08}
S.~Schenk, {On the Semantics of Trust and Caching in the Semantic Web}, in:
  Proc. of 7th International Semantic Web Conference (ISWC'2008), 2008, pp.
  533--549.

\bibitem{Dividino09}
R.~Q. Dividino, S.~Sizov, S.~Staab, B.~Schueler, {Querying for Provenance,
  Trust, Uncertainty and other Meta Knowledge in RDF}, Journal of Web Semantics
  7~(3) (2009) 204--219.

\bibitem{rdfs}
D.~Brickley, R.~Guha, \href{http://www.w3.org/TR/rdf-schema/}{{RDF Vocabulary
  Description Language 1.0: RDF Schema}}, {W3C Recommendation}, World Wide Web
  consortium, available at \url{http://www.w3.org/TR/rdf-schema/} (Feb.~10
  2004).
\newline\urlprefix\url{http://www.w3.org/TR/rdf-schema/}

\bibitem{sparql}
A.~Seaborne, E.~Prud'hommeaux, {SPARQL Query Language for RDF}, {W3C
  Recommendation}, World Wide Web consortium, available at
  \url{http://www.w3.org/TR/rdf-sparql-query/} (Jan.~15 2008).

\bibitem{Straccia10a}
U.~Straccia, N.~Lopes, G.~Lukacsy, A.~Polleres, {A General Framework for
  Representing and Reasoning with Annotated Semantic Web Data}, in: Proceedings
  of the Twenty-Fourth AAAI Conference on Artificial Intelligence (AAAI-10),
  AAAI Press, 2010, pp. xxx--xxx.

\bibitem{Lopes10a}
N.~Lopes, A.~Polleres, U.~Straccia, A.~Zimmermann, {AnQL: SPARQLing Up
  Annotated RDF}, in: Proceedings of the International Semantic Web Conference
  (ISWC-10), no. 6496 in Lecture Notes in Computer Science, Springer-Verlag,
  2010, pp. 518--533.

\bibitem{Udrea06a}
O.~Udrea, D.~R. Recupero, V.~S. Subrahmanian, {Annotated RDF}, in: The Semantic
  Web: Research and Applications, 3rd European Semantic Web Conference, ESWC
  2006, no. 4011 in Lecture Notes in Computer Science, Springer, 2006, pp.
  487--501.

\bibitem{Udrea10}
O.~Udrea, D.~R. Recupero, V.~S. Subrahmanian, {Annotated RDF}, ACM Transactions
  on Computational Logic 11~(2) (2010) 1--41.

\bibitem{Kifer92}
M.~Kifer, V.~Subrahmanian, {Theory of Generalized Annotated Logic Programming
  and its Applications}, Journal of Logic Programming 12 (1992) 335--367.

\bibitem{Green_Karvounarakis_Tannen:07}
T.~J. Green, G.~Karvounarakis, V.~Tannen, {Provenance Semirings}, in: L.~Libkin
  (Ed.), {Proceedings of the Twenty-Sixth ACM SIGACT-SIGMOD-SIGART Symposium on
  Principles of Database Systems, June 11-13, 2007, Beijing, China}, ACM Press,
  2007, pp. 31--40.

\bibitem{DBLP:conf/sigmod/KarvounarakisIT10}
G.~Karvounarakis, Z.~G. Ives, V.~Tannen, Querying data provenance, in: A.~K.
  Elmagarmid, D.~Agrawal (Eds.), SIGMOD Conference, ACM, 2010, pp. 951--962.

\bibitem{MazzieriDragoni:2005aa}
M.~Mazzieri, A.~F. Dragoni, {A Fuzzy Semantics for Semantic Web Languages}, in:
  P.~C.~G. da~Costa, K.~B. Laskey, K.~J. Laskey, M.~Pool (Eds.), ISWC-URSW,
  2005, pp. 12--22.

\bibitem{Mazzieri:2004aa}
M.~Mazzieri, {A Fuzzy {RDF} Semantics to Represent Trust Metadata}, in: 1st
  Workshop on Semantic Web Applications and Perspectives (SWAP2004), Ancona,
  Italy, 2004, pp. 83--89.

\bibitem{DBLP:journals/ws/CarrollBHS05}
J.~J. Carroll, C.~Bizer, P.~J. Hayes, P.~Stickler, Named graphs, Journal of Web
  Semantics 3~(4) (2005) 247--267.

\bibitem{Buneman10}
P.~Buneman, E.~Kostylev, {Annotation Algebras for RDFS}, in: The Second
  International Workshop on the role of Semantic Web in Provenance Management
  (SWPM-10), CEUR Workshop Proceedings, 2010.

\bibitem{hogan-2011}
A.~Hogan, {Exploiting RDFS and OWL for Integrating Heterogeneous, Large-Scale,
  Linked Data Corpora}, Ph.D. thesis, {Digital Enterprise Research Institute,
  National University of Ireland, Galway}, available from
  \url{http://aidanhogan.com/docs/thesis/}; defended. (2011).

\bibitem{Munoz07}
S.~Mu{\~n}oz, J.~P{\'e}rez, C.~Guti{\'e}rrez, {Minimal Deductive Systems for
  RDF}, in: E.~Franconi, M.~Kifer, W.~May (Eds.), {The Semantic Web: Research
  and Applications, 4th European Semantic Web Conference, ESWC 2007, Innsbruck,
  Austria, June 3-7, 2007, Proceedings}, Vol. 4519 of Lecture Notes in Computer
  Science, Springer, 2007, pp. 53--67.

\bibitem{haye-04}
P.~Hayes, {RDF Semantics}, {W3C Recommendation}, World Wide Web consortium,
  available at \url{http://www.w3.org/TR/rdf-mt/} (Feb.~10 2004).

\bibitem{Gutierrez04}
C.~Guti{\'e}rrez, C.~Hurtado, A.~O. Mendelzon, {Foundations of Semantic Web
  Databases}, in: A.~Deutsch (Ed.), {Proceedings of the Twenty-third ACM
  SIGACT-SIGMOD-SIGART Symposium on Principles of Database Systems, June 14-16,
  2004, Paris, France}, ACM, 2004, pp. 95--106.

\bibitem{iann09}
G.~Ianni, T.~Krennwallner, A.~Martello, A.~Polleres, {Dynamic Querying of
  Mass-Storage RDF Data with Rule-Based Entailment Regimes}, in: Bernstein
  et~al.  \cite{ISWC:2009}, pp. 310--327.

\bibitem{HajekP98}
P.~H{\'a}jek, Metamathematics of Fuzzy Logic, Trends in Logic, KluwerAcademic
  Publisher, 1998.

\bibitem{Klement00}
E.~P. Klement, R.~Mesiar, E.~Pap, Triangular Norms, Trends in Logic - Studia
  Logica Library, Kluwer Academic Publishers, 2000.

\bibitem{Abramsky94}
S.~Abramsky, A.~Jung, {Domain Theory}, in: S.~Abramsky, D.~M. Gabbay, T.~S.~E.
  Maibaum (Eds.), {Handbook of Logic in Computer Science - Volume 3: Semantic
  Structures}, Oxford University Press, 1994, pp. 1--168.

\bibitem{Ding_Finin_Peng_PinheirodaSilva_McGuinness:05}
L.~Ding, T.~Finin, Y.~Peng, P.~P. da~Silva, D.~L. McGuinness,
  \href{ftp://ftp.ksl.stanford.edu/pub/KSL_Reports/KSL-05-06.pdf}{{Tracking RDF
  Graph Provenance using RDF Molecules}}, Tech. rep., Knowledge System Lab
  (2005).
\newline\urlprefix\url{ftp://ftp.ksl.stanford.edu/pub/KSL_Reports/KSL-05-06.pdf}

\bibitem{Carroll_Bizer_Hayes_Stickler:05}
J.~Carroll, C.~Bizer, P.~J. Hayes, P.~Stickler, {Named graphs, provenance and
  trust}, in: A.~Ellis, T.~Hagino (Eds.), {Proceedings of the 14th
  International Conference on World Wide Web, WWW 2005, Chiba, Japan, May
  10-14, 2005}, ACML Press, 2005, pp. 613--622.

\bibitem{Flouris_Fundulaki_Pediaditis_Theoharis_Christophides:09}
G.~Flouris, I.~Fundulaki, P.~Pediaditis, Y.~Theoharis, V.~Christophides,
  {Coloring RDF Triples to Capture Provenance}, in: Bernstein et~al.
  \cite{ISWC:2009}, pp. 196--212.

\bibitem{Hartig:09}
O.~Hartig, {Provenance Information in the Web of Data}, in: C.~Bizer, T.~Heath,
  T.~Berners-Lee, K.~Idehen (Eds.), {Linked Data on the Web (LDOW 2009),
  Proceedings of the WWW2009 Workshop on Linked Data on the Web, Madrid, Spain,
  April 20, 2009}, Vol. 538 of CEUR Workshop Proceedings, CEUR, 2009.

\bibitem{Delbru_Polleres_Tummarello_Decker:08}
R.~Delbru, A.~Polleres, G.~Tummarello, S.~Decker, {Context Dependent Reasoning
  for Semantic Documents in Sindice}, in: A.~Fokoue, Y.~Guo, J.~H.~T. Liebig
  (Eds.), {4th International Workshop on Scalable Semantic Web Knowledge Base
  Systems (SSWS2008)}, 2008.

\bibitem{Labrador10}
N.~M. Labrador, U.~Straccia, Monotonic mappings invariant linearisation of
  finite posets, Tech. rep., Computing Research Repository, available as CoRR
  technical report at http://arxiv.org/abs/1006.2679 (2010).

\bibitem{DBLP:journals/tods/PerezAG09}
J.~P{\'e}rez, M.~Arenas, C.~Guti{\'e}rrez, {Semantics and complexity of
  SPARQL}, ACM Transactions on Database Systems 34~(3).

\bibitem{DBLP:conf/semweb/AnglesG08}
R.~Angles, C.~Gutierrez, {The Expressive Power of SPARQL}, in: A.~P. Sheth,
  S.~Staab, M.~Dean, M.~Paolucci, D.~Maynard, T.~W. Finin, K.~Thirunarayan
  (Eds.), International Semantic Web Conference, Vol. 5318, Springer, 2008, pp.
  114--129.

\bibitem{sparql11}
S.~Harris, A.~Seaborne, {SPARQL 1.1 Query Language}, {W3C Working Draft},
  {W3C}, \url{http://www.w3.org/TR/2010/WD-sparql11-query-20100601/} (2010).

\bibitem{DBLP:journals/cacm/Allen83}
J.~F. Allen, Maintaining knowledge about temporal intervals, Communications of
  the ACM 26~(11) (1983) 832--843.

\bibitem{xmlschema11-2}
D.~Peterson, S.~S. Gao, A.~Malhotra, C.~M. Sperberg-McQueen, H.~S. Thompson,
  {W3C XML Schema Definition Language (XSD) 1.1 Part 2: Datatypes}, {W3C
  Working Draft}, World Wide Web consortium, available at
  \url{http://www.w3.org/TR/2009/WD-xmlschema11-2-20091203/} (Dec.~3 2009).

\bibitem{WielemakerHuangvan-der-Meij:2008aa}
J.~Wielemaker, Z.~Huang, L.~{van der Meij}, {SWI-Prolog and the Web}, {Theory
  and Practice of Logic Programming} 8~(3) (2008) 363--392.

\bibitem{ISWC:2009}
A.~Bernstein, D.~R. Karger, T.~Heath, L.~Feigenbaum, D.~Maynard, E.~Motta,
  K.~Thirunarayan (Eds.), {The Semantic Web - ISWC 2009, 8th International
  Semantic Web Conference, ISWC 2009, Chantilly, VA, USA, October 25-29, 2009.
  Proceedings}, Vol. 5823 of Lecture Notes in Computer Science, Springer, 2009.

\bibitem{ESWC:2009}
L.~Aroyo, P.~Traverso, F.~Ciravegna, P.~Cimiano, T.~Heath, E.~Hyv{\"o}nen,
  R.~Mizoguchi, E.~Oren, M.~Sabou, E.~P.~B. Simperl (Eds.), {The Semantic Web:
  Research and Applications, 6th European Semantic Web Conference, ESWC 2009,
  Heraklion, Crete, Greece, May 31-June 4, 2009, Proceedings}, Vol. 5554 of
  Lecture Notes in Computer Science, Springer, 2009.

\end{thebibliography}
\bibliographystyle{elsarticle-num}

\appendix
\section{Proofs of theorems and propositions}
\label{sec:proofs}

\begin{small}

\subsection{Proof of Theorem~\ref{thm:strongquasi}}

We start by proving that for all $z,z'\in L_1, \overline{A}(z\oplus_1 z')\succeq_2\overline{A}(z)\otimes_2\overline{A}(z')$.

\begin{proof}
Let $z,z'\in L_1$. In order to prove the proposition, we introduce the notation $K^A_z=_\mathrm{def}\{J\subseteq A\mid z\preceq_1\bigoplusone{\tuple{x,y}\in J}x\}$. The property that we want to prove can be rewritten:
\[
 \begin{array}{l}
 \mathbf{lub}\{\bigotimestwo{\tuple{x,y}\in J}y \mid J\in K^A_{z\oplus_1 z'}\} \succeq_2\\
  \hspace{1.8cm}\mathbf{lub} \{\bigotimestwo{\tuple{x,y}\in J}y \mid J\in K^A_z\} \otimes_2\mathbf{lub} \{\bigotimestwo{\tuple{x,y}\in J'}y \mid J'\in K^A_{z'}\} \ .
 \end{array}
\]
\nd Let us introduce two intermediary lemmas:

\begin{lemma}
 For all $z,z'\in L_1, K^A_{z\oplus_1 z'}=K^A_z\cap K^A_{z'}$.
\end{lemma}
\begin{proof}
\nd We simply prove each inclusion separately:
\begin{description}
  \item[$\subseteq$:] let $J\in K^A_{z\oplus_1 z'}$, which implies that $z\oplus_1 z'\preceq_1\bigoplusone{\tuple{x,y}\in J}x$. So, $z\preceq_1\bigoplusone{\tuple{x,y}\in J}x$ and $z'\bigoplusone{\tuple{x,y}\in J}x$, that is $J\in K^A_z$ and $J\in K^A_{z'}$.
  \item[$\supseteq$:] let $J\in K^A_z\cap K^A_{z'}$, which implies that $z\preceq_1\bigoplusone{\tuple{x,y}\in J}x$ and $z'\preceq_1\bigoplusone{\tuple{x,y}\in J}x$, so $z\oplus_1 z'\preceq_1\bigoplusone{\tuple{x,y}\in J}x$ by definition of $\oplus_1$. Consequently, $J\in K^A_{z\oplus_1 z'}$.
 \end{description}
 \end{proof}
 
\begin{lemma}
 For all $z,z'\in L_1, K^A_{z\oplus_1 z'}=\{J\cup J'\mid (J,J')\in K^A_z\times K^A_{z'}\}$.
\end{lemma}

\begin{proof}
\nd Again, we prove each inclusion separately:
\begin{description}
  \item[$\subseteq$:] trivial since $J\in K^A_{z\oplus_1 z'}$ implies that $J\in K^A_z\cap K^A_{z'}$ and $J=J\cup J$.
  \item[$\supseteq$:] let $J\in K^A_z$ and $J'\in K^A_{z'}$. Clearly, $\bigoplusone{\tuple{x,y}\in J}x\preceq_1\bigoplusone{\tuple{x,y}\in J\cup J'}x$. Consequently, $J\cup J'\in K^A_z$. Symmetrically, we prove that $J\cup J'\in K^A_{z'}$.
 \end{description}
\end{proof}

\nd  This allows us to rewrite the problem into:
\[
 \begin{array}{l}
 \mathbf{lub}\{\bigotimestwo{\tuple{x,y}\in J\cup J'}y \mid (J,J')\in K^A_z\times K^A_{z'}\}\succeq_2\\
 \hspace{1.2cm}\mathbf{lub}\{\bigotimestwo{\tuple{x,y}\in J}y \mid J\in K^A_z\}\otimes_2\mathbf{lub}\{\bigotimestwo{\tuple{x,y}\in J'}y \mid J'\in K^A_{z'}\}
 \end{array}
\]
which is more concisely written:
\[
 \bigoplustwo{(J,J')\in K^A_z\times K^A_{z'}}\left(\bigotimestwo{\tuple{x,y}\in J\cup J'}y\right)\succeq_2\bigoplustwo{J\in K^A_z}\left(\bigotimestwo{\tuple{x,y}\in J}y\right)\otimes_2\bigoplustwo{J'\in K^A_{z'}}\left(\bigotimestwo{\tuple{x,y}\in J'}y\right) \ .
\]
This is easily established by distributivity of $\otimes_2$ over $\oplus_2$ in the right hand side and by remarking that~\footnote{Note that if 
$\otimes_2$ is idempotent, then the inequality becomes an equality.}
\[
\bigoplustwo{J'\in K^A_{z'}}\left(\bigoplustwo{J\in K^A_z}(\bigotimestwo{\langle x,y\rangle\in J}y\otimes_2\bigotimestwo{\langle x,y\rangle\in J'}y)\right)\preceq_2\bigoplustwo{(J,J')\in K^A_z\times K^A_{z'}}\left(\bigotimestwo{\langle x,y\rangle\in J\cup J'}y\right) \ .
\]
\end{proof}

\nd  The second part of the proof demonstrates that for all $z,z'\in L_1, \overline{A}(z\otimes_1 z') \succeq_2 \overline{A}(z)\oplus_2\overline{A}(z')$

\begin{proof}
Before giving the main arguments for the proof, we rewrite the goal as follows:~\footnote{We here reuse the notation $K^A_z$ introduced above.}
 \[
 \begin{array}{l}
 \mathbf{lub}\{\bigotimestwo{\tuple{x,y}\in J}y\mid J\in K^A_z\}\oplus_2 \mathbf{lub}\{\bigotimestwo{\tuple{x,y}\in J'}y\mid J'\in K^A_{z'}\} \\\hspace*{1cm}\preceq_2  \mathbf{lub}\{\bigotimestwo{\tuple{x,y}\in J}y\mid J\in K^A_{z\otimes_1 z'}\} \ , 
 \end{array}
 \]
\noindent which again can be made more concise with the following notation:
 \[
   \bigoplustwo{J\in K^A_z}\left(\bigotimestwo{\tuple{x,y}\in J}y\right)\oplus_2\bigoplustwo{J\in K^A_{z'}}\left(\bigotimestwo{\tuple{x,y}\in J'}y\right) \preceq_2 \bigoplustwo{J\in K^A_{z\otimes_1 z'}}\left(\bigotimestwo{\tuple{x,y}\in J}y\right) \ .
 \]
\noindent Associativity of $\oplus_2$ simplifies the equation further:
 \[
   \bigoplustwo{J\in K^A_z\cup K^A_{z'}}\left(\bigotimestwo{\tuple{x,y}\in J}y\right) \preceq_2 \bigoplustwo{J\in K^A_{z\otimes_1 z'}}(\bigotimestwo{\tuple{x,y}\in J}y) \ .
 \]
 We established the result by first proving the following lemma:

 \begin{lemma}\label{lem:union}
  For all $z,z'\in L_1, K^A_z\cup K^A_{z'}\subseteq K^A_{z\otimes_1 z'}$.
 \end{lemma}
\begin{proof}
 \nd Let $J\in K^A_z$. It holds that $z\preceq_1\bigoplusone{\tuple{x,y}\in J}x$ and $z\otimes_1 z'\preceq_1 z$. So $J\in K^A_{z\otimes_1 z'}$. Idem for any $J\in K^A_{z'}$.  
 \end{proof}

\nd This allows us now to easily see that 
\[
\bigoplustwo{J\in K^A_z\cup K^A_{z'}}\left(\bigotimestwo{\tuple{x,y}\in J}y\right) \preceq_2 \bigoplustwo{J\in K^A_{z\otimes_1 z'}}(\bigotimestwo{\tuple{x,y}\in J}y) \ .
\]
\nd  Notice that the opposite inequality does not hold in general.
\end{proof}

\subsection{Proof of Theorem~\ref{thm:normalisation}}

In order to prove the theorem, we first demonstrate the following proposition:

\begin{proposition}\label{prop:normalP}
 If $D_1=\tuple{L_1,\oplus_1,\otimes_1,\bot_1,\top_1}$ is a lattice then, for all $A\subseteq L_1\times L_2$ finite, $\overline{A}=\overline{\mathsf{Normalise}(A)}$.
\end{proposition}

\nd Let us assume that $D_1$ is a lattice.
We show the proposition by proving that at each step of the saturation and reduction, the set $R$ is such that $\overline{R}=\overline{A}$. This is trivially true at the initialisation of $\mathsf{Saturate}$. Now let us assume that $R$ satisfies this property at a certain step of the execution. We start by ensuring that 
\[
\overline{R}=\overline{R\cup\{\tuple{\bigoplusone{J\in X}\bigotimesone{\tuple{x,y}\in J}x,\bigotimestwo{J\in X}\bigoplustwo{\tuple{x,y}\in J}y}}\}
\]
\nd and 
\[
\overline{R}=\overline{R\cup\{\tuple{\bigotimesone{J\in X}\bigoplusone{\tuple{x,y}\in J}x,\bigoplustwo{J\in X}\bigotimestwo{\tuple{x,y}\in J}y}}\} \ .
\]

\nd We can decompose further the proof by simply showing that, given $\tuple{a,b},\tuple{c,d}\in R$, 
\[
\overline{R}=\overline{R\cup\{\tuple{a\otimes_1 c,b\oplus_2 d}\}}
\]
\nd and 
\[
\overline{R}=\overline{R\cup\{\tuple{a\oplus_1 c,b\otimes_2 d}\}} \ .
\]

\nd To structure the proof better, we split the proof into several lemmas corresponding to each of the aforementioned steps.

\begin{lemma}\label{lem:timesplus}
Let $z\in L_1$. The equality $\overline{R\cup\{\tuple{a\otimes_1 c,b\oplus_2 d}\}}(z)=\overline{R}(z)$ holds.
\end{lemma}

\begin{proof}
If the pair $\tuple{a\otimes_1 c,b\oplus_2 d}$ already belongs to $R$, the equality is trivial. Let us assume $\tuple{a\otimes_1 c,b\oplus_2 d}\notin R$ so that we can easily distinguish between sets that include $\tuple{a\otimes_1 c,b\oplus_2 d}$ and sets that do not. Using the definition of $\overline{R\cup\{\tuple{a\otimes_1 c,b\oplus_2 d}\}}(z)$, we can write:
\[
 \begin{array}{l}
 \overline{R\cup\{\tuple{a\otimes_1 c,b\oplus_2 d}\}}(z) =\overline{R}(z)\ \oplus_2 \\
 \hspace{0.2cm}\mathbf{lub}\{(\bigotimestwo{\tuple{x,y}\in J}y)\otimes_2(b\oplus_2 d) \mid J\subseteq R \text{ and } \\
 \hspace*{4.5cm}z\preceq_1(\bigoplusone{\tuple{x,y}\in J}x)\oplus_1(a\otimes_1 c)\} \ .
 \end{array}
\]
\nd Further, due to distributivity, 
\[
(\bigotimestwo{\tuple{x,y}\in J}y)\otimes_2(b\oplus_2 d)=((\bigotimestwo{\tuple{x,y}\in J}y)\otimes_2 b)\oplus_2((\bigotimestwo{\tuple{x,y}\in J}y)\otimes_2 d) \ .
\]
\nd We can therefore rewrite the previous equality to:
\[
\begin{array}{l}
 \overline{R\cup\{\tuple{a\otimes_1 c,b\oplus_2 d}\}}(z) = \overline{R}(z)\\
             \hspace{0.2cm}\oplus_2\ \mathbf{lub}\{(\bigotimestwo{\tuple{x,y}\in J}y)\otimes_2 b \mid J\subseteq R \text{ and }z\preceq_1(\bigoplusone{\tuple{x,y}\in J}x)\oplus_1(a\otimes_1 c)\}\\
             \hspace{0.2cm}\oplus_2\ \mathbf{lub}\{(\bigotimestwo{\tuple{x,y}\in J}y)\otimes_2 d \mid J\subseteq R \text{ and }z\preceq_1(\bigoplusone{\tuple{x,y}\in J}x)\oplus_1(a\otimes_1 c)\} \ .
\end{array}\]
\nd Additionally, since $D_1$ is a lattice, we have 
\[
(\bigoplusone{\tuple{x,y}\in J}x)\oplus_1(a\otimes_1 c) = (\bigoplusone{\tuple{x,y}\in J}x)\oplus_1 a)\otimes_1(\bigoplusone{\tuple{x,y}\in J}x)\oplus_1 c) \ .
\]
\nd So $z\preceq_1(\bigoplusone{\tuple{x,y}\in J}x)\oplus_1(a\otimes_1 c)$ implies $z\preceq_1(\bigoplusone{\tuple{x,y}\in J}x)\oplus_1 a$ and $z\preceq_1(\bigoplusone{\tuple{x,y}\in J}x)\oplus_1 c$. This means that $J\cup\{\tuple{a,b}\}\in\{K\subseteq R\mid z\preceq_1\bigotimesone{\tuple{x,y}\in J}x\}$ so necessarily, $(\bigotimestwo{\tuple{x,y}}y)\otimes_2 b\preceq_2\overline{R}(z)$. Analogically, we conclude that $(\bigotimestwo{\tuple{x,y}\in J}y)\otimes_2 d\preceq_2\overline{R}(z)$ and generalising to any suitable $J$, we conclude, using the equation above, that $\overline{R\cup\{\tuple{a\otimes_1 c,b\oplus_2 d}\}}(z) = \overline{R}(z)$.
\end{proof}

\nd Now let us prove this second equality:

\begin{lemma}\label{lem:plustimes}
Let $z\in L_1$. The equality $\overline{R\cup\{\tuple{a\oplus_1 c,b\otimes_2 d}\}}(z)=\overline{R}(z)$ holds.
\end{lemma}

\begin{proof}
We apply a similar method as for Lemma~\ref{lem:timesplus} to get to the following equality:
\[
\begin{array}{l}
 \overline{R\cup\{\tuple{a\oplus_1 c,b\otimes_2 d}\}}(z) = \overline{R}(z)\ \oplus_2\\
 \hspace{0.2cm}\mathbf{lub}\{(\bigotimestwo{\tuple{x,y}\in J}y)\otimes_2(b\otimes_2 d) \mid J\subseteq R \text{ and } \\
 \hspace*{4.5cm} z\preceq_1(\bigoplusone{\tuple{x,y}\in J}x)\oplus_1(a\oplus_2 c)\} \ .
 \end{array}
\]
\nd This means that $J\cup\{\tuple{a,b},\tuple{c,d}\}\in\{K\subseteq R\mid z\preceq_1\bigoplusone{\tuple{x,y}\in K}x\}$, which implies that $(\bigotimestwo{\tuple{x,y}\in J}y)\otimes_2(b\otimes_2 d)\preceq_2\overline{R}(z)$. Generalising this to any suitable $J$, we obtain the equality.
\end{proof}

\nd Now, let us prove that the reduce algorithm preserves the quasi homomorphism.

\begin{lemma}\label{lem:reduce}
$\overline{A}=\overline{\mathsf{Reduce}(A)}$.
\end{lemma}

\begin{proof}
Let $\tuple{a,b}\in R$ such that there exists $\tuple{a',b'}\in R$ such that $a\preceq_1 a'$ and $b\preceq_2 b'$. Using the same approach as in Lemma~\ref{lem:timesplus}, we obtain the following equality:
\[
 \begin{array}{l}
 \overline{R}(z) = \overline{R\setminus\{\tuple{a,b}\}}(z)\ \oplus_2 \\
  \hspace{0.5cm}\mathbf{lub}\{(\bigotimestwo{\tuple{x,y}\in J}y)\otimes_2 b \mid J\subseteq R\setminus\{\tuple{a,b}\} \text{ and } \\
  \hspace*{5cm}z\preceq_1(\bigoplusone{\tuple{x,y}\in J}x)\oplus_1 a\} \ .
  \end{array}
\]
\nd From the hypothesis, we have that $z\preceq_1(\bigoplusone{\tuple{x,y}\in J}x)\oplus_1 a'$ for any appropriate $J$. Moreover, for the same $J$, we have $(\bigotimestwo{\tuple{x,y}\in J}y)\otimes_2 b\preceq_2 b'$. Generalising to all adequate $J$, we entail that:
\[
\begin{array}{l}
 \overline{R\setminus\{\tuple{a,b}\}}(z)\succeq_2\mathbf{lub}\{(\bigotimestwo{\tuple{x,y}\in J}y)\otimes_2 b \mid J\subseteq R\setminus\{\tuple{a,b}\} \text{ and } \\ 
 \hspace*{5cm} z\preceq_1(\bigoplusone{\tuple{x,y}\in J}x)\oplus_1 a\}
 \end{array}
\]
\nd and, thus, $\overline{R}(z) = \overline{R\setminus\{\tuple{a,b}\}}(z)$.
Similarly, every pair $\tuple{\bot_1,y}$ or $\tuple{x,\bot_2}$ does not affect the function $\overline{R}$.
\end{proof}

\nd Now, the proof of Proposition~\ref{prop:normalP} follows from an inductive application of Lemmas~\ref{lem:timesplus},~\ref{lem:plustimes} and~\ref{lem:reduce}. Therefore, $\overline{A}=\overline{\mathsf{Normalise}(A)}$ holds.

\paragraph{Proof of the theorem}
The implication $\Leftarrow$ is a direct consequence of Proposition~\ref{prop:normalP}.

Let us prove the other direction. Let $A$ and $B$ two finite sets of pairs of primitive annotations in $L_1\times L_2$ such that $\overline{A}=\overline{B}$.
For $A\subseteq L_1\times L_2$ and $x\in L_1$, let $K_x^A = \{J\subseteq A\mid x\preceq_1\bigoplusone{\tuple{\alpha,\beta}\in J}\alpha\}$. We also remind that:
\begin{eqnarray*}
 \overline{A}: & L_1 & \to L_2 \\
               & x   & \mapsto \bigoplustwo{J\in K_x^A}\bigotimestwo{\tuple{\alpha,\beta}\in J}\beta
\end{eqnarray*}
Moreover, we introduce the following new notation:
\begin{eqnarray*}
 \widetilde{A}: & L_1 & \to L_1 \\
                & x   & \mapsto \bigotimesone{J\in K_x^A}\bigoplusone{\tuple{\alpha,\beta}\in J}\alpha
\end{eqnarray*}

\nd We establish the proof through the support of several intermediary lemmas.

\begin{lemma}\label{lem:yoverline}
If $\tuple{x,y}\in A$ then $y \preceq_2\overline{A}(x)$.
\end{lemma}

\begin{proof}
Let $\tuple{x,y}\in A$. We remark that $x\preceq_1\bigoplusone{\tuple{\alpha,\beta}\in\{\tuple{x,y}\}}\alpha$ and $\{\tuple{x,y}\}\subseteq A$, so:
\[
 y = \bigotimestwo{\tuple{\alpha,\beta}\in\{\tuple{x,y}\}}\beta \preceq_2 \bigoplustwo{J\in K_x^A}\bigotimestwo{\tuple{\alpha,\beta}\in J}\beta = \overline{A}(x) \ .
\]
\end{proof}

\begin{lemma}\label{lem:xtilde}
If $\tuple{x,y}\in A$ then $x\preceq_1\widetilde{A}(x)$.
\end{lemma}

\begin{proof}
Let $\tuple{x,y}\in A$. For all $J\in K_x^A$, $x\preceq_1\bigoplusone{\tuple{\alpha,\beta}\in J}\alpha$ so, since $D_1$ is a lattice, $x\preceq_1\bigotimesone{J\in K_x^A}\bigoplusone{\tuple{\alpha,\beta}\in J}\alpha$, that is, $x\preceq_1 \widetilde{A}(x)$.
\end{proof}

\begin{lemma}\label{prop:antitone}
If $D_1=\tuple{L_1,\oplus_1,\otimes_1,\bot_1,\top_1}$ is a lattice, then a quasihomomorphism is an antitone function, with respect to the orders induced by $\oplus_1$ and $\oplus_2$.
\end{lemma}

\begin{proof}
Assume that $D_1$ is a lattice. Let $f$ be a quasihomomorphism. Let $x,x'\in L_1$ be two annotation values such that $x\preceq_1 x'$. Then
$f(x)  =   f(x \otimes x') \succeq_2 f(x)\oplus_2 f(x')$ and, thus, $f(x)  \succeq_2 f(x')$.
\end{proof}

\nd Using the previous lemmas, we can prove two additional lemmas that will bring us to the final proof:

\begin{lemma}\label{lem:match}
If $\tuple{x,y}\in\mathsf{Normalise}(A)$ then $x=\widetilde{A}(x)$ and $y=\overline{A}(x)$.
\end{lemma}

\begin{proof}
Let $\tuple{x,y}\in\mathsf{Normalise}(A)$. Consequently, $\tuple{x,y}\in\mathsf{Saturate}(A)$.
Moreover, by definition of $\mathsf{Saturate}$, the pair $\tuple{\widetilde{A}(x),\overline{A}(x)}$ must exist in $\mathsf{Saturate}(A)$. Additionally, from Lemma~\ref{lem:yoverline} and Lemma~\ref{lem:xtilde}, we have that $x\preceq_1\widetilde{A}(x)$ and $y=\overline{A}(x)$, which implies that $\tuple{x,y}$ should be eliminated by the $\mathsf{Reduce}$ algorithm during normalisation, unless $x=\widetilde{A}(x)$ and $y=\overline{A}(x)$.
\end{proof}

\begin{lemma}\label{lem:exist}
For all $x\in L_1$, there exists $u\in L_1$ such that $\tuple{u,\overline{A}(x)}\in\mathsf{Normalise}(A)$ and $x\preceq_1 u$.
\end{lemma}

\begin{proof}
 Let $x\in L_1$. Again, $\tuple{\widetilde{A}(x),\overline{A}(x)}\in\mathsf{Saturate}(A)$. Then, due to the reduction algorithm, there must exist $\tuple{u,v}\in\mathsf{Normalise}(A)$ such that $u\succeq_1\widetilde{A}(x)$ and $v\succeq_2\overline{A}(x)$.We can consider the following assertions:
\begin{center}
{\footnotesize
\begin{tabular}{lrr}
  $D_1$ is a lattice             & (by hypothesis) & (H1) \\
  $\widetilde{A}(x) \preceq_1 u$ & (due to $\mathsf{Reduce}$) & (R1) \\
  $\overline{A}(x)  \preceq_2 v$ & (due to $\mathsf{Reduce}$) & (R2) \\
  $\overline{A}$ is antitone     & (from (H1) and Lemma~\ref{prop:antitone}) & (A1) \\
  $x \preceq_1 \widetilde{A}(x)$    & (from Lemma~\ref{lem:xtilde}) & (A2) \\
  $\overline{A}(x)  \succeq_2 \overline{A}(\widetilde{A}(x))$ & (from (A1) and (A2)) & (A3) \\
  $\overline{A}(\widetilde{A}(x)) \succeq_2 \overline{A}(u)$ & (from (R1) and (A1)) & (A4) \\
  $\overline{A}(u)  \succeq_2 v$ & (from Lemma~\ref{lem:yoverline}) & (A5) \\
  $\overline{A}(x) = v$          & (from (A3), (A4), (A5) and (R2)) & (C1) \\
  $x \preceq_1 u$                & (from (A2) and (R1)) & (C2) \\
\end{tabular}
}
\end{center}
Assertions (C1) and (C2) establish the lemma.
\end{proof}

\begin{proof} (Of Theorem~\ref{thm:normalisation})
Let $\tuple{x,y}\in\mathsf{Normalise}(A)$. From Lemma~\ref{lem:match}, we know that $x=\widetilde{A}(x)$ and $y=\overline{A}(x)$. Moreover, from the hypothesis of the theorem, $\overline{A}(x)=\overline{B}(x)$. Hence, due to Lemma~\ref{lem:exist}, there exists $u\in L_1$ such that 
$\tuple{u,y}\in\mathsf{Normalise}(B)$ and $x\preceq_1 u$. By using the same reasoning, we can infer that there exists $v\in L_1$ such that 
$\tuple{v,y}\in\mathsf{Normalise}(A)$ and $u\preceq_1 v$. But due to Lemma~\ref{lem:match}, we have that $y = v$ and therefore $\tuple{x,y}\in\mathsf{Normalise}(B)$.

The situation is symmetrical with respect to $A$ and $B$, so finally $\mathsf{Normalise}(A)=\mathsf{Normalise}(B)$.
\end{proof}

\end{small}

\end{document}